\documentclass{article}

\usepackage[top=1in, bottom=1in, left=1in, right=1in]{geometry}

\usepackage{tikz}
\usepackage{pgfplots}
\pgfplotsset{compat=1.17}

\PassOptionsToPackage{numbers, compress}{natbib}

\usepackage{amsthm}

\usepackage[utf8]{inputenc}
\usepackage[T1]{fontenc}

\usepackage{booktabs}
\usepackage{longtable,tabularx}
\usepackage{microtype}
\usepackage{multicol}
\usepackage{multirow}
\usepackage{caption}
\usepackage{tcolorbox}
\usepackage{savesym}
\usepackage{algorithm}
\usepackage{algorithmic}
\usepackage{subfigure}

\usepackage{thm-restate}

\savesymbol{Bbbk}
\usepackage{amsmath,amssymb,amsfonts}
\usepackage{url}
\usepackage{tikz}
\usepackage{mathrsfs}
\usepackage{nicefrac,bbm}
\usepackage{hyperref}
\hypersetup{
  colorlinks=true,
  linkcolor=green!30!black,
  linktoc=all,
  citecolor=blue,
  bookmarksnumbered=true,
  pdfproducer={},
  pdfinfo={CreationDate={2022},ModDate={2022}},
  pdfauthor={},
  pdfkeywords={},
  pdfsubject={},
  pdfcreator={},
  pdflang={}
}
\usepackage{cleveref}
\usepackage{mathtools,pifont,enumitem}
\usepackage[framemethod=TikZ]{mdframed}
\usepackage{color,fancybox,graphicx,subfigure}

\savesymbol{st}
\usepackage{soul}
\restoresymbol{SOUL}{st}
\usepackage[normalem]{ulem}

\usepackage{xurl}
\usepackage{scalerel}
\usetikzlibrary{math}
\usepackage{bbding}
\usepackage{array}
\usepackage{dsfont}

\usepackage{empheq}
\usepackage{stackengine}

\newtheorem{theorem}{Theorem}[section]
\newtheorem{lemma}[theorem]{Lemma}
\newtheorem{definition}[theorem]{Definition}

\newtheorem{corollary}[theorem]{Corollary}

\newtheorem{claim}[theorem]{Proposition}

\newtheorem{fact}[theorem]{Fact}

\newtheorem{remark}[theorem]{Remark}
\newtheorem{example}[theorem]{Example}

\crefname{section}{Section}{Sections}
\crefname{theorem}{Theorem}{Theorems}
\crefname{assumption}{Assumption}{Assumptions}
\crefname{lemma}{Lemma}{Lemmas}
\crefname{definition}{Definition}{Definitions}
\crefname{conjecture}{Conjecture}{Conjectures}
\crefname{corollary}{Corollary}{Corollaries}
\crefname{construction}{Construction}{Constructions}
\crefname{claim}{Proposition}{Propositions}
\crefname{observation}{Observation}{Observations}
\crefname{proposition}{Proposition}{Propositions}
\crefname{fact}{Fact}{Facts}
\crefname{question}{Question}{Questions}
\crefname{problem}{Problem}{Problems}
\crefname{remark}{Remark}{Remarks}
\crefname{example}{Example}{Examples}
\crefname{equation}{Equation}{Equations}
\crefname{appendix}{Appendix}{Appendices}
\crefname{algorithm}{Algorithm}{Algorithms}
\crefname{model}{Model}{Models}
\crefname{figure}{Figure}{Figures}

\AfterEndEnvironment{definition}{\noindent\ignorespaces}
\AfterEndEnvironment{assumption}{\noindent\ignorespaces}
\AfterEndEnvironment{lemma}{\noindent\ignorespaces}
\AfterEndEnvironment{theorem}{\noindent\ignorespaces}
\AfterEndEnvironment{proposition}{\noindent\ignorespaces}
\AfterEndEnvironment{fact}{\noindent\ignorespaces}
\AfterEndEnvironment{question}{\noindent\ignorespaces}

\newcommand{\yesnum}{\addtocounter{equation}{1}\tag{\theequation}}

\newcommand{\Ast}{\ensuremath{{\cA}_{\rm st}}}
\newcommand{\Agroup}{\ensuremath{{\cA}_{\rm group}}}
\newcommand{\Ainst}{{\ouralgo{}}}
\newcommand{\ouralgo}{\ensuremath{\cA_{\rm inst\text{-}wise}}}

\newcommand{\cA}{\mathcal{A}}

\newcommand{\cD}{\mathcal{D}}
\newcommand{\cE}{\mathcal{E}}

\newcommand{\cI}{\mathcal{I}}

\newcommand{\cL}{\mathcal{L}}

\newcommand{\cR}{\mathcal{R}}

\newcommand{\ells}{{\ell^\star}}
\newcommand{\mb}[1]{\ensuremath{\boldsymbol{#1}}}
\newcommand{\husigma}{{\mb{\hu}, \mb{\sigma}}}

\newcommand{\sfrac}[2]{#1/#2}

\renewcommand{\epsilon}{\varepsilon}

\newcommand{\Ex}{\operatornamewithlimits{\mathbb{E}}}
\newcommand{\E}{\operatornamewithlimits{\mathbb{E}}}

\def\sabs#1{| #1 |}

\newcommand{\inparen}[1]{\left(#1\right)}
\newcommand{\inbrace}[1]{\left\{#1\right\}}

\newcommand{\hf}{\widehat{f}}

\newcommand{\hu}{\widehat{u}}

\newcommand{\eat}[1]{}

\newcommand{\cN}{\mathcal{N}}

\newcommand{\Area}{{\mathsf {Area}}}

\newcommand{\sP}{\mathscr{P}}
\newcommand{\sR}{\mathscr{R}}
\newcommand{\sU}{\mathscr{U}}

\newcommand{\sidebysidecaption}[4]{%
  \begin{minipage}[t]{#1}
    \vspace*{0pt}
    #3
  \end{minipage}
  \hspace{-10mm}
  \begin{minipage}[t]{#2}
    \vspace*{0pt}
    #4
\end{minipage}%
}

\title{Centralized Selection with Preferences in the Presence of Biases}

\author{L. Elisa Celis \\ Yale University \and Amit Kumar \\ {IIT Delhi}  \and Nisheeth K. Vishnoi  \\Yale University \and Andrew Xu \\ Yale University}

\date{}

\begin{document}

\maketitle

\begin{abstract}
This paper considers the scenario in which there are multiple institutions, each with a limited capacity for candidates, and candidates, each with preferences over the institutions. 
A central entity evaluates the utility of each candidate to the institutions, and the goal is to select candidates for each institution in a way that maximizes utility while also considering the candidates' preferences.
The paper focuses on the setting in which candidates are divided into multiple groups and the observed utilities of candidates in some groups are biased--systematically lower than their true utilities.
The first result is that, in these biased settings, prior algorithms can lead to selections with sub-optimal true utility and significant discrepancies in the fraction of candidates from each group that get their preferred choices.
Subsequently, an algorithm is presented along with proof that it produces selections that achieve near-optimal group fairness with respect to preferences while also nearly maximizing the true utility under distributional assumptions.
Further, extensive empirical validation of these results in real-world and synthetic settings, in which the distributional assumptions may not hold, are presented.

\end{abstract}

\newpage

\tableofcontents

\newpage

\section{Introduction}\label{sec:intro-contributions}
  This paper studies a centralized selection problem that occurs in high-stakes contexts such as education and employment.
In this setting there are $p$ institutions, each with capacities $k_1, k_2, \dots, k_p$, and $n$ candidates, and $\sum_i k_i$  is much less than $n$.
A central entity evaluates the utility of each candidate to the institutions (same for each institution) and the candidates declare preferences for the institutions. 
The high-level goal of this entity is to then select a subset of at most $k_\ell$ candidates for  $\ell$-th institution while trying to maximize total utility and ensuring the preferences of the candidates are taken into account.
    Examples include the Joint Seat Allocation Authority (JOSAA) in India, which places over 50,000 students in top engineering colleges each year \cite{baswana2019centralized,seatMatrixJOSSA2022}, and comparable centralized allocation mechanisms across China \cite{wiki:gaokao}.
        Other examples include  public schools in NYC and Boston, which have used centralized admissions systems  \cite{roth1992Handbook,roth1999redesign},
       and online labor markets (such as TaskRabbit) that use centralized systems to match workers to tasks based on their skills and preferences for working hours, location, and type-of-task \cite{hannak2017freelance,Taskrabbit-matching14}. 

    Several works in Game Theory \cite{roth1989incompleteinfo,aziz2016stable}, Machine Learning \cite{das2005two,muthukrishnan2009ad,liu2020matchinbandit,liu2021matching,yifei2022matching}, and Economics \cite{roth1985common,roth1999redesign,liu2014stable, Rastegari2013MatchingPartialInformation} examine this centralized selection problem and its various forms.
    The variant of interest here is where candidates' utilities $u_1,u_2,\dots,u_n\geq 0$ are cardinal and their preferences are ordinal.
    The utilities can denote candidates' test scores, performance in interviews, or suitability for hiring.
   Candidate $i$’s preferences are given by a ranking $\sigma_i$ 
  of the $p$ institutions, ordered from most to least preferred.
    The Gale-Shapley algorithm~\cite{galeshapley1962} for this problem orders candidates in decreasing order of utility and assigns them to their most preferred institution which still has empty slots. 
    This algorithm can be shown to maximize utility and ensure a type of stability--there are no two candidates who both prefer each other's assigned institution to the ones they have been assigned.

     However, the true or latent utilities of candidates are often estimated through evaluation processes like tests, interviews, or user ratings,  which are known to yield biased estimates against underprivileged groups \cite{rooth2010automatic, kite2016psychology, regner2019committees,elsesser2019lawsuitSATACT,corinne2012science,wenneras2001nepotism,lyness2006fit,CelisKMV23}.
     Standardized exams like the SAT, GRE, and ACT, often pivotal in college admissions, disproportionately affect candidates from minority and low-income demographics \cite{elsesser2019lawsuitSATACT}.
    Further examples of biased evaluations include biased peer-review of fellowship applications against women \cite{wenneras2001nepotism}, and stricter promotion standards for women in managerial positions \cite{lyness2006fit}.

    Selections based on such biases in estimated utilities may adversely affect the opportunities of candidates from disadvantaged groups in multiple ways.
   For one, if an algorithm makes selections in an attempt to maximize estimated utilities, the representation of candidates from disadvantaged groups could be significantly low.
   Moreover, candidates from disadvantaged groups might be assigned to lower preferences compared to other groups.
 This latter inequality of opportunity may, in turn, influence candidates'   future educational, employment, and economic opportunities \cite{lecher2021nycschoolbias,laverde2022distance,grover2017brookingsSchoolBias,hannak2017freelance,upturn_help_wanted}, and is the focus of this paper.
    For instance, a study by \cite{laverde2022distance} showed that within the Boston Public Schools' assignment system, white pre-kindergarten students are more frequently placed in higher-performing schools compared to minority students, contradicting the system's intended purpose of ensuring fair access to quality education; see also  \cite{lecher2021nycschoolbias,grover2017brookingsSchoolBias}.
    And, besides disadvantaging candidates from certain groups, biases in utilities also reduce the utility for institutions, who may be assigned candidates with a low true utility but a high {estimated} utility.
    To reduce the adverse effects of such biases, constraints that increase the representation of underprivileged candidates, anti-bias training that aims to reduce bias in evaluations, and anonymized evaluations that aim to blind the evaluator to the socially salient attributes of candidates have been proposed \cite{BlindAuditions2000, BarbaraAscription2000, Sowell_2008,gawande2010checklist, Bohnet_2016,zestcott2016examining,agarwal2020sway}.
  A  line of work in the algorithmic fairness literature \cite{KleinbergR18,celis2020interventions, EmelianovGGL20,celis2021longterm} studies the effectiveness of {\em group-wise} constraints that require at least a certain number of candidates from disadvantaged groups to be selected, in the special case of the problem with only {\em one} institution (and, hence, candidates' preferences are vacuous). 
   For instance, \cite{KleinbergR18} considers a model of bias where candidates are divided into two groups: the advantaged group $G_1$ and the disadvantaged group $G_2$.
    For an {\em unknown} $0<\beta\leq 1$,  the estimated utility of each disadvantaged candidate  $i$ is $\hu_i = \beta\cdot u_i$, and that of each advantaged candidate $j$ is $\hu_j = u_j.$
\cite{KleinbergR18,celis2020interventions} show that when the utilities are i.i.d., then requiring a proportional number of candidates from both advantaged and disadvantaged groups to be selected improves the latent utility of the assignment. 
    However, these results do not generalize to the case when $p>1$. 
    Simple examples show that just ensuring group representation in assignments doesn't guarantee that candidates from each group will be selected for their top choice institution. Specifically, disadvantaged groups may see lower selection rates for their preferred institutions due to biases.
We discuss additional related works in detail in \cref{sec:related_works}.

\paragraph{Our contributions.}
 We study the question: {To what extent can we achieve simultaneous goals of preference fairness, representational fairness, and utility maximization in the multi-institution selection problem with preferences in the presence of biases?}  
We assume that the true utilities and preferences are i.i.d. from distributions $\cD$ and $\cL$ respectively; see \cref{sec:model}.
We consider preference-based fairness metric $\sP$ 
which is the expectation of the ratio of the fraction of candidates from the disadvantaged group ($G_2$) that get assigned to their most preferred institution to the analogous fraction for $G_1$ (\cref{eq:preference}). 
To measure representational fairness, we consider the metric $\sR$ which
is the expectation of the ratio of the fraction of candidates in $G_2$ that are selected to that of $G_1$ (\cref{eq:representation}).
To measure utility, we consider a utility ratio $\sU$, which is defined as the expected value of the ratio of the true utility achieved in the presence of bias to that achieved without any bias (\cref{eq:utility}). 

Our first result studies the performance of the Gale-Shapley algorithm and provides upper bounds for $\sU$ and $\sP$ as a function of $\beta$; see \cref{thm:specialcase} and \cref{sec:effectofbias} for extensions.
When the distribution \(\mathcal{D}\) is uniform on \([0,1]\) and \(\beta\) approaches 0, both preference-based and representational fairness decrease to 0, while the utility ratio approaches \(\frac{2}{3}\).

We consider a type of representational constraint: institution-wise constraints that require a proportional number of candidates from each group to be matched to \textit{each} institution.
The algorithm \ouralgo{} (\cref{alg:GSTwoCopies}) then outputs an assignment 
that is group-wise stable 
subject to satisfying these institution-wise constraints.
Our second result is that \ouralgo{} has $\sP\approx 1$, $\sR \approx 1$, $\sU\approx 1$ (\cref{thm:institutionWiseConstraints}).
This result holds for any distributions $\cD$ and $\cL$.
The proofs of both our results rely on a novel Lipschitz property of the Gale-Shapley algorithm and non-trivial concentration and scaling results for the total number of candidates who are assigned their top choice by it; see~Section~\ref{sec:lipschitz} for details.

  We provide empirical validation of our theoretical results on synthetic and real-world datasets in \cref{sec:empirical} and \cref{sec:additional:empirical:relaxedbounds,sec:additional:empirical:noise,sec:additional:empirical:specialcase} to demonstrate their robustness under deviations from assumptions.
    In \cref{sec:empirical}, we study the performance of $\Ainst$ on real-world data from the 2009 IIT-JEE.
    In \cref{sec:empirical}, we evaluate $\Ainst$ on a synthetic dataset when the preferences of the two groups are drawn from different distributions. 
In \cref{sec:additional:empirical:noise}, we evaluate the performance of $\Ainst$ under different models of bias.
In all simulations, \ouralgo{} maintains a high preference-based fairness and utility ratio compared to the baselines.
For instance, for the IIT-JEE dataset, we observe that \ouralgo{} achieves significantly higher preference-based fairness ($\geq$0.9) than group-wise constraints ($\leq$0.50) (\cref{fig:simulation:real_world_data}). 
In \cref{sec:additional:empirical:relaxedbounds}, we also study the performance of institution-wise constraints under relaxed bounds.
Finally, in \cref{sec:additional:empirical:specialcase}, we show that \cref{thm:specialcase} is robust for different utility and bias models.

      \section{Other Related Work}\label{sec:related_works}

          The game-theoretic analysis of assignment and selection problems began with the foundational work by \cite{galeshapley1962}. Since then, a substantial body of literature, including \cite{roth1985common,roth1989incompleteinfo,roth1999redesign,das2005two,Rastegari2013MatchingPartialInformation,liu2014stable,aziz2016stable,liu2020matchinbandit,liu2021matching,yifei2022matching}, has expanded on this topic, as detailed in overviews like \cite{roth1992Handbook,manlove2013algorithmics}.
This research examines various aspects of assignment problems. It explores game-theoretic properties such as stability and truthfulness  \cite{roth1985common,roth1999redesign}, analyzes these properties under conditions of incomplete or uncertain information \cite{roth1989incompleteinfo,Rastegari2013MatchingPartialInformation,liu2014stable,aziz2016stable}, and uses machine learning tools to impute preferences and utilities \cite{das2005two,liu2020matchinbandit,liu2021matching,yifei2022matching}.

\smallskip
       Several works empirically study biases in algorithmic assignment/selection systems used in the real world, from centralized admissions systems to online labor markets \cite{grover2017brookingsSchoolBias,hannak2017freelance,lecher2021nycschoolbias,laverde2022distance}.
        \cite{remi2022statistical} complement these empirical studies with theoretical analysis: 
            They consider a model of bias utilities of candidates in the disadvantaged group has a higher noise than the utilities of candidates in the advantaged group and show that, somewhat surprisingly, increasing the noise for the disadvantaged candidates worsens the outcomes for both groups by reducing the probability that a candidate is assigned their first preference.
            In comparison, we consider a different model of bias (\cref{sec:model}) and bound the effect of bias on representational and preference-based fairness $\sP$.
            Our result implies a bound on the ratio of the probability that a disadvantaged candidate gets their first preference compared to an advantaged candidate.
                       However, unlike \cite{remi2022statistical}, we also demonstrate that institution-wise representational constraints can mitigate the adverse effect of bias (\cref{thm:institutionWiseConstraints,sec:empirical}).
            Motivated by these, several works design algorithms for finding assignments subject to different types of fairness constraints  \cite{Chierichetti0LV19,biswas2021toward,freeman2021two,Karni2022fairmatching}.
          {Notably, the algorithm of \cite{baswana2019centralized} has been implemented in an Indian context.}
            We focus on mitigating the adverse effects of bias in the centralized selection problem (\cref{thm:institutionWiseConstraints}).

        {\paragraph{The effect of representational constraints.}
            Several works study fundamental algorithmic tasks when the inputs may be biased.
            Examples include, subset selection \cite{KleinbergR18,EmelianovGGL20,salem2020closing,celis2021longterm,garg2021standardized,mehrotra2022intersectional,mehrotra2023submodular,boehmer2023multiwinner,MehrotraV23}, ranking \cite{celis2020interventions}, and classification \cite{blum2020recovering}. 
            Subset selection and ranking are particularly relevant to our study.

\cite{KleinbergR18} proposes a mathematical framework for understanding bias and demonstrates that implementing the Rooney Rule—a type of group-wise constraint \cite{passariello-implicit-rooney-wsj}—enhances the latent utility of selected subsets. \cite{EmelianovGGL20} explore a bias model where estimated utilities are noisier for disadvantaged candidates, showing positive impacts of group-wise constraints on latent utility. 
 Additionally, \cite{boehmer2023multiwinner} analyzes the benefits of group-wise constraints in multi-winner elections—a form of subset selection with submodular objectives based on multiple preference lists.

These studies explore various dimensions of group-wise constraints' effectiveness. For instance, \cite{salem2020closing} focus on online selection scenarios, \cite{celis2021longterm} investigate the long-term effects of biases on evaluators, \cite{mehrotra2022intersectional} broaden \cite{KleinbergR18}'s model to intersecting protected groups, and \cite{mehrotra2023submodular} adapt it to submodular objectives in recommendation systems. Our work builds on \cite{KleinbergR18}'s bias model but diverges by considering multiple institutions and focusing on ensuring fairness relative to candidates' preferences, where group-wise constraints alone are insufficient.

\smallskip
Ranking, which requires ordering candidates where earlier positions are more desired, can be regarded as a specialized form of assignment. 
Each rank position acts as an institution with a single capacity, and all candidates share identical preferences for earlier ranks over later ones. 

\smallskip
\cite{celis2020interventions} extend \cite{KleinbergR18}'s insights to ranking by showing that maintaining proportional group representation at each rank ensures near-optimal latent utility. However, the institution-specific constraints we propose differ significantly from the prefix-based constraints of \cite{celis2020interventions}, which assume uniform candidate preferences and are not applicable in contexts like ours where institutions vary in capacity.
In addition, our study incorporates preference-based fairness, a crucial element in the settings we aim to address (\cref{thm:institutionWiseConstraints}).

\paragraph{Matchings under distributional constraints.}
Several works study matchings under distributional constraints; see
\cite{Goto,KAMADA2017107,Kojima1} and the references therein.
Both–the works on distributional constraints in matching and this work–consider the task of computing a centralized assignment between two sets of agents (denoting, e.g., individuals, institutions, or items).
However, there are several differences:
Works on distributional constraints in matching consider the formulation where agents on both sides have preferences, whereas we consider the problem where agents on one side (denoting candidates) have individual preferences and the agents on the other side (denoting institutions) have certain (common) utilities for the individuals and their goal is utility maximization.
Our motivation comes from centralized admission contexts such as the admissions administered by the Joint Seat Allocation Authority in India where institutions share a common utility, e.g., scores in a standardized exam, and the goal of a centralized body is to match individuals to institutions to maximize the (total) utility of the institutions while respecting the preferences of the individuals as far as possible.
Moreover, the constraints considered by works in matching and mechanism design are similar to the ones considered in this work: the constraints are lower bounds or upper bounds on the number of individuals that can be matched to one or a set of institutions. The key difference is that the constraints in these works apply to all individuals, whereas we consider different constraints for candidates in different socially salient groups (to counter bias against certain groups).
To the best of our knowledge, unlike this work, works on distributional constraints in matchings do not consider biases in the preferences and/or utilities.

\paragraph{Other generalizations of stable assignments.}
There has been much work on generalizations of stable matchings, including weighted stable matchings and stable matchings in non-bipartite settings (see e.g., \cite{VANDEVATE1989147, IRVING1985577, Rothblum1992CharacterizationOS}). In the context of bias, 
\cite{remi2022statistical} considered a stable assignment setting with two groups of candidates and two institutions where both sides have preferences. In this model, the correlation between the two utility values assigned to a candidate by the institutions depends on the group to which the candidate belongs. However, the distribution of the utility of a candidate for a specific institution remains the same for both groups. In contrast,  the utilities of a candidate for different institutions remain the same in our model, and hence, are always perfectly correlated. Instead, the distribution of the utility of a candidate is group-specific -- for one of the groups, the utility values get scaled down by the bias parameter $\beta$.

\section{Model}\label{sec:model}
 Here we present the centralized selection problem, the model of bias, the model for utilities and preferences, and the fairness/utility metrics. Each component follows prior works in fairness and/or centralized assignment literature.

  \paragraph{Centralized selection with preferences.} 
       Given $n$ candidates and $p$ institutions with capacities $\mb{k}=(k_1,k_2,\dots,k_p)$ the goal is to select $K:= \sum_{\ell=1}^p k_\ell$ candidates and assign them to the institutions subject to the capacity constraint.
        We consider the centralized setting where each candidate $i$ has a true of {\em latent} utility $u_i \geq 0$ that captures the value generated if candidate $i$ is selected--it does not depend on the institution. 
        Further, each candidate also has a preference list $\sigma_i$ over the set $[p]\coloneqq\inbrace{1,2,\dots,p}$ of institutions.
        A selection $M$ is a partial function from the set of candidates $[n]$ to the set of institutions $[p]$. 
        For any  $M$ and index $\ell \in [p]$, let $M^{-1}(\ell)$ denote the set of candidates assigned to the $\ell$-th institution. 
        We extend this notation to a subset of $S \subseteq [p]$ of institutions -- $M^{-1}(S)$ denotes the set of candidates assigned to an institution in $S$. 
        One of the goals of the selection problem is to find an $M$ that maximizes the latent utility $\sum_{i\in M^{-1}([p])} u_i$ of the candidates assigned to an institution. %
        another desirable property is stability: consider two candidates $i$ and $j$ with $i$ having higher utility, then $M(i)$ is preferred over $M(j)$ in $\sigma_i$. 
       The Gale-Shapley algorithm is widely used for finding such an $M$ \cite{galeshapley1962,roth1999redesign} as they have many desirable properties, including latent-utility maximization and stability \cite{roth1985common,roth1992Handbook}.

{
\begin{remark}
This model of centralized selection with preferences mirrors several real-world admission systems, where each student is represented by a single utility value and schools rank these values to manage the complexities of large-scale admission processes.
For example, 
the Joint Entrance Examination (JEE Main), overseen by India's National Testing Agency, is among the most prestigious annual college entrance exams globally. In 2023, it drew over 1.17 million candidates. Candidates are evaluated in Mathematics, Physics, and Chemistry, with their scores combined into a single ranking for all participating institutions. For additional information, refer to the JEE Main FAQ \cite{JEE_Main_FAQ}.
The {\em Civil Services Examination in India} is also an annual test for entering around 15 civil services. Attracting over 500,000 applicants annually, the exam has ten sections. Scores from these sections are merged into one overall score. For more details, see  \cite{wiki:Civil_Services_Exam}.
China's Gaokao, an annual standardized test for undergraduate admissions, attracted over 12.9 million candidates in 2023. Students are assessed in six subjects, with their total score calculated as a weighted sum of these scores. Applicants select universities based on this score and their regional preferences. More details are available in  \cite{doi:10.1086/689773}.
\end{remark}
}

\paragraph{Model of bias in utilities.}
       We study the model introduced by \cite{KleinbergR18}.
     In this model, the candidates are divided into two groups: the advantaged group $G_1$ and the disadvantaged group $G_2$.
     The estimated utilities in this model are parameterized by an {\em unknown} bias parameter $0<\beta\leq 1$: given $\beta$, the estimated utility $\hu_i$ for a disadvantaged candidate is $\beta$-times their latent utility, i.e., $\hu_i=\beta\cdot u_i$, and that of an advantaged candidate $j$ is the same as their latent utility, i.e., $\hu_i=u_i$.

     The motivation to consider  $\beta<1$ and, hence, $\hat{u}_i\leq u_i$, comes from contexts (discussed in the introduction) where the estimated utility $\hat{u}_i$ of a candidate in the disadvantaged group systematically underestimates their true utility $u_i$.
     That said, even in this bias model, one can allow for $\hat{u}_i \geq u_i$ by setting $\beta \geq 1$. Our results continue to hold when $\beta \geq 1$. 
Moreover, in our empirical results, we consider the setting where $\beta$ may have some noise; see \ref{sec:additional:empirical:betanoise}.
The motivation to assume $\beta$ is unknown comes from the fact that the analysis is focused on the ``one-round'' setting of the selection problems. 
Finally,  the bias model assumes that the group identities are known and are reported truthfully. 
  This model is easily generalized to multiple disjoint groups by introducing a bias parameter for each group \cite{celis2020interventions}.

    \paragraph{Generative model of utilities and preferences.}
     Following \cite{KleinbergR18,celis2020interventions,EmelianovGGL20}, we let the latent utility $u_i$ of each candidate $i$ be drawn some distribution $\cD$ independent of all other candidates.
     Further, we assume that the preference list $\sigma_i$ of each candidate $i$ is drawn independently from some distribution $\cL$ of the set of all preference lists of $p$ institutions.

     I.I.D. utilities encode the fact that there are no systematic differences in utilities across groups.
     Different choices of $\cD$ arise in different contexts: 
        measures of popularity and success are observed to have power-law distributions \cite{clauset2009power}, 
        the percentile of candidates in a population has a uniform distribution on $[0,100]$, 
        and 
        normal distributions model utilities in standardized tests \cite{dorans2002recentering}.

    I.I.D. preference lists encode the assumption that candidates in either group have the same distribution of preferences over institutions.
    This is motivated by contexts where preferences are largely determined by the ``quality'' of the institution:
        For instance, most candidates applying to the JOSAA in India have similar preferences over universities \cite{mind2022IITallotment,verma2022JOSAA}, 
        in Boston and New York Public Schools, most parents prefer high-performing schools over others \cite{laverde2022distance},
       and 
       workers naturally prefer tasks with higher pay per time over tasks with a lower pay per time \cite{hannak2017freelance}. 
  {Note that the i.i.d. assumption allows candidates to have different preferences. For example, if $\cal L$ assigns probability $0.5$ to each of two preference lists $\sigma_1$ and $\sigma_2$ over $[p]$, then roughly half the candidates have preference $\sigma_1$ (and similarly for $\sigma_2$). }

    \paragraph{Utility ratio.}
    For a fixed utility vector $\mb{v}$ and a preference vector $\mb{\sigma}$,  given a selection $M_{\mb{v},\mb{\sigma}}:[n] \rightarrow [p]$, we  define the utility of $M_{\mb{v},\mb{\sigma}}$ as the sum of utilities of all candidates that are assigned to any institution, i.e.,   
    $\mathsf{U}(M_{\mb{v},\mb{\sigma}}):=\sum_{i \in M^{-1}([p])} v_i$.
    Since $\mb{\hat{u}}$ is a scaled-down version of $\mb{u}$, note that the maximum value $\mathsf{U}(\cdot)$ can take is the sum of top $K$ entries in the true utility vector $\mb{u}$; we denote this by $\mathsf{U}^\star(\mb{u})$.
    We measure the performance of an assignment with respect to $\mathsf{U}^\star$.
    Consider an algorithm $\cA$ that, given preferences $\mb{\sigma}$ and utilities $\mb{\hu}$, outputs a selection $M_{\mb{\hu},\mb{\sigma}}$. The utility-ratio of $\cA$ is defined as:
    \begin{equation}\label{eq:utility}
           \textstyle
  \sU_{\cD,\cL}(\cA) 
        \coloneqq 
        \Ex_{\mb{u}\sim \cD,\mb{\sigma}\sim \cL} \frac{\mathsf{U}(M_{\mb{\hat{u}},\mb{\sigma}})}{\mathsf{U}^\star(\mb{u})}. \end{equation} 
       {It is important to note that the algorithm's assignment relies on the {\em estimated} utilities $\mb{\hu}$ of the candidates. However, when measuring the utility ratio, we use the {\em true} utilities of the selected candidates.}

  \paragraph{Fairness metrics.} 
There is a wide body of work on metrics for fairness \cite{fairmlbook}. The two fairness metrics considered here are grounded in the idea of ``proportional representation'', a common and useful notion of fairness in many contexts.
        \textit{Representational fairness} disregards candidate preferences and only considers how many candidates from each group are assigned to at least one institution. 
For group $j \in \{1,2\}$, let $\rho_j$ denote the fraction of candidates in $G_j$ that are selected by $M_{\mb{\hu}, \mb{\sigma}}$.
Representational fairness measures the disparity in the values of $\rho_1$ and $\rho_2$. 
We consider a multiplicative notion of comparing the performance across groups.
For an algorithm $\cA$, its representational fairness is defined as: 
  \begin{equation}\label{eq:representation}
           \textstyle  \sR_{\cD,\cL}(\cA) \coloneqq 
            \Ex_{\mb{u}\sim \cD,\mb{\sigma}\sim \cL} \frac{\min_{j\in \{1,2\}}\rho_j}{\max_{j'\in \{1,2\}} \rho_{j'}}.
            \end{equation}
        This definition is similar in spirit to the fairness metrics used by past works \cite{celis2019classification}: they require the value of some desirable quality to be similar for different groups.
        By definition, $\sR_{\cD,\cL}(\cA)$ is a value between 0 and 1. 
        $\sR_{\cD,\cL}(\cA)$ is close to 1 if the $\cA$ ensures that a proportional number of candidates from each group are assigned to at least one institution.
        The larger the value of $\sR_{\cD,\cL}(\cA)$ is the more ``fair'' $\cA$ is.

   \emph{Preference-based fairness} captures the disparity in the fraction of candidates in each group that are assigned to their ``top'' preferences.
    The definition of top preference is context-dependent: 
        if there are a small number of institutions, as in school admissions, the first preference of students is important \cite{laverde2022distance} (see also \cite{fairExposureAshudeep, celis2018ranking,manning2010introduction}), and if there are a large number of institutions, as in the online labor market, then the relevant notion may be the first $\ell$ preferences for some $1\leq\ell\leq p$.
Let $\pi_j$ denote the fraction of candidates in group $G_j$ that get their first preference. For an algorithm $\cA$, its preference-based fairness is defined as: 

  \begin{equation}\label{eq:preference}
           \textstyle  \sP_{\cD,\cL}(\cA) \coloneqq 
            %
            \Ex_{\mb{u}\sim \cD,\mb{\sigma}\sim \cL} \frac{\min_{j\in \{1,2\}}\pi_j}{\max_{j'\in \{1,2\}} \pi_{j'}}.
            \end{equation}
             $\sP_{\cD,\cL}(\cA)$ is a value between 0 and 1 and large values are preferable. 
These three metrics extend to the setting of multiple disjoint groups by considering the minimum and maximum over all groups.
Similarly, for $\ell \in [p]$, one can consider the performance $\pi_j^{(\ell)}$, which considers the fraction of candidates in group $G_j$ that get an assignment from among their top-$\ell$ choices; and similarly  define $\sP^{(\ell)}_{\cD,\cL}(\cA).$
When $\cD$ and $\cL$ are clear, we drop the subscripts from $\sU_{\cD,\cL},\sR_{\cD,\cL}$, and $\sP_{\cD,\cL}$ and use $\sU$, $\sR$, and $\sP$, $\sP^{(\ell)}$ etc.

\section{Theoretical Results}  
\label{sec:theory-results}
  In this section, we first present theoretical results on the effect of the bias parameter on the fairness and utility metrics of a stable assignment algorithm. Subsequently, we develop an algorithm for enhancing these metrics by adding institutional representation constraints. Finally, we present an overview of the ideas involved in proving these results. 
    
    \paragraph{Performance of a stable assignment in the presence of bias.}
       We consider how stable allocation algorithms that take into account the preferences of participating members perform in the presence of bias. 
       In our special setting, where the preferences of institutions are derived from a centralized evaluation of candidates, there is a unique stable assignment of candidates that can be obtained by the following algorithm: consider the candidates in decreasing order of observed utilities. 
        When considering a candidate, assign it to its most preferred institution that still has an available spot (see~\Cref{def:stable} for the definition of a stable assignment, and \cref{cl:stableassignment} for the uniqueness of such an assignment).   This algorithm, denoted $\Ast$, is formally described in~\Cref{alg:GSLatUtil}. When there is no bias, i.e. $\beta = 1$, it can be shown that $\Ast$ has near-optimal representational and preference-bases fairness, and near-optimal utility (see~\Cref{cor:betaone}). Our first result shows that the performance of $\Ast$ deteriorates as $\beta$ decreases:

\begin{restatable}{theorem}{effectofbiasspecial}
            \label{thm:specialcase}
            Consider an instance where the utilities of the candidates are drawn from the uniform distribution on $[0,1]$ and the distribution over preferences is arbitrary. Assume $n_1 = n_2=K$. Then,    $\sP(\Ast) \leq \beta + O \left( \frac{p\sqrt{\log n}}{{\sqrt{n}}} \right),$
            $\sR(\Ast) = \beta \pm O \left( \frac{\sqrt{\log n}}{{\sqrt{n}}} \right),$
            and 
            $\sU(\Ast) = \frac{2}{3} + \frac{  4 \beta}{3(\beta+1)^2} 
            \pm O \left( \frac{\sqrt{\log n}}{{\sqrt{n}}} \right).$
        \end{restatable}
\noindent
The above result shows that when the bias parameter $\beta$ is close to zero, then the fairness metrics $\sR(\Ast)$ and $\sU(\Ast)$  deteriorate substantially, and there is a significant drop in the utility of the selected candidates as well. It is also worth noting that the error terms in the theorem decay as   $p/\sqrt{n}$ and are negligible for reasonable values of $n$ and $p$. 
The proof of this result appears in Section~\ref{sec:effectofbias} and develops novel results on Lipschitz and concentration properties of a stable assignment in our setting. It is also worth noting that the proof first shows a high probability bound on the utility and fairness metrics and then uses this to bound the expectation of the respective metrics.  We present an outline of the proof at the end of this section.

The theorem extends to arbitrary values of $n_1, n_2$ and $K$ (see~\Cref{thm:effectOfBias}  for details). For example, when $n_1 = n_2$ and $K = \alpha n_1$ for some parameter $\alpha \geq 1-\beta$, then both $\sR(\Ast)$ and $\sP(\Ast)$ are at most $\frac{\alpha - (1-\beta)}{\alpha \beta + (1-\beta) }$. This result generalizes significantly to arbitrary log-concave distributions of utilities that include (truncated) Gaussian and Pareto (see Section~\ref{sec:logconcave} and also empirical results in Section~\ref{sec:paretogaussiantest4.1}). It is also worth noting that such a result does not hold for arbitrary distributions of utilities (see~\Cref{ex:logconcave}). We also empirically validate the robustness of ~\Cref{thm:specialcase} on other models of bias (see Sections~\ref{sec:othermodels1-4.1}, \ref{sec:othermodels2-4.1}).

   \paragraph{Enhancing preference-based fairness via institution-wise representational constraints.}
Our second result shows an algorithm that can enhance preference-based fairness while ensuring near-optimal utility.
   \begin{restatable}
   {theorem}{thminstitutionWiseConstraints}
    \label{thm:institutionWiseConstraints}
            Let $\eta_1, \eta_2, \eta_3 > 0$ be parameters  such that $|G_j| \geq \eta_1 n$ for each $j \in \{1,2\}$, $K \geq \eta_2 n$,  and $k_\ell \geq \eta_3 K$ for each $\ell \in [p]$. 
            There is an allocation algorithm $\Ainst$ (stated formally in~\Cref{alg:GSTwoCopies}) such that, 
            for any distribution of utilities and preference lists, and bias parameter $\beta$, 
               $ \sP(\Ainst) \geq 1 - O\left(\frac{p \sqrt{ \log K}}{\eta_1 \eta_3 \sqrt{K}} \right), $
         $  \textstyle 
                \sU(\ouralgo{})\ 
                \geq \ 
                1 - O\left(\frac{\sqrt{\log n}}{\sqrt{\eta_2 n}} \right)$ and
                $\sR(\ouralgo{}) = 1.$
        \end{restatable}

     \noindent   
        The error bounds in the theorem are negligible in typical scenarios. 
        Consider for example a setting where $|G_1| = |G_2| = n/2, K = cn$ for some constant $c$, and all $p$ institutions have the same capacity. Then $\eta_1 = 1/2, \eta_2 = c$ and $\eta_3 = \sfrac{1}{p}$. Thus the error bounds for $\sP(\Ainst)$ and $\sU(\Ainst)$ decay as $\sfrac{p^2}{\sqrt{K}}$ and $\sfrac{1}{\sqrt{n}}$ respectively.
        The proof of this result appears in Section~\ref{sec:proofof:thm:institutionWiseConstraints}. As in the case of~\Cref{thm:specialcase}, we show a stronger statement that the fairness and the utility metrics of $\Ainst$ are close to 1 with high probability. 
        The algorithm~$\Ainst$ circumvents the strong lower bounds mentioned in~\Cref{thm:specialcase} arising due to the bias parameter $\beta$. 
        It is worth emphasizing that the results in~\Cref{thm:institutionWiseConstraints} hold for arbitrary distributions $\cD$ and $\cL$ over utilities and preferences, and the stated bounds in the theorem do not depend on the bias parameter. The theorem does 
 require that the distribution of the utilities and the preference lists of candidates are identical and independent. Removing this i.i.d. assumption seems challenging and we leave it as an interesting avenue for future work. 
         We now motivate the ideas in the algorithm~$\Ainst$ {by arguing that several natural algorithms do not satisfy the conclusions of~\Cref{thm:institutionWiseConstraints}: 

\paragraph{(i) Algorithm $\Ainst$:}
As~\Cref{thm:specialcase} shows, running the algorithm $\Ast$ directly on the original instance could lead to a significant loss in utility and preference-based fairness. A natural fix would be first to learn (and remove the effect of) the bias parameter $\beta$ and then apply $\Ast$ on the ``corrected'' input.
We give a detailed discussion in Section~\ref{sec:discussionAinst} on why such an approach is highly impractical. 

\paragraph{(ii) Algorithm $\Ainst$ with proportional group-wise representational constraints:} This approach is rooted in in prior works \cite{KleinbergR18, celis2020interventions} on the single institution setting.  
For instance, we could run $\Ast$ on the set of candidates selected from the top  $\sfrac{|G_j|}{n}$ fraction from each group $G_j$ (see~\Cref{alg:groupwiseconstraints} for a formal description). 
While this algorithm can be shown to enhance representational fairness while obtaining near-optimal utility, it may have low preference-based fairness  (see Section~\ref{sec:discussionAinst}). 

\paragraph{(iii) Algorithms that directly enforce preference-based fairness: }
An approach for generalizing the known results for the single institution setting to multiple institutions while directly enforcing high preference-based fairness would be the following (assume for the sake of simplicity that $|G_1|=|G_2|$) -- each institution $\ell$ considers the candidates that prefer $\ell$ as their first choice. We assign the highest utility candidates from this subset to institution $\ell$ while ensuring that no more than half the capacity at this institution is allocated to a single group (in case some slots are vacant at the end, we can fill them with the highest utility unassigned candidates). 
Such an algorithm $\cA$ would indeed achieve high preference-based fairness, but may not obtain near-optimal utility even when there is no bias -- see Section~\ref{sec:discussionAinst} for an example.  

      We consider a different family of representational constraints: institution-wise constraints.
        These constraints require an adequate number of candidates from each group to be assigned to \textit{each} institution.
Formally, we require that for each institution $\ell \in [p]$,  $\sfrac{|G_j|}{n} \cdot k_\ell$ candidates from each group $G_j$ are assigned to it, where $k_\ell$ is the capacity of the $\ell$-th institution. 
  The algorithm $\ouralgo{}$, which is a strict generalization of the approaches for single institution setting in  \cite{KleinbergR18, celis2020interventions}, satisfies the proportional representational constraints as follows:
            We invoke two independent instantiations of the 
            algorithm $\Ast$ described at the beginning of Section \ref{sec:theory-results}.
            The first instance is run with the capacity of each institution $\ell$ set to $k_{\ell}\cdot \sfrac{n_1}{n}$ and restricting the candidates to $G_1$ only.
            Similarly, the second instance is run with the capacity of each institution $\ell$ set to $k_{\ell}\cdot \sfrac{n_2}{n}$ and restricting the candidates to those from $G_2$ only. 
            A formal description of \ouralgo{} appears in \Cref{alg:GSTwoCopies}. 
It is worth noting that this algorithm does not explicitly enforce preference-based constraints and does not require the knowledge of $\beta$. It is easy to show that $\Ainst$ is group-wise stable, Pareto-efficient, and strategy-proof (see Section~\ref{sec:discussionAinst}). 
The algorithm $\Ainst$ and the results in~\Cref{thm:institutionWiseConstraints} can also be extended to settings with multiple groups and preference-based fairness metric involving a fraction of candidates that are assigned an institution among their top-$\ell$ choices for a parameter $\ell$.

We empirically validate the conclusions of~\Cref{thm:institutionWiseConstraints} on the following robust settings: (i) real-world data from India's centralist IIT admission system, where the i.i.d. assumption on the utilities and the preference lists of candidates may not hold and compare $\Ainst$ with algorithms based on group-wise constraints only (see Section~\ref{sec:empirical}), (ii)   relaxed versions of the institutional constraints where less than $\sfrac{|G_j|}{n_j} \cdot k_\ell$ slots may be reserved in the $\ell$-the institution for each group $G_j$ (see~\Cref{alg:grouprelaxed} for the formal description of such an algorithm and Section~\ref{sec:additional:empirical:relaxedbounds} for related empirical results), (iii) execute $\Ainst$ on data generated using other models of bias (see Sections~\ref{sec:additional:empirical:betanoise}, \ref{sec:additional:empirical:impvar}), and (iv) allow two different distributions over the preferences for candidates from the two respective groups (see Section~\ref{sec:additional:empirical:dispersionpref}). Note that our results cannot be extended directly to settings where institutions have varying utilities for candidates. Indeed, in such settings, the algorithm $\Ast$ (restricted to candidates from the same group) is not well defined. One could replace $\Ast$ with another stable matching algorithm, but such an algorithm may not have both near-optimal utility and preference-based fairness.

       \paragraph{{Outline of the proofs of Theorems \ref{thm:specialcase} and \ref{thm:institutionWiseConstraints}}.}
             We now describe the main ideas in the proofs of~\Cref{thm:specialcase} and~\Cref{thm:institutionWiseConstraints}. 
             It is well known that the algorithm $\Ast$, also known as a {\em serial dictatorship mechanism}, is stable (see e.g. ~\cite{Svensson99}).  However, proving either of these results requires us to delve much deeper into the properties of the algorithm $\Ast$. In particular, we study the Lipschitz properties of the assignment given by $\Ast$ and the concentration properties of this assignment when preference lists are drawn from a distribution. We begin by stating the following concentration bound.
             \begin{lemma}[\bf Informal; see~\Cref{lem:concentrationbounds}]
                 \label{lem:concinformal}
                 Let $G$ be a subset of candidates belonging to the same group. Let $S$ be the subset of selected candidates and $\pi(S)$ denote the candidates in $S$ that receive their top choice. Then, 
                  $ \textstyle \Pr\left[\left| |S \cap G| - \Ex[|S \cap G|] \right| \geq t\right] \leq 2 e^{\sfrac{-2t^2}{|G|}}, $
        and 
        $ \Pr \left[ \left| |\pi(S) \cap G| - \Ex[|\pi(S) \cap G|] \right| \geq t \right] \leq 2 e^{-\sfrac{2t^2}{(p^2K)}}. $
             \end{lemma}
            We now give an outline of the proof of this result. We can express $|S \cap G|$ as $\sum_{i \in G} X_i$, where $X_i$ are negatively correlated random variables.  Therefore, we can apply Chernoff-type concentration bounds here. However, the quantity $|\pi(S) \cap G|$ is much trickier to handle. 
           Clearly, $|\pi(S) \cap G|$ can be again be written as $\sum_{i \in G} Y_i$, where $Y_i$ is an indicator variable denoting whether the $i$-th candidate in $G$ received the most preferred choice. 
            Unlike the case of $X_i$'s, the variables $Y_{i}$'s may not be independent or negatively correlated. 
            One possibility is to use concentration bounds relying on the Lipschitz property of the assignment (see e.g., \Cref{thm:lipschitz}), i.e.,  changing the preference list of only one candidate
            changes the assignment slightly. 
            However, one can construct examples for general stable matching instances (see~\Cref{ex:counterexample:stable:Lipschitz}), where such a change can dramatically alter the stable assignment.
            Our problem has the additional feature that all the institutions agree on a central ranking of the candidates. 
            For such settings, we show that 
             there is a unique stable assignment (\Cref{cl:stableassignment}), and is given by  algorithm $\Ast$ (formally described in~\Cref{alg:GSLatUtil}). We can compare two runs of $\Ast$ on instances that differ on only one candidate. Let $S$ and $S'$ be the set of candidates selected by $\Ast$ in these two runs. We show that {the size of the set difference between $\pi(S)$ and $\pi(S')$, i.e., $|(\pi(S) \setminus \pi(S')) \cup (\pi(S') \setminus \pi(S))|$ is at most $p+1$} (\Cref{lem:lipschitzstable1}). 
            This Lipschitz property of $\Ast$ allows us to show the second concentration bound mentioned in~\Cref{lem:concinformal}.

            We now give an overview of the proof of~\Cref{thm:specialcase} (and the more general result~\Cref{thm:effectOfBias}). Assume for the sake of simplicity that $|G_1| = |G_2|.$ We first consider $\sR(\Ast)$. Let $S$ be the set of candidates selected by $\Ast$. The metric $\sR(\cA)$ considers the expectation of the ratio $\frac{|S \cap G_2|}{|S \cap G_1|}.$ Using~\Cref{lem:concinformal}, we can closely approximate this by $\frac{\Ex[|S \cap G_2|]}{\Ex[|S \cap G_1|]}.$  We now find the threshold $\Delta \in [0,1]$ such that $S$ is closely approximated by candidates with utility at least $\Delta$. This allows us to estimate both of the desired expectations.  The utility metric $\sU(\Ast)$ can 
            be estimated by evaluating the expected utility of top $\Ex[|S \cap G_j|]$ candidates from each of the groups $G_j$. 
            We now consider $\sP(\Ast)$ which is defined as $\Ex \left[\frac{\pi(S \cap G_2)}{\pi(S \cap G_1)} \right]$ (recall that $\pi(X)$ denotes the number of candidates in a set $X$ that are assigned their top choice by $\Ast$). Again, using the concentration bound in~\Cref{lem:concinformal}, it suffices to bound  $\frac{\Ex[\pi(S \cap G_2)]}{\Ex[\pi(S \cap G_1)]}$.    However, we cannot simply write down a closed form expression for $\Ex[\pi(S \cap G_j)]$. Instead, we show (here $\Delta$ is as defined above):
            \begin{lemma}[\bf Informal; see \cref{cl:AB}]
            \label{lem:informalpref1}
                Let $A$ and $B$ be the set of candidates with observed utility in the range $[\Delta,1]$ and $[0, \Delta]$ respectively. Conditioning on $g_1 := |A|, g_1':= |B \cap G_1|, g_2' :=|B \cap G_2|$,  $\pi(G_1)$ and $\pi(G_2)$ are closely approximated by $\pi(A) + \frac{g_1'}{g_1'+g_2'} \cdot \pi(B)$ and $\frac{g_2'}{g_1'+g_2'} \cdot \pi(B)$ respectively.
            \end{lemma}
            Using the above result along with the fact that candidates in $A$ are more likely to get their first choice than in $B$, and hence $\pi(A)/g_1 \geq \pi(B)/(g_1'+g_2')$ (\Cref{cl:expectationbounds}), we can bound $\sU(\Ast)$. Thus, we show all the three desired bounds in~\Cref{thm:specialcase}.

            We now give an overview of the proof of~\Cref{thm:institutionWiseConstraints}. The metric $\sR(\Ainst)=1$ by the definition of the algorithm $\Ainst$. Indeed, we ensure that the number of selected candidates from each group $G_j$ is proportional to $|G_j|$. The utility ratio $\sU(\Ast)$ is also close to 1 for a similar reason: the institution-wise representational constraints imply that the total utility of the top $K$ candidates is close to the sum over each group $G_j$ of the total utility of the top  $\frac{|G_j|}{n} \cdot K$ candidates from $G_j$. 
            We now consider $\sP(\Ainst)$, which is defined as $\Ex \left[\frac{|\pi(S \cap G_2)|/|G_2|}{|\pi(S \cap G_1)|/|G_1|} \right]$.~\Cref{lem:concinformal} allows us to approximate this by the ratio of the corresponding expectations $\Ex[|\pi(S \cap G_j)|/|G_j|]$. In light of the representational constraints,   one would expect  $\Ex[|\pi(S \cap G_j)|]/|G_j|]$ to be the same for both groups. Somewhat surprisingly, this fact turns out to be false (see~\Cref{ex:expectation}). In particular, consider the following two instances (consisting of just one group): in the instance $\cI$ there are $n$ candidates and capacity $k_j$ at each institution $j \in [p]$; whereas in the ``scaled'' instance  $\cI_a$, where $a \geq 1$,  there are $n/a$ candidates with $k_j/a$ capacity at each institution $j \in [p]$. 
            Let $M$ and $M_a$ be the corresponding assignments given by $\Ast$, and  $\pi(M)$ and $\pi(M_a)$ be the number of candidates who receive their top choices in these two assignments. Then $\Ex[\pi(M)]$ may not equal $ a\Ex[\pi(M_a)]$. However, we can bound their gap: 
            \begin{restatable}{lemma}{lemgap}
                \label{lem:gap}
        $\left|  \frac{\Ex[\pi(M)]}{a} - \Ex[\pi(M_a)] \right| = O(p\sqrt{K \log K}). $
            \end{restatable}
            \noindent
            The proof of this result 
            relies on yet another Lipschitz property of the algorithm $\Ast$: changing the capacity of an institution by one unit changes the allocation for at most $p$ candidates (\Cref{lem:lipschitzs2}). 
            \Cref{lem:gap} allows us to show that $\Ex \left[\frac{|\pi(S \cap G_2)|/|G_2|}{|\pi(S \cap G_1)|/|G_1|} \right]$ is close 1. Thus, we can show that all three metrics are near optimal for $\Ainst$. 

\section{Pseudocode for Algorithms with Institution-Wise, Group-Wise and No Constraints}\label{sec:additional:pseudocode}
    In this section, we provide the pseudocodes for the three algorithms we compare in \cref{sec:empirical}, as well as extensions for relaxed bounds in \cref{sec:additional:empirical:relaxedbounds}. These include a special case of the Gale-Shapley algorithm ($\Ast$; \cref{alg:GSLatUtil}), the algorithm to satisfy group-wise proportional representation constraints ($\Agroup$; \cref{alg:groupwiseconstraints}), the algorithm to satisfy institution-wise constraints ($\Ainst$; \cref{alg:GSTwoCopies}), and modifications of \cref{alg:groupwiseconstraints} and \cref{alg:GSTwoCopies} under relaxed bounds.
        We implement algorithms and run all empirical work in \texttt{Python3}. The code implementations of algorithms, fairness metrics, as well as simulation procedures and data, can be found \href{https://github.com/sandrewxu/CentralizedSelectionwithPreferenceBias}{here}.\footnote{ \url{https://github.com/sandrewxu/CentralizedSelectionwithPreferenceBias}}
        
        \paragraph{Unconstrained algorithm $\Ast$.}
        This is a special case of the Gale-Shapley algorithm for stable matching \cite{galeshapley1962} where all the institutions agree on a common central ranking of the candidates.
        The algorithm considers the candidates in decreasing order of observed utility. When considering a candidate, it assigns it to the highest preferred institution that still has an available slot. 
        \begin{algorithm}[ht!]
            \caption{Unconstrained algorithm ($\Ast$)}\label{alg:GSLatUtil}
            \begin{algorithmic}[1]
                \REQUIRE $S \text{ (set of candidates)}$; $\textbf{\textit{k}}=(k_1, k_2, ..., k_p) \text{ (institutional capacities)}$; $\sigma = (\sigma_1, \sigma_2, \ldots, \sigma_n) \text{ (preferences for } i \in S$ \text{)}; $\hat{u} = (\hat{u}_1, \hat{u}_2, \ldots, \hat{u}_n) \text{ (estimated utilities for } i \in S \text{)}$
                \STATE $r_{unmatched} \gets$ candidates in $S$ sorted in order of decreasing utility
                \STATE $M \gets $ assignment filled with -1 (no institution)
                \FOR {candidate $a$ in $r_{unmatched}$}
                \STATE $M_a \gets $ top choice institution $n$ such that $k_n > 0$
                \STATE $k_n \gets k_n-1$
                \STATE \textbf{until} $k_n = 0$ for all $n \in k$
                \ENDFOR
                \STATE \textbf{return} $M$
            \end{algorithmic}
        \end{algorithm}

        \paragraph{Algorithm with group-wise constraints.} The algorithm first selects 
        the top (according to observed utilities) $\frac{|G_j|}{n} \cdot K$ candidates from each group $G_j$. It then 
        runs the algorithm $\Ast$ to assign these $K$ candidates to institutions. This guarantees representation-based fairness for both groups.

        \begin{algorithm}[ht!]
            \caption{Group-wise algorithm ($\Agroup$)}\label{alg:groupwiseconstraints}
            \begin{algorithmic}[1]
                \REQUIRE $G_1 \text{ (set of candidates in group 1)}$; $G_2 \text{ (set of candidates in group 2)}$; $\textbf{\textit{k}}=(k_1, k_2,\ldots, k_p) \text{ (institutional capacities)}$; $\sigma = (\sigma_1, \sigma_2, \ldots , \sigma_n) \text{ (preferences for each } i \in S$ \text{)}; $\hat{u} = (\hat{u}_1, \hat{u}_2,\ldots, \hat{u}_n) \text{ (estimated utilities for each } i \in S$ \text{)}
                \STATE $m_1 \gets (\sum_\ell k_{\ell}) \cdot |G_1| / n$
                \STATE $S_{1} \gets$ first $m_1$ candidates in $G_1$ sorted in order of decreasing utility
                \STATE $m_2 \gets (\sum_\ell k_{\ell}) \cdot |G_2| / n$
                \STATE $S_{2} \gets$ first $m_2$ candidates in $G_2$ sorted in order of decreasing utility
                \STATE $S_0 \gets S_1 \cup S_2$
                \STATE \textbf{return} $\Ast(S_0, \textbf{\textit{k}}, \sigma, \hat{u})$
            \end{algorithmic}
        \end{algorithm}

        \paragraph{Algorithm with institution-wise constraints.} This algorithm creates two independent instances, one for each group. In the instance corresponding to group $G_j$, the capacity at an institution $i$ is set to $\frac{|G_j|}{n} \cdot k_i$. Then $\Ast$ is independently executed on each of these two instances.

        \begin{algorithm}[ht!]
            \caption{Institution-wise algorithm (\ouralgo{})}\label{alg:GSTwoCopies}
            \begin{algorithmic}[1]
                \REQUIRE $G_1 \text{ (set of candidates in group 1)}$; $G_2 \text{ (set of candidates in group 2)}$; $\textbf{\textit{k}}=(k_1, k_2, \ldots, k_p) \text{ (institutional capacities)}$; $\sigma = (\sigma_1, \sigma_2, \ldots, \sigma_n) \text{ (preferences for each } i \in S$ \text{)}; $\hat{u} = (\hat{u}_1, \hat{u}_2, \ldots, \hat{u}_n) \text{ (estimated utilities for each } i \in S$ \text{)}
                \STATE $A \gets \{k_i \cdot |G_1| / n: i \in [p]\}$ {\em \small \color{blue} $/*$ Define capacity vector $A$ for the $p$ institutions for candidates in $G_1$ $*/$} 
                \STATE $B \gets \{k_i \cdot |G_2| / n: i \in [p] \}$ {\em \small \color{blue} $/*$ Define capacity vector $B$ for the $p$ institutions for candidates in $G_2$$*/$}
                \STATE $M_{1} \gets \Ast(G_1, A, \{\sigma_i : i \in A\}, \{\hat{u}_i : i \in G_1\})$
                \STATE $M_{2} \gets \Ast(G_2, B, \{\sigma_i : i \in B\}, \{\hat{u}_i : i \in G_2\})$
                \STATE \textbf{return} $M_1 \cup M_2$ 
            \end{algorithmic}
        \end{algorithm}
        
        \paragraph{Algorithm with relaxed group-wise constraints.}
        This algorithm satisfies relaxed group-wise constraints (\cref{sec:additional:empirical:relaxedbounds})  specified by a parameter $\alpha$, where $0 \leq \alpha \leq 1$. For each group $G_j$, it first selects the top $\alpha K \cdot   \frac{|G_j|}{n}$ candidates. From the remaining pool of $n-\alpha K$ candidates, it selects the top $K - \alpha K$ candidates. Finally, the algorithm $\Ast$ is executed on these $K$ selected candidates. Observe that when $\alpha=0$, this algorithm is same as the unconstrained algorithm $\Ast$; and when $\alpha=1$, we recover the algorithm $\Agroup.$

        \begin{algorithm}[ht!]
            \caption{Relaxed group-wise constraints algorithm}\label{alg:grouprelaxed}
            \begin{algorithmic}[1]
                \REQUIRE $G_1 \text{ (set of candidates in group 1)}$; $G_2 \text{ (set of candidates in group 2)}$; $\textbf{\textit{k}}=(k_1, k_2, \ldots, k_p) \text{ (institutional capacities)}$; $\sigma = (\sigma_1, \sigma_2, \ldots, \sigma_n) \text{ (preferences for each } i \in S$ \text{)}; $\hat{u} = (\hat{u}_1, \hat{u}_2, \ldots, \hat{u}_n) \text{ (estimated utilities for each } i \in S$ \text{)}; $\alpha$ (proportion of strictness)
                \STATE $m_1 \gets \alpha 
                \cdot \sum k_{\ell} \cdot |G_1| / n$ {\em \small \color{blue} $/*$ Minimum representation for candidates in $G_1$ is proportional to $|G_1|$ and $\alpha$  $*/$}
                \STATE $m_2 \gets \alpha \cdot \sum k_{\ell} \cdot |G_2| / n$ {\em \small \color{blue} $/*$ Minimum representation for candidates in $G_2$ is proportional to $|G_2|$ and $\alpha$  $*/$}
                \STATE $m_{excess} \gets \sum k_{\ell} - m_1 - m_2$ {\em \small \color{blue} $/*$ Capacities not allocated to a specific group when $\alpha < 1$  $*/$}
                \STATE $r_{unmatched} \gets$ candidates in $S$ sorted in order of decreasing utility
                \STATE $M \gets $ assignment filled with -1 (no institution)
                \FOR {candidate $a$ in $r_{unmatched}$}
                \STATE $M_a \gets \text{top choice institution } n \text{ such that } k_n > 0 \text{ and } m_{group} > 0 \text{ or } m_{excess} > 0$
                \STATE $k_n \gets k_n-1$
                \IF {$m_{group} > 0$}
                \STATE $m_{group} \gets m_{group} - 1$ {\em \small \color{blue} $/*$ $m_{group}$ is $m_1$ if the candidate is in $G_1$ and $m_2$ if the candidate is in $G_2$ $*/$} 
                \ELSE
                \STATE $m_{excess} \gets m_{excess} - 1$
                \ENDIF
                \STATE \textbf{until } $k_n = 0$ for all $n \in k$
                \ENDFOR 
                \STATE \textbf{return} $M$
            \end{algorithmic}
        \end{algorithm}

        \paragraph{Algorithm with relaxed institution-wise constraints.}
        This algorithm satisfies relaxed institution-wise constraints (\cref{sec:additional:empirical:relaxedbounds}) specified by a parameter $\alpha$, where $0 \leq \alpha \leq 1$. For each institution $i$ and group $G_j$, it reserves $\alpha k_i \cdot |G_j|/n$ slots for candidates from group $G_j$, and the remaining $(1-\alpha)k_i$ ``excess'' slots are available for candidates from either group. The algorithm $\Ast$ is now modified as follows: it again considers candidates in decreasing order of observed utility values. When considering a candidate $a$,  say from group $G_j$, it assigns it to the highest ranked institution $i$ that has an available slot in either the ones reserved for $G_j$ or in the unreserved pool. 
        Observe that if $\alpha = 1$, we recover $\Ainst$, and when $\alpha = 1$, we execute the unconstrained algorithm $\Ast$.
        \begin{algorithm}[ht!]
            \caption{Relaxed institution-wide constraints algorithm}\label{alg:instrelaxed}
            \begin{algorithmic}[1]
                \REQUIRE $G_1 \text{ (set of candidates in group 1)}$; $G_2 \text{ (set of candidates in group 2)}$; $\textbf{\textit{k}}=(k_1, k_2, \ldots, k_p) \text{ (institutional capacities)}$; $\sigma = (\sigma_1, \sigma_2, \ldots, \sigma_n) \text{ (preferences for each } i \in S$ \text{)}; $\hat{u} = (\hat{u}_1, \hat{u}_2, \ldots, \hat{u}_n) \text{ (estimated utilities for each } i \in S$ \text{)}; $\alpha$ (proportion of strictness)
                \STATE $k_{G_1} \gets \{k_i \cdot \alpha \cdot |G_1| / n: i \in [p]\}$ {\em \small \color{blue} $/*$ Minimum institutional representation for candidates in $G_1$, proportional to $|G_1|$ and $\alpha$  $*/$} 
                \STATE $k_{G_2} \gets \{k_i \cdot \alpha \cdot |G_2| / n: i \in [p] \}$ {\em \small \color{blue} $/*$ Minimum institutional representation for candidates in $G_2$, proportional to $|G_2|$ and $\alpha$  $*/$}
                \STATE $k_{excess} \gets \{k_i \cdot (1-\alpha \cdot (|G_1|+|G_2|) / n): i \in [p] \}$ {\em \small \color{blue} $/*$ Capacities not allocated to a specific group when $\alpha < 1$ $*/$} 
                \STATE $r_{unmatched} \gets$ candidates in $S$ sorted in order of decreasing utility
                \STATE $M \gets $ assignment filled with -1 (no institution)
                \FOR {candidate $a$ in $r_{unmatched}$}
                \STATE $M_a \gets \text{ top choice institution } n \text{ such that } k_{group}[n] > 0 \text{, or } k_{excess}[n] > 0$
                \IF {$k_{group}[n] > 0$}
                \STATE $k_{group}[n] \gets k_{group}[n] - 1$ {\em \small \color{blue} $/*$ $k_{group}$ is $k_{G_1}$ if the candidate is in $G_1$ and $k_{G_2}$ if the candidate is in $G_2$ $*/$} 
                \ELSE
                \STATE $k_{excess}[n] \gets k_{excess}[n] - 1$
                \ENDIF
                \STATE \textbf{until } $k_n = 0$ for all $n \in k$
                \ENDFOR 
                \STATE \textbf{return} $M$
            \end{algorithmic}
        \end{algorithm}

\section{Empirical Results}\label{sec:empirical}
        We empirically examine the robustness of \cref{thm:institutionWiseConstraints} under non-ideal conditions. We evaluate \cref{thm:institutionWiseConstraints} using data from India's centralized IIT admissions system, and evaluate a setting in which preferences are not drawn from a single distribution. 
The code and data can be found \href{https://github.com/sandrewxu/CentralizedSelectionwithPreferenceBias}{here}\footnote{ \url{https://github.com/sandrewxu/CentralizedSelectionwithPreferenceBias}}.
       
        \paragraph{Baselines and metrics.} \label{sec:empirical:definitions}
            We consider three algorithms in each simulation: the unconstrained algorithm  ($\Ast$), the algorithm subject to group-wise constraints ($\Agroup$), and the algorithm that satisfies institution-wide constraints ($\Ainst$). Pseudocodes
             are provided in \cref{sec:additional:pseudocode}.
            We consider two groups, a general and a 
 disadvantaged group, and measure preference-based fairness using $\sP^{(1)}$ and $\sP^{(3)}$ as defined in \cref{sec:model}.
            We generate preferences via the well-studied Mallows model, which mirrors many real-world scenarios \cite{Mallows}. Preferences are distributed according to a central ranking $\rho$ and the Kendall-Tau distance ($d_{\rm KT}$) between rankings, defined as the number of pairwise disagreements between two rankings. A dispersion parameter, $0 < \phi \leq 1$, defines the distribution such that for any ranking $\sigma$, $\Pr (\sigma) \propto \phi^{d_{\rm KT}(\sigma, \rho)}$. %

            \begin{figure*}[t!]
    \centering
    \includegraphics[width=\linewidth]{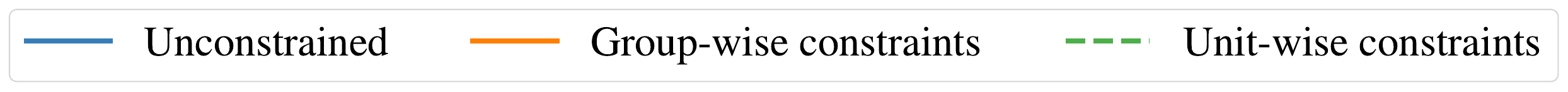}
    \par
    \subfigure[$\sP^{(1)}$ with gender]{
        \includegraphics[width=0.4\linewidth, trim={0cm 0cm 0cm 0cm},clip]{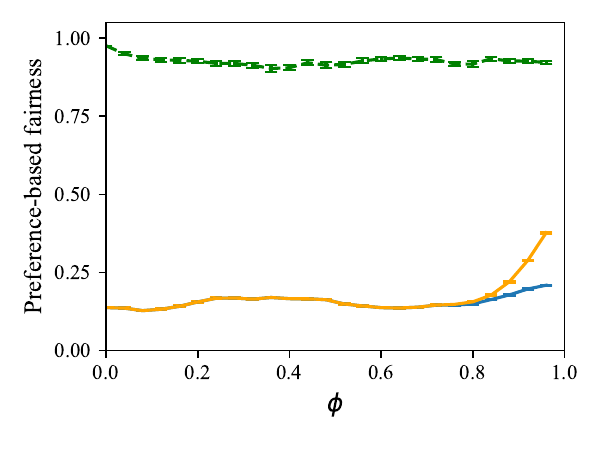}
    }
    \hspace{12mm}
    \subfigure[$\sP^{(1)}$ with birth-category]{
        \includegraphics[width=0.4\linewidth, trim={0cm 0cm 0cm 0cm},clip]{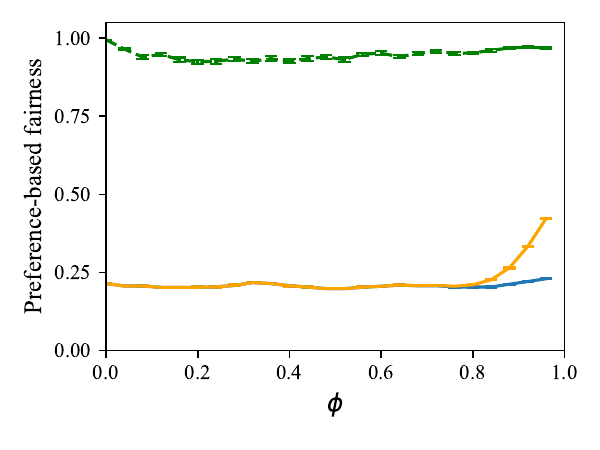}
    }
    \par
    \subfigure[$\sP^{(3)}$ with gender]{
        \includegraphics[width=0.4\linewidth, trim={0cm 0cm 0cm 0cm},clip]{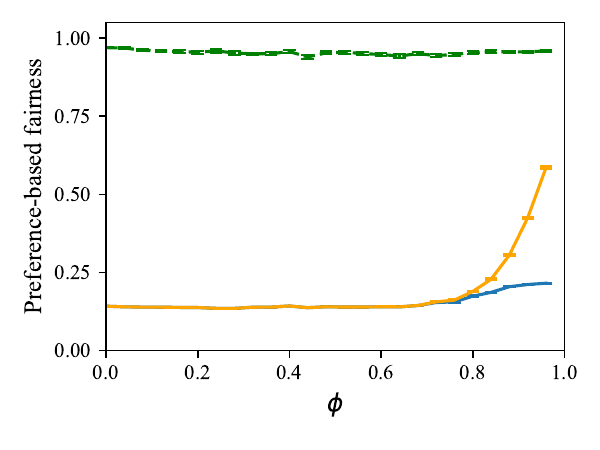}
    }
    \hspace{12mm}
    \subfigure[$\sP^{(3)}$ with birth-category]{
        \includegraphics[width=0.4\linewidth, trim={0cm 0cm 0cm 0cm},clip]{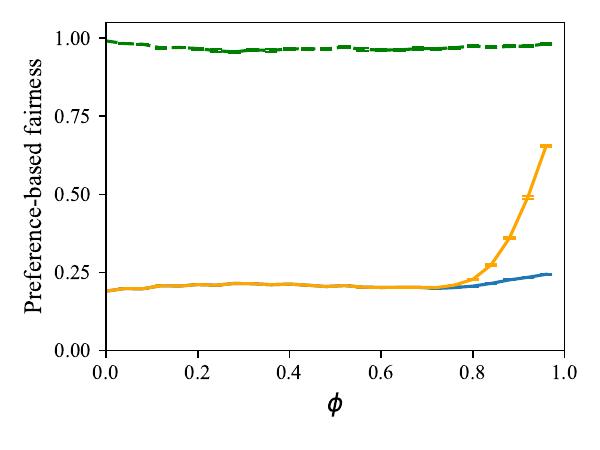}
    }
    \caption{
        Preference-based fairness as measured by $\sP^{(1)}$ and $\sP^{(3)}$ using either gender or birth-category as the protected attribute
        with data from the 2009 JEE test under centralized admission (see \cref{sec:simulation:real_world_data} for details).
        The $x$-axis denotes $\phi$, the dispersion parameter of the Mallows distribution.
        Error bars denote the standard error of the mean over 50 iterations.
        Institution-wise constraints achieve near-optimal preference-based fairness.
    }
    \label{fig:simulation:real_world_data}
    \vspace*{-0.2in}
\end{figure*}

\paragraph{Simulation with real-world data.}\label{sec:simulation:real_world_data}
            We use admissions data from India's centralized admissions system for the Indian Institutes of Technology (IITs) to evaluate institution-wise constraints in a real-world, non-ideal setting.

        \medskip
            \noindent\textit{{Data.}}
               Admission to IITs is determined by a centralized system that considers an applicant's score in the Joint Entrance Exam (JEE) as their utility and their preferences over different major-institute pairs \cite{baswana2019centralized}. 
                The data contains the test scores of 384,977 students from IIT-JEE 2009, self-reported gender and official birth category, opening and closing ranks, and capacities of major-institute pairs \cite{IITJEEBrochure}.
                Individual preferences were not reported, so we model preferences over institutions, defined as major-institute pairs, using a Mallows distribution with $\rho$ defined by increasing order of $C_{(M, I)},$ the closing All India Rank of candidates admitted into the major $M$ at institute $I$ in 2009.
                Further discussion about the dataset, the notion of latent utilities, and preference generation is in \cref{sec:additional_JEE}.

\medskip
        \noindent\textit{{Setup.}}
            We run two simulations, one considering gender and the other considering birth category as the protected attribute.
            We consider the $p=33$ institutions corresponding to the top major-institute pairs in $\rho$. The last index such that the $p$-th pair $(M,I)$ in $\rho$ has $C_{(M,I)}\leq 1000$ is $p=33$. 
            We fix $k_\ell$ to be the 2009 capacity of the $\ell$-th major-institute pair in $\rho$.
            For $\phi \in [0, 1]$, we report the average of $\sP^{(1)}$ and $\sP^{(3)}$ over 50 iterations, where we regenerate student preferences according to the Mallows distribution in each iteration. Implementation details appear in \cref{sec:additional_JEE}.

\medskip
    \noindent\textit{{Results and observations.}}
            Results appear in \cref{fig:simulation:real_world_data}. Our main observation is that institution-wise constraints achieve near-optimal preference-based fairness across all preference distributions, whereas group-wise constraints do not.
            For all $\phi$, $\sP^{(1)}(\Ainst) \geq 0.9$ and $\sP^{(3)}(\Ainst) \geq 0.95$.
            On the other hand, group-wise constraints and no constraints result in low preference-based fairness: $\sP^{(1)}(\Agroup) \leq 0.5$, $\sP^{(3)}(\Agroup) \leq 0.7$, and $\sP^{(3)}(\Ast) \leq 0.25$. 
            When $\phi \leq 0.75$, $\sP^{(3)}(\Agroup) \leq 0.25$.
            Under all conditions, institution-wide constraints achieve near-optimal preference-based fairness, higher than both other algorithms. Given $\phi$ $\leq 0.75$, which likely models the admissions system for IITs, there is no distinguishable difference between $\Ast$ and $\Agroup$ in terms of preference-based fairness.

        \paragraph{Simulation with distinct preference distributions.}
        \label{sec:simulation:synthetic_data}
            We also evaluate the efficacy of institution-wise constraints when groups have different preference distributions. This tests \cref{thm:institutionWiseConstraints} when the assumption that preferences are generated from a single distribution is not fulfilled. This situation can occur in many real-life scenarios: notably, Historically Black Colleges and Universities (HBCUs) in the United States are highly prestigious and a coveted choice for black students, but are rarely a top choice for students of other racial groups.

\begin{figure*}[t!]
    \centering
    \includegraphics[width=\linewidth]{figures_FINAL/legend.pdf}
    \par
    \vspace{-2mm}
    \subfigure[\small $\cD_{\rm Gauss}$ with $\beta{=}\frac{1}{4}$]{
        \includegraphics[width=0.4\linewidth, trim={0cm 0cm 0cm 0cm},clip]{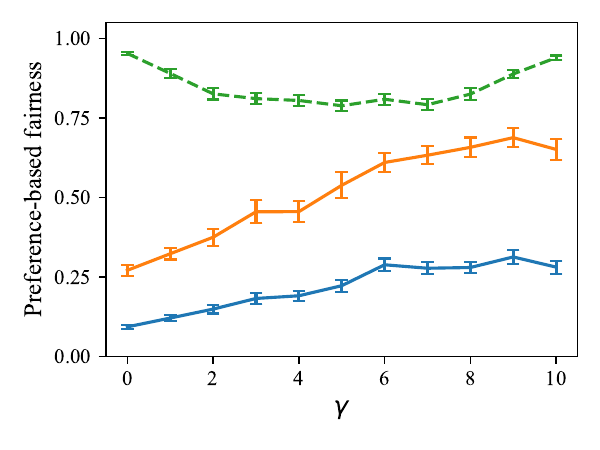}
    }
    \hspace{12mm}
    \subfigure[\small $\cD_{\rm Pareto}$ with $\beta{=}\frac{1}{4}$]{
        \includegraphics[width=0.4\linewidth, trim={0cm 0cm 0cm 0cm},clip]{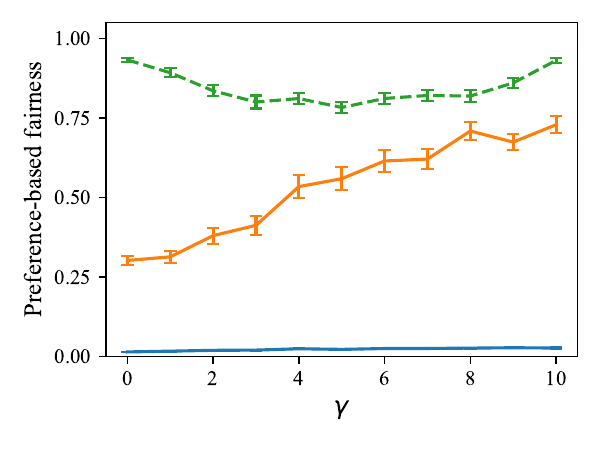}
    }
    \vspace{-0.1in}
    \caption{ 
        Preference-based fairness as measured by $\sP^{(1)}$, 
        with synthetic data where non-i.i.d. preferences are generated from Mallows distributions. 
        The $x$-axis denotes $\gamma$, the Kendall-Tau distance between the central rankings,  
        and the error bars denote the standard error of the mean over 50 iterations. 
        We observe that institution-wise constraints achieve higher preference-based fairness than group-wise and unconstrained settings.
    }
    \label{fig:simulation:synthetic_data}
    \vspace*{-0.2in}
\end{figure*}

        \medskip
        \noindent\textit{{Data (Latent utilities and preferences).}}
            We consider two different distributions for latent utility: $\cD_{\rm Gauss}$, the truncated Gaussian distribution bounded below by 0, and $\cD_{\rm Pareto}$, the Pareto distribution with shape parameter $\alpha = 3$.
            The Gaussian distribution is used because it often models the outcomes of standardized tests \cite{dorans2002recentering}. The fat-tailed Pareto distribution with shape parameter $2 \leq \alpha \leq 3$ models many real-life phenomena such as citations of a paper or success in creative professions \cite{clauset2009power}.
            When generating preferences, we use the Mallows distribution with $\phi = 0.25$ centered around $\rho_1, \rho_2$ for groups $G_1$ and $G_2$. For distance $0 \leq \gamma \leq \gamma_{\rm max} \coloneqq \binom{p}{2}$ 
            , $\rho_1$ and $\rho_2$ are randomly generated under the constraint $d_{\rm KT}(\rho_A, \rho_B) = \gamma$.

\medskip
        \noindent\textit{{Setup.}}
            We set $n = 1000, p = 5,$ and $k_{i} = 100$ for $i \in [p]$. 
            For $\cD \in \{\cD_{\rm Gauss}, \cD_{\rm Pareto}\}$ and $\beta \in \{\frac{1}{4}, \frac{1}{2}, \frac{3}{4}\}$, we vary $\gamma \in [0, \gamma_{max}]$ and calculate $\sP^{(1)}$ and $\sU$ over 50 iterations. 

           \medskip
            \noindent\textit{{Results and observations.}}
                \cref{fig:simulation:synthetic_data} shows the results using the $\sP^{(1)}$ metric with $\beta = \frac{1}{4}$. Results with other $\beta$ and a discussion of $\sU$ appear in \cref{sec:additional:empirical:additional_plot_synthetic}.
                We observe that implementing institution-wide constraints retains a high degree of preference-based fairness, while group-wise constraints do not.
                For all $\gamma$, $\sP^{(1)}(\Ainst) \geq 0.75$ while $\sP^{(1)}(\Agroup) \leq 0.7$. When $\gamma = 0$, central ranking is consistent across groups and $\sP^{(1)}(\Agroup) \leq 0.3$ and $\sP^{(1)}(\Ainst) \geq 0.9$. 
                As $\gamma$
                increases, $\sP(\Ainst)$ slightly decreases, then increases again. This is because an increase in $\gamma$ decreases the ability of $\Ainst$ to guarantee high $\sP$, but as $\gamma$ approaches $\gamma_{max}$, rankings tend to become more opposite and more people from each group match with their top choice institutions, resulting in high $\sP$. The preference-based fairness obtained through group-wise constraints and no constraints generally increases with $\gamma$. 
                Even when groups follow different preference distributions, institutional constraints outperform group-wise constraints, especially when $\gamma \leq 3$.
                
            To better model the complexities of socio-economic biases, we also explore other models of bias. 
            In \cref{sec:additional:empirical:betanoise}, we utilize the noisy $\beta$-bias model, where $\beta$ is sampled from a truncated Gaussian distribution with a standard deviation of $0.1$ within the range $[0,1]$. Our findings reveal that for $0.15 \leq \beta \leq 0.85$, the fairness achieved in this model is within $0.05$ of that predicted by the $\beta$- bias model. We also consider the implicit variance model in \cref{sec:additional:empirical:impvar}, which introduces Gaussian noise directly to the true utility of each candidate. The mean of the noise is centered at 0, and the variance is higher for the disadvantaged group. We vary the standard deviation of the Gaussian noise within the interval $[0,2]$ and observe that our algorithm maintains a preference-based fairness above $0.9$, while the introduced noise reduces the fairness of baseline algorithms.

        Recognizing that a rigorous implementation of \cref{thm:institutionWiseConstraints} may seem challenging in some contexts, we also consider an algorithm that relaxes the strict capacity reservation constraint of institutions in \cref{sec:additional:empirical:relaxedbounds}. We introduce the parameter $\alpha$, which ranges from $0$ to $1$, to allow a proportionate fraction of an institution's capacity to be reserved for various groups. By adjusting $\alpha$ to values less than 1, we offer institutions the flexibility to proportionally reserve a smaller portion of their capacity, making the algorithm adaptable to diverse policy requirements and institutional constraints. We find that this can provide an effective method for improving $\sP$, with a natural tradeoff between preference-based fairness and the strength of the constraints on the institutions.
        
        Finally, we empirically evaluate \cref{thm:specialcase} under different utility distributions, as well under the noisy $\beta$-bias model and the implicit variance model in \cref{sec:additional:empirical:specialcase} to see that under 
        these settings, the predictions of \cref{thm:specialcase} still generally hold.

\section{Proofs}\label{sec:proofs}

In this section, we formally prove~\Cref{thm:specialcase} (and its generalization~\Cref{thm:effectOfBias}) and~\Cref{thm:institutionWiseConstraints}. In Section~\ref{sec:lipschitz}, we prove crucial properties of a stable assignment in our setting. These properties are needed in proving~\Cref{thm:specialcase} and~\Cref{thm:institutionWiseConstraints}. In particular, we show that there is a unique stable assignment of candidates to institutions (\Cref{cl:stableassignment}). Further, this stable assignment has the following Lipschitz properties:  changing the preference list of one candidate or changing the capacity of an institution by one unit affects the assignment of at most $p+1$  candidates (\Cref{lem:lipschitzstable1}, \Cref{lem:lipschitzs2}). Using this  Lipschitz property, we show that given any instance and a subset $S$ of candidates, the number of selected candidates in $S$ and the number of candidates in $S$ receiving their top choice are concentrated around their respective means (\ref{lem:concentrationbounds}). Finally, we show that when we consider an instance $\cI$ (on a single group of candidates only), and {\em scale} it by a factor $a \geq 1$, i.e., scale down the number of candidates and the capacities by factor $a$, the expected fraction of selected candidates that get their top choice in the two instances remain close (\Cref{lem:gap}). 

In Section~\ref{sec:effectofbias}, we prove~\Cref{thm:effectOfBias} which is a generalization of~\Cref{thm:specialcase} to arbitrary values of $|G_1|, |G_2|$ and $K$. In Sections~\ref{sec:BoundingR}, \ref{sec:boundingsU} and \ref{sec:boundingsp}, we bound the metrics $\sR(\Ast), \sU(\Ast)$ and $\sP(\Ast)$ respectively. We show the properties of the algorithm $\Ainst$ by proving~\Cref{thm:institutionWiseConstraints} in Section~\ref{sec:proofof:thm:institutionWiseConstraints}. We show that $\sR(\Ainst), \sU(\Ainst)$ and $\sP(\Ainst)$ are near-optimal in Sections~\ref{sec:boundingRinst},~\ref{sec:boundingUinst} and~\ref{sec:boundingPinst} respectively. We now mention the concentration bounds that will be useful in the subsequent proofs. 
  
    \paragraph{Preliminary concentration bounds.}
    We state the concentration inequalities that will be useful in the analysis. 

        \begin{theorem}[\bf Hoeffding's Bound]
        \label{thm:Hoeffding}
        Let $X_1, \ldots, X_m$ be independent random variables such that $X_i \in [0,1]$ for each $i=1, \ldots, m$. Let $X$ denote $X_1 + \ldots + X_m$ and $\mu := \E[X]$. Then 
        $$ \Pr[X- \mu(X) \geq t ] \leq e^{-2t^2/m}, $$
        and 
        $$ \Pr[X- \mu(X) \leq t ] \leq e^{-2t^2/m}.$$
    \end{theorem}
    The above bound also holds when the random variables $X_1, \ldots, X_m$ are negatively correlated~\cite{dubhashibook}. For more general functions than the sum of the random variables, similar concentration bounds hold provided the function is Lipschitz. 

        \begin{theorem}[\cite{dubhashibook}]
        \label{thm:lipschitz}
        Let $T: \Omega^n \rightarrow \Re$ be a function of $n$ variables, each lying in a domain $\Omega$, such that for any two vectors $\mb{a}$ and $\mb{a'}$ differing in only one coordinate, 
        $$|T(\mb{a})-T(\mb{a'})| \leq c,$$
        for some parameter $c$.
    Let $X_1, \ldots, X_n$ be $n$ independent random variables, each taking values in $\Omega$. Then, 
    $$\Pr \left[|T(X_1, \ldots, X_n)-\Ex[T] | \geq t \right] \leq 2 e^{-2t^2/(c^2n)}.$$
    Here $\Ex[T]$ denotes the expectation of $T(Y_1, \ldots, Y_n)$, where the random variables $Y_j$ are independent, and each $Y_j \sim X_j$. 
        \end{theorem}
    
    \subsection{Properties of a stable assignment of candidates to institutions}\label{sec:lipschitz}
   We fix an input $\cI$ consisting of $n$ candidates and $p$ institutions, with institution $\ell$ having capacity $k_\ell$. The utilities of the candidates are drawn independently from a distribution $\cD$ over $[0,1]$ and the preferences over the institutions are drawn independently for a distribution $\cL$. This input is instantiated by a pair of tuples $(\husigma)$ sampled from $\cD \times \cL$, where $\mb{\hu}$ is the sequence of observed utilities of the $n$ candidates and $\mb{\sigma}$ is the sequence of their preferences. We begin by defining a stable assignment. 
    \begin{definition}
    \label{def:stable}
        Given $\husigma$, consider an assignment  $M$ of candidates to institutions. We say that $M$ is {\em stable} if for every pair of candidates $i$ and institution $j \neq M(i)$, the following condition must not hold: candidate $i$ prefers $j$ to its assignment $M(i)$, and institution $j$ prefers $i$ to all the candidates assigned to it by $M$.  
    \end{definition}

\noindent
    Observe that a stable assignment cannot leave a slot unassigned if $n \geq K$. Given arbitrary rankings of candidates by each institution and vice versa, stable assignments always exist; and they may not be unique. However, in the special setting where the institutions agree on a central ranking of the candidates, there 
    is a unique stable assignment.
    \begin{claim}
    \label{cl:stableassignment}
        Given $\husigma,$ there is a unique stable assignment of candidates to institutions. 
    \end{claim}
    \begin{proof}
        We prove this by induction on the number of candidates. As a base case, when there are no candidates, the statement is vacuously true. 

        Suppose the claim is true for $n-1$ candidates, and consider an instance $\cI :=(\husigma)$ with $n$ candidates. Let $M$ be a stable assignment. Let $i^\star$ be the candidate with the highest utility $\hu_{i^\star}$. We first claim   $M$ must assign $\hu_{i^\star}$ to its most preferred institution, say $j^\star$. Suppose not. Consider the candidate $i^\star$ and the institution $j^\star$. The candidate $i^\star$ prefers $j^\star$ to $M(i)$. Further, $j^\star$ prefers $i^\star$ to all the candidates assigned to it by $M$. But this contradicts the fact that $M$ is stable. 

        Now we define an instance $\cI'$ obtained by restricting $\cI$ to the candidates in $[n] \setminus \{i^\star\}$, and institution $j^\star$ having $k_{j^\star}-1$ capacity (all other institutions have the same capacity as in $\cI$). The assignment $M'$ given by the restriction of $M$ to $[n] \setminus \{i^\star\}$ is also a stable assignment for the instance $\cI'$. By induction hypothesis, $M'$ is uniquely determined by the instance $\cI'$. It follows that $M$ is also uniquely determined. 
    \end{proof}

\noindent
     We shall use $\Ast$ to denote the unique stable assignment algorithm. The proof of~\Cref{cl:stableassignment} shows that $\Ast$ can be specified as follows: sort the candidates in decreasing order of estimated utility. 
     Now consider the candidates in this order, and when considering candidate $i$, assign it to the highest preferred institution with a vacant slot. It is also worth noting that $\Ast$ does not depend on the actual values of the estimated utilities, but only on the relative ordering of the candidates based on these values. A formal description of this algorithm is given in~\Cref{alg:GSLatUtil}.

    \paragraph{Lipschitz properties of the stable assignment $\Ast$.}
    We now show that changing the preference of only one candidate in an instance changes the resulting assignment in a few coordinates. This result shall be useful in deriving concentration bounds for the number of candidates that get their most preferred choice of institution. We first give an example showing that such Lipschitz properties may not hold in general stable matching instances. 
    \begin{example}
        \label{ex:counterexample:stable:Lipschitz}
        Consider an instance consisting of two institutions $A$ and $B$, each having capacity $K/2$. There are $K$ candidates that are partitioned into two types: type-I candidates prefer institution $A$ to $B$ and type-II candidates prefer $B$ to $A$. There are $K/2$ candidates of both types. Now, institution $A$ prefers candidates of type-II to those of type-I (the preference ordering for candidates of the same type can be arbitrary), whereas institution $B$ prefers candidates of type-I to those of type-II. 

        Consider the following assignment $M$: assign all type-II candidates to $A$ and all type-I candidates to $B$. It is easy to check that this is a stable assignment. Now, suppose we remove a candidate of type-I. We argue that any stable assignment $M'$ must assign all type-I candidates to $A$.  Since there are only $K-1$ candidates, one slot shall remain unassigned. Suppose this slot is from institution $B$. Since there $K/2$ candidates of type-II, $M'$ must have assigned a candidate $j$ of type-II to $A$. But this causes an instability because $j$ prefers $B$ to $A$. Therefore, this vacant slot must be from institution $A$. Now, if $M'$ assigns a type-I candidate $i$ to $B$, it is again an instability because $i$ prefers $A$ and $A$ has a vacant slot. Thus, $M'$ must assign all type-I candidates to $A$ (and all type-II candidates to $B$). 

        We see that $M$ and $M'$ do not agree on the assignment of any candidate. Similar extreme examples can be shown to exist if we change the preference list of one candidate or change the capacity of an institution by one unit. 
    \end{example}

\noindent
   The above example does not apply to our setting because all institutions agree on a common ranking of the candidates. We show the following Lipschitz property of the assignment given by $\Ast$:
    \begin{lemma}
        \label{lem:lipschitzstable1}
        Consider instances $\cI$ and $\cI'$ on a set of $n$ candidates with the utility vector given by $\mb{\hu}$. Assume that the corresponding preference orderings $\mb{\sigma}$ and $\mb{\sigma'}$ differ in at most one coordinate (i.e., there is at most one candidate $i^\star$ such that $\sigma_{i^\star} \neq \sigma'_{i^\star}$). Let $M$ and $M'$ be the assignments given by $\Ast$ on the instances $\cI$ and $\cI'$ respectively. Then, $M$ and $M'$ differ in at most $p+1$ candidates.   Further, for any institution $\ell \in [p]$, the set of candidates assigned to $\ell$ in the two assignments, i.e., $M^{-1}(\ell)$ and $M'^{-1}(\ell)$, differ in at most 1 candidate (other than the candidate $i^\star$). 
    \end{lemma}
    \begin{proof}
        Order the candidates in decreasing order of $\hu_i$ values. Let $i^\star$ be the candidate for which $\sigma_{i^\star} \neq \sigma'_{i^\star}$. We run the algorithm $\Ast$ on these two instances. Recall that the algorithm $\Ast$ considers the candidates in decreasing order of utility, and when considering a candidate $i$, it assigns $i$ to the most preferred institution with a vacant slot.

        Let $M_i$ and $M_i'$ be the resulting assignments after $\Ast$ has considered the first $(i-1)$ candidates in the instances $\cI$ and $\cI'$ respectively. Clearly, $M_{i^\star} = M_{i^\star}'$. We claim that the following invariant holds for each $i > i^\star$: let $J_i$ be the set of institutions for which $M_i$ occupies more slots than $M_i'$. Define $J'_i$ similarly. Then $|J_i|, |J_i'| \leq 1$. Further, if $j \in J_i$, then $M_i$ occupies exactly one more slot on institution $j$ than that in $M_i'$ (and similarly for $J_i'$). 

        We prove the above statement by induction on $i$. When $i= i^\star+1$, this follows from the fact that $M_{i^\star} = M'_{i^\star}$ and candidate $i^\star$ may get assigned to two different institutions in the respective instances. Now suppose $i > i^\star+1$ and assume that the induction hypothesis holds for $M_i$ and $M'_i$. In the next step, the candidate $i$ needs to be assigned. Suppose $M$ and $M'$ assign candidate $i$ to $j$ and $j'$ respectively. If $j=j'$, then it is easy to see that the invariant conditions hold for $M_{i+1}$ and $M'_{i+1}$  as well.  Hence, assume $j \neq j'$. Assume w.l.o.g. that candidate $i$ prefers $j$ to $j'$. Then institution $j$ must be full in the assignment $M_i'$ (but it is not full in $M_i$). Therefore $j \in J_i'$ (in fact, by the induction hypothesis, $J_i' = \{j\}$). By induction hypothesis, $j$ had only one vacant slot in $M_i$, and this becomes full in $M_{i+1}$. Thus, institution $j$ is fully occupied in both $M_{i+1}$ and $M'_{i+1}$, and hence does not belong to $J_{i+1}'$.  Now $i$ is assigned to $j'$ by $M'_{i+1}$. If $j' \in J_i$, then $J_{i+1}$ becomes empty. Otherwise, $j'$ gets added (as the only institution) to $J_{i+1}'$. In either case, the invariant condition continues to hold after candidate $i$ has been assigned in both instances. This completes the proof of the induction hypothesis. 

        Let $M$ and $M'$ be the resulting assignments in the two instances. Let $D$ be the set of candidates $i$ (other than $i^\star$) such that $M$ and $M'$ assign $i$ to different institutions. For each $i \in D$,  we associate an institution $j(i)$ as follows: it is the institution among $M(i)$ and $M'(i)$ which is preferred by $i$.

        We claim that for any two distinct $i_1, i_2 \in D$, $j(i_1) \neq j(i_2)$. Indeed, suppose $\hu_{i_1} > \hu_{i_2}$. Then $\Ast$ considers $i_1$ before $i_2$. Assume w.l.o.g. that $M$ assigns $i_1$ to $j(i_1)$. Consider the assignments $M$ and $M'$ just before $i_1$ is considered. It must be the case that $j(i_1)$ is already full in $M'$. By the invariant condition, $j(i_1)$ would have only one vacant slot in the assignment $M$ at this time. This vacant slot would get occupied by $i_1$. Thus $j(i_1)$ gets full in both the assignments. Henceforth, no candidate can be assigned to $j(i_1)$. Thus, $j(i_2) \neq j(i_1)$. 

        The above claim shows that $|D| \leq p$ and this proves the first statement in the lemma. The second statement follows from the fact that $D$ contains at most one candidate that is assigned to a particular institution $\ell$ by either $M$ or $M'$.  
    \end{proof}
\noindent
    The above result shows that if two instances differ in the preference ordering of only one candidate, the resulting assignments can differ for at most $(p+1)$ candidates. Hence, the number of candidates that get their first choice in the two instances may also differ by at most $(p+1)$. A similar result holds if two instances differ slightly in the available capacity at each institution. 
    \begin{lemma}
        \label{lem:lipschitzs2}
        Consider instances $\cI$ and $\cI'$ on a set of $n$ candidates and $p$ institutions with the utility vector given by $\mb{\hu}$ and preference lists given by $\mb{\sigma}$. In both the instances, each institution $\ell \in [p]$ has capacity $k_\ell$, except for one particular institution $\ell^\star$ where $\cI$ has capacity $k_{\ells}$, but $\cI'$ has capacity $k_\ells+1$.  Let $M$ and $M'$ be the assignments given by $\Ast$ on the instances $\cI$ and $\cI'$ respectively. Then, $M$ and $M'$ differ in at most $p$ candidates.   Further, for any institution $\ell \in [p]$, the set of candidates assigned to $\ell$ in the two assignments, i.e., $M^{-1}(\ell)$ and $M'^{-1}(\ell)$, differ in at most 1 candidate. 
    \end{lemma}
    \begin{proof}
        We use~\Cref{lem:lipschitzstable1} to prove this result. Create a new instance $\cI_1$ which has $n+1$ candidates and $p+1$ institutions. The first candidate, call it candidate 0, is not present in the instances $\cI$ and $\cI'$; but candidate $i \in [n]$ in $\cI_1$ corresponds to candidate $i$ in $\cI$ (or $\cI'$). Similarly, label the first institution in $\cI_1$ as 0, and the remaining institutions are in 1-1 correspondence with those in $\cI$ or $\cI'$. 

        Candidate 0 has the highest utility and has institution $\ell^\star$ as its most preferred institution. For the rest of the candidates $i \in [n]$, their preference list in $\cI_1$ is given by  $\sigma_i$ along with institution 0 as the last choice. 
        It is easy to check that when we run $\Ast$ on $\cI_1$ to get an assignment $M_1$, candidate 0 is assigned to $\ell^\star$. The next $K := \sum_\ell k_\ell$ candidates are assigned to the same institutions as in $M$. Thus, among the first $K+1$ candidates, $M$ and $M_1$ agree 
        on all candidates except for candidate 0. 

        Now, create another instance $\cI_2$ which is the same as $\cI_1$ except that candidate 0 has institution 0 as the first choice. It is now easy to verify that the resulting assignment $M_2$ agrees with $M'$ on all the first $K+1$ candidates, except for candidate 0. Applying~\Cref{lem:lipschitzstable1} to $\cI_1$ and $\cI_2$, we see that for each institution $\ell \in [p]$, 
        $M_1^{-1}(\ell)$ and $M_2^{-1}(\ell)$ differ in at most 1 candidate. The lemma now follows from the fact that $M$ and $M_1$ agree on all candidates in $[K]$, and similarly for $M'$ and $M_2$. 
    \end{proof}
\noindent
Applying the above lemma repeatedly, we get the following result.
    \begin{corollary}
        \label{cor:lipschitz}
        Consider two instances $\cI$ and $\cI'$ on $n$ candidates and $p$ institutions with the same utility and preference lists of the candidates. For each $\ell \in [p]$, let $k_\ell$ and $k_\ell'$ denote the capacity of institution $\ell$ in the instances $\cI$ and $\cI'$ respectively. Let $D$ denote $\sum_\ell |k_\ell-k_\ell'|$. Then the stable assignments in the two instances differ on at most $p \cdot D$ candidates. 
    \end{corollary}

    \paragraph{Concentration bounds for assignments generated by the stable assignment.} We now apply the Lipschitz property proved in~\Cref{lem:lipschitzstable1} to show that several random variables associated with the assignment on a random instance are concentrated around their respective means.  
    \begin{lemma}
        \label{lem:concentrationbounds}
        Consider an instance $\cI$ on $n$ candidates and $p$ institutions with total capacity $K$,  where the utilities and the preference lists of candidates are drawn independently from distributions $\cD$ and $\cL$ respectively. Let $G$ be a subset of candidates belonging to the same group and $S$ denote the set of candidates selected by executing $\Ast$ on $\cI$. Then. 
        $$ \Pr\left[\bigg| |S \cap G| - \Ex[|S \cap G|] \bigg| \geq t\right] \leq 2 e^{\sfrac{-2t^2}{|G|}}.$$
        Condition on the observed utilities of the candidates.   
        Let $\pi(S)$ be the set of candidates that are assigned their first choice by $\Ast$ and $G$ be any subset of candidates. Then, 
        $$ \Pr \left[ \bigg| |\pi(S) \cap G| - \Ex[|\pi(S) \cap G|] \bigg| \geq t \right] \leq 2 e^{-\sfrac{2t^2}{(p^2K)}}, $$
        where the probability is over the random choice of preference lists of candidates. 
    \end{lemma}
    \begin{proof}
        For each candidate $i$, let $X_i$ be the indicator variable for the event that candidate $i$ is selected by $\Ast$. Note that $X_i$ is selected only if $\hu_i$ appears among the top-$K$ observed utilities. Now observe that if $i,i'$ are two candidates belonging to the same group, then $X_i$ and $X_{i'}$ are negatively correlated. Indeed, suppose both belong to group $G_j$ for some $j \in \{1,2\}$. Condition on the number of candidates selected by $\Ast$ from $G_j$: say this number is $K_j$. Let $E$ be the event that $K_j$ candidates from $G_j$ are selected. Then $i$ is selected if and only if $\hu_i$ occurs among the top $K_j$ candidates from  $G_j$. Thus, $\Pr[X_i=1|E] = \frac{1}{|G_j|}.$ Now, $\Pr[X_i=1|X_{i'} = 1, E] = \frac{1}{|G_j|-1} < \Pr[X_i=1|E]$. Removing the conditioning (since this holds for all choices of number of selected candidates from $G_j$), we see that $\Pr[X_i=1|X_{i'} = 1] < \Pr[X_i=1]$. Thus, the random variables $X_i, i \in G,$ are negatively correlated.  

        Let $S$ be the set of selected candidates. Then $|S \cap G| = \sum_{i \in G} X_i$. The first inequality in the lemma now follows from applying~\Cref{thm:Hoeffding} for negatively correlated random variables. 
        
        We now prove the second inequality in the lemma. We have conditioned on the observed utilities of the candidates. 
        Order the candidates in decreasing order of observed utility. Now, let $Z_i$ be the preference list chosen by candidate $i$. Observe that the random variables $Z_i$ are independent. Also observe that the assignment generated by $\Ast$ is a function of $Z_1, \ldots, Z_K$ (since we have fixed the utilities of all the candidates). Let $T(Z_1, \ldots, Z_K)$ be the number of candidates in $G$ that are assigned their top choice.~\Cref{lem:lipschitzstable1} shows that if $(Z_1', \ldots, Z_K')$ are another sequence of preference lists of the $K$ candidates respectively with $Z_i = Z_i'$ for all but one candidate, then $|T(Z_1, \ldots, Z_k) - T(Z_1', \ldots, Z_K')| \leq p$. The desired inequality in the lemma now follows by applying~\Cref{thm:lipschitz} to the function $T(\cdot)$. 
    \end{proof}
    
    \paragraph{Approximate additivity properties of the stable assignment.}
    Consider a setting given by $p$ institutions, $n$ candidates, and capacity $k_\ell$  for each institution $\ell \in [p]$. Assume that the true utilities are drawn from a density $\cL$ and the preference lists are drawn from a distribution $\cD$. Let $\cI$ denote this instance. For a  real $a \geq 1$, define the {\em scaled} instance $\cI_a$ as follows: the set of institutions is the same as in $\cI$, and the utilities of candidates and the preference lists are also drawn from $\cL$ and $\cD$ respectively. But there are $n/a$ candidates, and each institution $\ell$ has $k_\ell/a$ slots (assume that all of these quantities are integers). Given utilities $\mb{\hu}$ and preference lists $\mb{\sigma}$ of the $n/a$ candidates, let $M_a$ denote the assignment obtained by running $\Ast$ on $\husigma$. Let $\pi(M_a)$ denote the number of candidates that are assigned their first choice in the assignment $M_a$. We shall use $M$ to denote the assignment $M_1$ for the original instance $\cI$ (which is the same as $\cI_1$). 

    It turns out that for any fixed integer $a$, $\Ex[\pi(M_a)]$ may not be equal to $\Ex[\pi(M)/a].$  
    \begin{example}   \label{ex:expectation}
    Consider a setting with two institutions, say $a$ and $b$, with each institution having capacity 1. Further, assume that there are 2 candidates. The utility of each candidate is drawn from the uniform distribution on $[0,1]$. The preference list of each candidate is drawn uniformly at random from the two orderings $\{(a,b), (b,a)\}$. We claim that the expected number of candidates that are assigned their top choice is $3/2.$ Indeed, the first candidate (i.e., the one with the higher utility) always gets the first choice. The second candidate gets the first choice if and only if its preference list is different from that of the first candidate. This happens with probability $1/2$. Hence, the expectation of the number of candidates getting their top choice is 1.5. 

    However now consider the same instance as above, but now each institution has capacity 2 and there are 4 candidates. In other words, we are scaling the above instance by doubling the capacity and the number of candidates. A routine calculation shows that the expected number of candidates receiving their top choice is strictly larger than 3. Indeed, let the candidates (ordered by decreasing utility be 1,2,3,4). Now two possibilities can happen: 
    \begin{itemize}
        \item Candidates 1 and 2 have the same preference list: Assume w.l.o.g. that the two candidates have preference list $(a,b)$. Then both the candidates are assigned institution $a$. Now it is easy to check that the remaining two candidates get their top choice with probability $1/2$ each. Hence, the expected number of candidates receiving their first choice is 3. 
        \item Candidates 1 and 2 have distinct preference lists. In this case, candidate 3 shall always get its first choice, and candidate 4 receives its first choice with probability $1/2$. Hence, the desired expectation is 3.5. 
    \end{itemize}
    Since both the above cases can occur with equal probability, the expected number of candidates receiving their first choice is $3.25$, which is strictly larger than 3. 
    \end{example}
    
\noindent    
    On a positive note, we now show that the gap between the quantities $\Ex[\pi(M_a)]$ and $\Ex[\pi(M_{1})/a]$ can be bounded by $O(p\sqrt{K \log K}). $

    \lemgap*
    \begin{proof}
        We condition on the utilities and the preference list of the $n$ candidates in $\cI$. 
        Let $M$ be the assignment obtained by running $\Ast$ on this instance. Observe that $M$ is a deterministic assignment since we have fixed the utilities and the preference lists of the candidates.  

        Now, we generate a (random) instance of $\cI_a$ as follows: we choose $n/a$ out of these $n$ candidates in $\cI$ uniformly at random. Call this random subset $X$. We consider the corresponding instance $\cI_a$ where the candidates are given by $X$, and for each $i \in X$, its utility is given by $\hu_i$ and its preference list is given by $\sigma_i$. Let $M_X$ denote the resulting assignment (which is a random assignment because $X$ is a random subset of $[n]$). We now relate  $\Ex[\pi(M_X)]$ and $\pi(M)$ (recall that $M$ is a deterministic assignment).   
        \begin{claim}
            \label{cl:gap}
            $$\left| \Ex_X[\pi(M_X)]-\pi(M)/a \right| \leq O(p\sqrt{K \log K}).$$
        \end{claim}
        \begin{proof}
            For an institution $\ell \in [p]$, let $J_\ell$ be the set of candidates assigned to $\ell$ by $M$, i.e., $J_\ell := M^{-1}(\ell).$ Define $X_\ell := X \cap J_\ell$. 
            Define $\cE_\ell$ as the following event: 
            $$ \left| |X_\ell| - |J_\ell|/a \right| \leq 8 \cdot \sqrt{K \log K}.$$
            Since  $|J_\ell| = k_\ell,$
            we claim that $\cE_\ell$ happens with probability at least $1-1/K^4$. Indeed, for each $i \in J_\ell$, define the indicator variable $I_i$ which is 1 iff $i \in X_\ell$. The variables $I_i$ are negatively correlated and $X_\ell = \sum_{i \in J_\ell} I_i$. The desired concentration bound now follows from~\Cref{thm:Hoeffding}. 

            Let $X_0 := X \cap \pi(M)$ be the number of candidates in $X$ that are assigned their first choice by the assignment $M$. Define $\cE_0$ as the event that  
            $$ \left| |X_0| - \pi(M)/a \right| \leq 8 \cdot \sqrt{K \log K}.$$
            We can again argue as above that $\cE_0$ occurs with probability at least $1-1/K^4$. Thus, all the events $\cE_\ell$, $\ell=0, \ldots, p$, occur with probability at least $1-1/K^3$. 

            Assume that $\cE_\ell$ occurs for all $\ell \in \{0\} \cup [p]$, and now, we condition on the cardinalities of the sets $X_\ell$. Define $x_\ell := |X_\ell|$. Define a (random) instance $\cI'$  where the set of candidates is given by $X$ and each institution $\ell$ has $x_\ell$ seats. Then the assignment $M$ restricted to $X$ is a stable assignment for this instance. Let $\pi(\cI')$ be the number of candidates in $\cI'$ that are assigned their first preference by the stable assignment. Then, $\pi(\cI') = |X_0|.$ For an institution $\ell$, the fact that $\cE_\ell$ occurs implies that $|k_\ell/a - x_\ell| \leq 8 \sqrt{K \log K}. $

            Recall that the instance $\cI_a$ is defined on $X$ as the set of candidates and institution $\ell$ has $k_\ell/a$ seats. Comparing $\cI_a$ and $\cI'$, and applying~\Cref{cor:lipschitz}, we see that 
            $$|\pi(M_X)-\pi(\cI')| = O(p \sqrt{K \log K}).$$
             Since $\pi(\cI') = |X_0|$, using the definition of the event $\cE_0$, we see that 
             $$|\pi(M_X) - \pi(M)/a| = O(p \sqrt{K \log K}).$$
            Since the above inequality holds for any choice of $x_\ell, \ell \in \{0\} \cup [p]$, it also holds when we remove conditioning on the sizes of the sets $X_\ell$. 
              Now if any of the events $\cE_\ell, \ell \in \{0\} \cup [p],$ do not occur, we upper bound $\pi(M_x)$ by $K$.  Since the probability of such an event is at most $1/K^3$,  
            $$\left| \Ex_X[\pi(M_X)]-\pi(M)/a \right| \leq  O(p \sqrt{K \log K}) + \frac{K}{K^3} = O(p \sqrt{K \log K}).$$
            This proves the desired claim. 
        \end{proof}
\noindent
        Now consider an input to the instance $\cI$ with $n$ candidates where $(\husigma)$ is drawn from the distribution $\cD^n \times \cL^n.$ Now if we choose a random subset $X$ of $n/a$ candidates, the resulting distribution on $(\husigma)$ restricted to $X$ is drawn from $\cD^{n/a} \times \cL^{n/a}$. Thus, $\Ex_X(\pi(M_X))$ is same as $\Ex(\pi(I_a))$ and hence, the lemma follows from~\Cref{cl:gap}.
    \end{proof}

    \subsection{Proof of~\Cref{thm:specialcase} and extensions}
    \label{sec:effectofbias}
        In this section, we prove the following result that characterizes the fairness and utility metrics of the stable assignment algorithm $\Ast$ as a function of the parameter $\beta$.~\Cref{thm:specialcase} follows from the following result specialized to the setting where $n_1 = n_2 = K$.

        \begin{restatable}[]{theorem}{thmeffectOfBias}\label{thm:effectOfBias}
             Consider an instance where the utilities of the candidates are drawn from the uniform distribution on $[0,1]$  and the preference lists are drawn from an arbitrary distribution. Assume there are constants  $\eta_1, \eta_2$ such that $n_j \geq \eta_1 n$ for each $j \in \{1,2\}$, $K \geq \eta_2 n$.
             For 
                 any algorithm $\cA$ that, given $\mb{\hu}$ and $\mb{\sigma}$, outputs an assignment maximizing the estimated utility,
            
             \begin{align*}
        %
                 \textstyle \sR(\cA)& =  \textstyle   
                \textstyle \max\inbrace{
                    \textstyle \frac{\textstyle   
                        K-n_1(1-\beta)
                     }{\textstyle   
                         K \beta + n_2(1-\beta)
                    }, 0\textstyle 
                 }\textstyle   
                 \pm 
                \textstyle O \left( \frac{\sqrt{\log n}}{\eta_1 \eta_2 \sqrt{n}}\right)\\ 
                 \textstyle   \sU(\cA) & =  
                 \textstyle 
                     \frac{\textstyle\sum_{j=1}^2 f(\alpha_j, n_j)}{\textstyle f(K,n)} \pm \textstyle O\left(\frac{\sqrt{\log n}}{\eta_2\sqrt{n}} \right)
             \end{align*}
          \noindent  
             Here $f(x,y) := x - \frac{x(x+1)}{2(y+1)}$,
            $\alpha_2\coloneqq K-\alpha_1$,
             $\alpha_1 := K - \frac{n_2}{\beta n_1+ n_2}\cdot \max\inbrace{K - (1-\beta)n_1, 0}.$
        Assuming $\cA$ is the stable assignment algorithm $\Ast$,  $\sP(\cA)$ is at most 
         $$\textstyle   
                 \textstyle \max\inbrace{
                     \textstyle \frac{\textstyle   
                        K-n_1(1-\beta)
                     }{\textstyle   
                         K \beta + n_2(1-\beta)
                     }, 0\textstyle 
                }\textstyle   
                 +
                \textstyle O \left(  \frac{p\sqrt{\log n} }{\eta_1 \eta_2 \sqrt{n}} \right)$$
          \end{restatable}

        \noindent We divide the proof into three parts.
        In the three parts,  we bound $\sR$, $\sU$, and $\sP$ respectively.
        In the first two parts, the algorithm $\cA$ is any utility maximizing algorithm, but in the third part, $\cA$ is the same as $\Ast$ (observe that $\Ast$ is also utility maximizing and hence, the first two parts apply to $\Ast$ as well).

        \newcommand{\usigma}{\ensuremath{\mb{u},\mb{\sigma}}}

        \subsubsection{Step 1: Bounding \texorpdfstring{$\sR$}{R}}
        \label{sec:BoundingR}
            For any $\mb{u}=(u_1,u_2,\dots,u_n)$ and $\mb{\sigma}=(\sigma_1,\sigma_2,\dots,\sigma_n)$, let $M_{\mb{\hu},\mb{\sigma}}$ be the assignment output by $\cA$ given $\mb{\hu}$ and $\mb{\sigma}$.
            Recall that $\cA$ outputs an assignment maximizing the estimated utilities and hence, selects the top $K$ candidates according to their $\hu_i$ values.
            Let $S_{\husigma}$ be the subset of candidates that are assigned in $M_{\husigma},$ i.e., $S_{\husigma} = M_{\husigma}^{-1}([p])$.
            When $\husigma$ is clear from the context, we drop the subscripts in $M_{\husigma{}}$ and $S_{\husigma{}}$.

            The main step in bounding $\sR$ (and $\sU$) is to compute bounds on $\sabs{S\cap G_j}$, $j \in \{1,2\}$, that hold with high probability.
            These bounds are stated in the following lemma. Recall that $n_j$ denotes $|G_j|$. 
            \begin{lemma}\label{lem:step2:noConst}
                The following events happen with probability at least $1-O(1/n^2)$: for each $j \in \{1,2\}$, 
                \[
                    \sabs{S\cap G_j} \in \alpha_j \pm O(\sqrt{n \log n}),
                \]
                where $\alpha_2\coloneqq K-\alpha_1$ and $\alpha_1$ is defined as follows
                \begin{align*}
                    \alpha_1 :=
                    K - \frac{n_2}{\beta n_1+ n_2}\cdot \max\inbrace{K - (1-\beta)n_1, 0}.
            \yesnum\label{def:alphaA}
                \end{align*}
            \end{lemma}

        \begin{proof}[Proof of \cref{lem:step2:noConst}]
                With probability 1, $S$ is unique as the observed utilities $\inparen{\hu_i:i\in G_j}, j \in \{1,2\}$, are drawn i.i.d. from continuous distributions (namely, the uniform distributions on $[0,1]$ for $G_1$ and on $[0,\beta]$ for $G_2$).
            Since all candidates in $G_2$ have observed utility at most $\beta$, any candidate in $G_1$ with an observed utility larger than $\beta$ would get preference over the candidates in $G_2$. 
            Motivated by this, we give the following definition:
       \begin{definition}
           Given values $v_1, v_2 \in [0,1],$ and $j \in \{1,2\}$, let $G_j(v_1, v_2)$ denote the subset $\{i \in G_j: v_1 \leq \hu_j \leq v_2\}$ of $G_j$.
       \end{definition}
            \noindent
Now, we define several desirable events and show that each of them holds with high probability. 
            \begin{definition}
                [\bf Event $\cE_1$] 
                \label{def:event1}
                Define $\cE_1$ as the event that $$|G_1(\beta,1)| \in [(1-\beta)n_1 - 8 \sqrt{n \log n}, (1-\beta)n_1 + 8 \sqrt{n \log n}]. $$
            \end{definition}
            
    \begin{claim}
        \label{cl:eventE1}
        The event $\cE_1$ occurs with probability at least $1-2/n^4$. 
    \end{claim}
    \begin{proof}
        We apply~\Cref{thm:Hoeffding} with $X_i = 1$ iff candidate $i \in G_1$ belongs to $A_{\beta, 1}$. Note that $\hu_i = u_i$ for each $i \in G_1$. Clearly, $\E[X_i] = (1-\beta)$. Hence, it follows from ~\Cref{thm:Hoeffding} that 
        $$\Pr\left[\Big||G_1(\beta, 1)| - (1-\beta)n_1\Big| \geq 8 \sqrt{n \log n} \right] \leq 2/n^4.$$ This proves the desired result. 
    \end{proof}
    \noindent
             The above result suffices to estimate $|S \cap G_j|$ when $K \leq (1-\beta)n_1$:
\begin{claim}
\label{cl:upperboundcRcase1}
    Suppose $K \leq (1-\beta)n_1$. Then $|S \cap G_1| \geq K - 8\sqrt{n \log n}$ with probability at least $1-1/n^4$. Therefore, $|S \cap G_2| \leq 8\sqrt{n \log n}$ with probability at least $1-1/n^4$. 
\end{claim}
\begin{proof}
 Assume that the event $\cE_1$ occurs. Then 
 $$|G_1(\beta,1)| \geq n_1(1-\beta)-8 \sqrt{n \log n} \geq K-8 \sqrt{n \log n}.$$
 Since any candidate in $G_1(\beta,1)$ gets preference over any candidate in $G_2$, the claim follows from~\Cref{cl:eventE1}.
\end{proof}
\noindent
When $K \leq (1-\beta)n_1$, the quantity $\alpha_1 = K$ and $\alpha_2 = 0$. Thus,~\Cref{cl:boundcR} follows from~\Cref{cl:upperboundcRcase1} in this special case. Therefore, assume that $K > (1-\beta)n_1$ for the rest of the proof. Define $K' := K-(1-\beta)n_1$ and let $\Delta$ denote $\frac{K'}{n_1+n_2/\beta}$. The parameter $\Delta$ has been chosen such that the expected number of candidates whose observed utility lies in the range $[\beta-\Delta, 1]$ is exactly $K$. Indeed, $\E[|G_1(\beta-\Delta,1)|]$ is $(1-\beta+\Delta)n_1$ and $\E[|G_2(\beta-\Delta,1)|]$ is $\Delta n_2/\beta$ (note that the observed utility of a candidate in $G_2$ is uniformly distributed in the range $[0, \beta]$). Therefore, the expected number of candidates whose estimated utility lies in the range $[\beta-\Delta,1]$ is 
    \begin{align}
        \label{eq:Ksum}
        (1-\beta+\Delta) n_1 + \Delta n_2/\beta = (1-\beta)n_1 + \Delta \left( n_1 + n_2/\beta \right) =  K.
    \end{align}
    The rest of the proof shows that the above-mentioned events occur with high probability.

    \begin{definition}[\bf Events $\cE_2$ and $\cE_3$]
    \label{def:event23}
    Define $\cE_2$ as the event that $$|G_1(\beta-\Delta,1)| \in  n_1(1-\beta+\Delta) \pm 8 \sqrt{n \log n},$$ and $\cE_3$ as the event that $$|G_2(\beta-\Delta,1)| \in n_2 \Delta/\beta \pm 8 \sqrt{n \log n}.$$
    \end{definition}

    \noindent
    The following result follows in the same manner as~\Cref{cl:eventE1}.
    \begin{claim}
\label{cl:event2and3}
        Each of the events $\cE_2$ and $\cE_3$ occur with probability at least $1-2/n^4$. 
    \end{claim}
    \noindent
     Using union bound,~\Cref{cl:event2and3} shows that both the events $\cE_2, \cE_3$ occur with probability at least $1-4/n^4$. 
     \begin{claim}
         \label{cl:boundcR}
         Assume that the events $\cE_2, \cE_3$ occur. Then the number of candidates from $G_1$ and from $G_2$ that get selected by the assignment $M$ lie in the range $n_1(1-\beta+\Delta) \pm 8 \sqrt{n \log n}$ and $\Delta n_2/\beta \pm 8 \sqrt{n \log n}$ respectively.
     \end{claim}
     \begin{proof}
         Suppose, for the sake of contradiction, that the assignment $M$ selects more than $n_1(1-\beta+\Delta) + 8 \sqrt{n \log n}$ candidates from $G_1$. Since $K = n_1(1-\beta+\Delta) + \Delta n_2/\beta$ (using~\eqref{eq:Ksum}), the assignment $M$ selects less than $\Delta n_2/\beta-8 \sqrt{n \log n}$ candidates from $G_2$. Since the event $\cE_2$ occurs, the assignment $M$ selects a candidate $i \in G_1$ with $\hu_i < \Delta$. Since the event $\cE_3$ occurs, there must be a candidate $i \in G_2$ with $\hu_i > \Delta$ such that $i$ is not assigned by $M$. But this is a contradiction. Thus, $M$ selects at most $n_1(1-\beta+\Delta) + 8 \sqrt{n \log n}$ candidates from $G_1$.

         The argument that the assignment $M$ selects more than $n_1(1-\beta+\Delta) - 8 \sqrt{n \log n}$ candidates from $G_1$ follows similarly.
     \end{proof}
     \noindent
  It is easy to check that the quantities $\alpha_1$ and $\alpha_2$ as defined in the statement of~\Cref{lem:step2:noConst} are equal to  $n_1(1-\beta+\Delta)$ and $\Delta n_2/\beta$ respectively. Thus,~\Cref{lem:step2:noConst} follows from~\Cref{cl:event2and3} and~\Cref{cl:boundcR}. 
\end{proof}
\noindent
Using~\Cref{lem:step2:noConst}, we can easily bound $\cR(\cA)$. First, consider the case when $K \leq (1-\beta)n_1$. In this case the quantities $\alpha_1$ and $\alpha_2$ are $K$ and $0$ respectively.~\Cref{lem:step2:noConst} shows that with probability at least $1-1/n^2$ (recall that $q_j$ denotes the fraction of candidates from $G_j$ that get assigned by $M$), 
$$\frac{q_2}{q_1} \leq \frac{n_1}{n_2} \cdot \frac{8 \sqrt{n \log n}}{K - 8 \sqrt{n \log n}} \leq \frac{8 \sqrt{n \log n}}{\eta_1 \eta_2 n/2} = O \left( \frac{\sqrt{\log n}}{\eta_1 \eta_2 \sqrt{n}} \right),  $$
where the second last inequality follows from the fact that $K \geq \eta_2 n \geq 16 \sqrt{n \log n}$ (for large enough $n$) and $n_2 \geq \eta_1 n \geq \eta_1 n_1$. Thus, with probability at least $1-1/n^2$, 
$\frac{\min(q_1, q_2)}{\max(q_1, q_2)} = O \left( \frac{\sqrt{\log n}}{\eta_1 \eta_2 \sqrt{n}} \right).$
Since the l.h.s. above can never exceed $1$, we see that 
$$ \cR(\cA) = \E \left[ \frac{\min(q_1, q_2)}{\max(q_1, q_2)}\right] \leq \left(1 - 1/n^2 \right) O \left( \frac{\sqrt{\log n}}{\eta_1 \eta_2 \sqrt{n}} \right) + \frac{1}{n^2} = O \left( \frac{\sqrt{\log n}}{\eta_1 \eta_2 \sqrt{n}} \right). $$
This proves the desired bound on $\cR(\cA)$ when $K \leq (1-\beta)n_1$. 

Now consider the case when $K > (1-\beta)n_1$.~\Cref{cl:boundcR} implies that, with probability at least $1-1/n^2$, 
\begin{align}
\label{eq:repupperbound}
\frac{q_2}{q_1} \leq \frac{n_1}{n_2} \cdot \frac{ \Delta n_2/\beta+ 8 \sqrt{n \log n}}{n_1(1-\beta+\Delta) - 8 \sqrt{n \log n}} = \frac{\Delta/\beta}{(1-\beta+\Delta)} \cdot \left( 1 - \frac{8 \sqrt{n \log n}}{n_1(1-\beta+\Delta)} \right)^{-1} \cdot \left( 1+ \frac{8 \sqrt{n \log n}}{\Delta n_2/\beta} \right). \end{align}
Recall that $\Delta = \frac{K'}{n_1 + n_2/\beta},$ where $K' = K - (1-\beta)n_1$. Using this, we see that 
\begin{align}
    \label{eq:Deltaeta2}
    1-\beta + \Delta \geq \frac{K}{n_1 + n_2 } = \frac{K}{n} \geq \eta_2.
\end{align}

\noindent
Combining the above observations with the fact that $n_1, n_2 \geq \eta_1 n$, the r.h.s. of~\eqref{eq:repupperbound}  can be expressed as 
\begin{align}
    \label{eq:q_1q_2bound1}
    \frac{q_1}{q_2} \leq \frac{\Delta/\beta}{1-\beta+\Delta} + O \left(\frac{\sqrt{\log n}}{\eta_1 \eta_2 \sqrt{n}} \right). 
\end{align}
A similar argument shows that with probability at least $1-1/n^2$, 
\begin{align}
    \label{eq:q_1q_2bound2}
    \frac{q_1}{q_2} \geq \frac{\Delta/\beta}{1-\beta+\Delta} - O \left(\frac{\sqrt{\log n}}{\eta_1 \eta_2 \sqrt{n}} \right). 
\end{align}
 A direct calculation shows that $1-\beta + \Delta \geq \Delta/\beta$. Indeed,  this is the same as showing $\Delta \leq \beta$, which holds when $n \geq K$. Thus, assuming~\eqref{eq:q_1q_2bound1} holds, 
$\frac{\min(q_1,q_2)}{\max(q_1, q_2)} = \frac{q_1}{q_2} - \left(\frac{\sqrt{\log n}}{\eta_1 \sqrt{n}} \right). $
Therefore, $$ \cR(\cA) = \E \left[ \frac{\min(q_1, q_2)}{\max(q_1, q_2)}\right] \stackrel{\eqref{eq:q_1q_2bound1},\eqref{eq:q_1q_2bound2}}{=}  \frac{\Delta/\beta}{1-\beta+\Delta} \pm O \left(\frac{\sqrt{\log n}}{\eta_1 \sqrt{n}} \right). $$
Using the definition of $\Delta$, one can check that  
$$ \frac{\Delta/\beta}{1-\beta + \Delta} = \frac{K-(1-\beta)n_1}{\beta K + (1-\beta)n_2}. $$
This proves the desired bound on $\cR(\cA)$ as stated in~\Cref{thm:effectOfBias}.

     \subsubsection{Step 2: Bounding \texorpdfstring{$\sU$}{U}}
     \label{sec:boundingsU}
      We now bound $\sU(\cA)$; we use a similar strategy as the one used for bounding $\cR(\cA)$. We first derive concentration bounds on the sum of random values drawn from a certain range.     We begin with a standard result on $k$-th order statistic for $n$ independent draws from the uniform distribution on $[0,1]$:
      \begin{fact}[\cite{ahsanullah2013introduction}]\label{fact:os}
              Let $X_1, \ldots, X_n$ be $n$ independent random variables, each drawn from the uniform distribution on $[0,1]$. Let $Y_k$ denote the $k$-th largest value in $\{X_1, \ldots, X_n\}$. 
              For all $k\in [n]$, $\E[Y_k] = \frac{n+1-k}{n+1}.$
            \end{fact}

\noindent   
    The above result implies the following concentration bound for the sum of top-$k$ among $n$ values drawn uniformly at random from $[0,1]$:
    \begin{claim}
        \label{cl:orderstatisticbound}
        Let $X_1, \ldots, X_n$ be $n$ i.i.d. random variables, each drawn uniformly at random from $[0,1]$. Given a parameter $k \in [n]$, define the function $T_k:[0,1]^n \rightarrow \Re$ as the sum of the 
  top $k$ values among $x_1, \ldots, x_n$. Then 
        $$\Pr \left[ |T_k(X_1, \ldots, X_n) - f(k,n)| \geq t \right] \leq 2 e^{-2t^2/n}.$$
        Here $f(k,n) := k - \frac{k(k+1)}{2(n+1)}$. 
    \end{claim}
    \begin{proof}
        \Cref{fact:os} implies that $\E[T_k] = \sum_{i=1}^k \frac{n+1-i}{n+1} = k - \frac{k(k+1)}{2(n+1)}=f(k,n).$
        It is easy to check that the function $T_k$ satisfies the conditions in~\Cref{thm:lipschitz}, and hence the desired result follows from~\Cref{thm:lipschitz}. 
    \end{proof}
    \noindent
     Now we define several events: 
     \begin{definition}
        Define $\cE_4$ as the following event: for each $j \in \{1,2\}$, the top-$\alpha_j$ utilities among candidates in the group $G_j$ lies in the range $f(\alpha_j, n_j) \pm 8 \sqrt{n \log n}$. Define $\cE_5$ as the event that the top-$K$ utilities among all the candidates $[n]$ lies in the range 
        $f(K,n) \pm 8 \sqrt{n \log n}$. Finally, define $\cE_6$ as the event that for each $j \in \{1,2\}$, $|S \cap G_j|$ lies in the range $\alpha_j \pm 8 \sqrt{n \log n}$. 
     \end{definition}
     Recall that $S = M^{-1}([p])$ is the set of candidates that are assigned an institution by $M$.~\Cref{cl:orderstatisticbound} shows that each of the events $\cE_4$ and $\cE_5$ happen with probability at least $1 - 2/n^4$, and~\Cref{lem:step2:noConst} shows that $\cE_6$ happens with probability at least $1 - O(1/n^2)$. Applying union bound, all these three events happen with probability at least $1 - O(1/n^2)$. 
     \begin{claim}
         \label{cl:utilitybound}
         Assume that the events $\cE_4, \cE_5, \cE_6$ happen. Then 
         $$ \frac{\mathsf{U}(M_{\mb{\hat{u}},\mb{\sigma}})}{\mathsf{U}^\star(\mb{u})} = \frac{\sum_{j=1}^2 f(\alpha_j, n_j)}{f(K,n)} \pm O\left(\frac{\sqrt{\log n}}{\eta_2\sqrt{n}} \right).$$
     \end{claim}
     \begin{proof}
         Assume that the events $\cE_4, \cE_5, \cE_6$ occur. Since $M$ assigns at most $\alpha_j + 8 \sqrt{n \log n}$ candidates from group $G_j$ (event $\cE_6$), and the top-$\alpha_j$ candidates from $G_j$ have total utility at most $f(\alpha_j, n_j) + 8 \sqrt{n \log n}$,  the total utility of the candidates  in $S \cap G_j$ is at most 
         $$f(\alpha_j, n_j) + 16 \sqrt{n \log n},$$
         because the utility of any particular candidate can be at most 1. Thus, the total utility of the candidates selected by $M$ is at most 
         $$\sum_{j=1}^2 f(\alpha_j, n_j) + 32 \sqrt{n \log n}.$$
         Event $\cE_5$ implies that $\mathsf{U}^\star(\mb{u})\geq f(K,n)-8 \sqrt{n \log n}.$ Thus, 
         $$ \frac{\mathsf{U}(M_{\mb{\hat{u}},\mb{\sigma}})}{\mathsf{U}^\star(\mb{u})} \leq \frac{\sum_{j=1}^2 f(\alpha_j, n_j) + 32 \sqrt{n \log n}}{f(K,n)-8 \sqrt{n \log n}}.$$
         Now, $f(K,n) = K-\frac{K(K+1)}{2(n+1)} \geq K/2 \geq \eta_2 n/2$. Similarly $f(\alpha_1,n_1) + f(\alpha_2,n_2) \geq \alpha_1/2 + \alpha_2/2 = K/2 \geq \eta_2 n/2$. Therefore, the r.h.s. above is at most   
         $$\frac{\sum_{j=1}^2 f(\alpha_j, n_j)}{f(K,n)} \left(1+ O\left(\frac{\sqrt{\log n}}{\eta_2 \sqrt{n}} \right) \right) \leq \frac{\sum_{j=1}^2 f(\alpha_j, n_j)}{f(K,n)} + O\left(\frac{\sqrt{\log n}}{\eta_2 \sqrt{n}} \right),$$
         where the last inequality follows from the fact that $\sum_{j=1}^2 f(\alpha_j, n_j) \leq f(K,n)$. 
         A similar argument shows that 
         $$ \frac{\mathsf{U}(M_{\mb{\hat{u}},\mb{\sigma}})}{\mathsf{U}^\star(\mb{u})} \geq \frac{\sum_{j=1}^2 f(\alpha_j, n_j)}{f(K,n)} -O\left(\frac{\sqrt{\log n}}{\eta_2 \sqrt{n}} \right).$$
         This proves the desired result. 
     \end{proof}
     \noindent
     Since the ratio $ \frac{\mathsf{U}(M_{\mb{\hat{u}},\mb{\sigma}})}{\mathsf{U}^\star(\mb{u})}$ always lies in the range $[0,1]$, and the event $\cE_4 \cap \cE_5 \cap \cE_6$ happens with probability at least $1-O(1/n^2)$, it follows that 
     $$\sU(\cA) = \Ex \left[\frac{\mathsf{U}(M_{\mb{\hat{u}},\mb{\sigma}})}{\mathsf{U}^\star(mb{u}} \right] = \frac{\sum_{j=1}^2 f(\alpha_j, n_j)}{f(K,n)} \pm O\left(\frac{\sqrt{\log n}}{\eta_2 \sqrt{n}} \right) + O(1/n^2) = \frac{\sum_{j=1}^2 f(\alpha_j, n_j)}{f(K,n)} \pm O\left(\frac{\sqrt{\log n}}{\eta_2 \sqrt{n}} \right).$$
     This proves the desired bound on $\sU(\cA)$ as stated in~\Cref{thm:effectOfBias}.

         \subsubsection{Step 3: Bounding \texorpdfstring{$\sP$}{P}}
         \label{sec:boundingsp}
           In this section, we bound $\sP(\cA)$. We use the notation developed in~\Cref{sec:BoundingR}. 
           Recall that $G_j(v_1, v_2)$ denotes the set of  candidates $i \in G_j$ for which $\hu_{i} \in [v_1, v_2]$. 
           We assume that $K \leq (1-\beta)n_1$ (the other case is only easier). We also assume that the events $\cE_1, \cE_2, \cE_3$ occur (as defined in~\Cref{def:event1} and~\Cref{def:event23}). More specifically, we condition on the sizes of the sets $G_1(\beta, 1), $ $ G_1(\beta-\Delta, \beta)$ and
           $ G_2(\beta-\Delta, \beta)$.  Call these values $g_1, g_1', g_2'$ respectively. Note that the events $\cE_1, \cE_2, \cE_3$ require
           that
           \begin{align}
               \label{eq:gbound}
               g_1 \in (1-\beta)n_1 \pm 8 \sqrt{n \log n}, \ g_1+g_1' \in n_1(1-\beta+\Delta) \pm 8 \sqrt{n \log n}, \ g_2' \in n_2 \Delta/\beta \pm 8 \sqrt{n \log n}. 
           \end{align}
          
\noindent
    One manner in which we can condition on the above events is as follows:
      \begin{itemize}
          \item[(i)] Select $g_1$  values uniformly at random in the range $[\beta,1]$ and assign these (observed) utility values to the first $g_1$ candidates. 
          \item[(ii)] Select $g_1'+g_2'$ values uniformly at random in the range $[\beta-\Delta,\beta]$ and assign these estimated utility values to the next $g_1'+g_2'$ candidates. 
          \item[(iii)] Select $n-(g_1+g_1'+g_2')$ values uniformly at random in the range $[0,\beta-\Delta]$ and assign these estimated utility values to the remaining candidates. 
          \item[(iv)] Assign the first $g_1$ candidates (who were assigned utility values in step~(i)) to $G_1$. Consider the next $g_1'+g_2'$ candidates who were assigned values in step~(ii) above. Select $g_1'$ candidates (without replacement) uniformly at random from this subset of candidates and assign them to $G_1$, and the remaining $g_2'$ candidates to $G_2$. Finally, randomly partition the remaining $n-(g_1+g_1'+g_2')$ candidates (used in step~(iii) above) between $G_1$ and $G_2$ in a similar manner. 
      \end{itemize}  
      Finally, we assign each candidate $i$ a random ordering $\sigma_i$ chosen independently from the distribution $\cal D$. Now we run $\Ast$ on these candidates. Note that $\Ast$ does not depend on the group to which a candidate belongs. Hence, $\Ast$ only depends on the first three steps mentioned above and does not depend on step~(iv). 

For the sake of brevity of notation, let $A$ denote the first $g_1$ candidates in $[n]$ (i.e., the candidates that are assigned values in step~(i) above) and $B$ denote the next $g_1'+g_2'$ candidates (i.e, the ones assigned values in step~(ii) above). Let $\pi(A)$ and $\pi(B)$ denote the set of candidates in $A$ and $B$ respectively that are assigned their most preferred institutions by the algorithm $\Ast$. Using~\Cref{lem:concentrationbounds}, we see that 

        \begin{align}
        \label{eq:bounded-diffA}
        \Pr[|\pi(A)- \E[\pi(A)]| \geq p\sqrt{8n \log n}] \leq 1/n^2,
        \end{align}
        and 
        
        \begin{align}
        \label{eq:bounded-diffB}
            \Pr[|\pi(B)- \E[\pi(B)]| \geq p\sqrt{8n \log n}] \leq 2/n^2.
        \end{align}
      \noindent 
    We now relate $\E[\pi(A)]$ and $\E[\pi(B)]$. 
    \begin{claim}
        \label{cl:expectationbounds}
        $ \E[\pi(A)]/g_1 \geq  \E[\pi(B)]/(g_1'+g_2').$
    \end{claim}
    \begin{proof}
        For a candidate $i$, let $D_i$ be the event that it is assigned the top preferred institution. We now argue that for any $i \in [n]$, $\Pr[D_i] \geq \Pr[D_{i+1}]$. We condition on the assignment of institutions to the first $(i-1)$ candidates. After this assignment, let $S$ be the set of institutions that have a vacant slot. Now let $\sigma \sim \cD$ be the preference list for candidate $i$ and $j^\star$ be the first institution in this list. Then $\Pr[D_i]$ is equal to the probability that $j^\star \in S$. Now after assigning an institution to $i$, the set of institutions with vacant slots will be a subset of $S$. Therefore, $\Pr[D_{i+1}]$is at most the probability that $j^\star \in S$. 
        Thus, we see that $\Pr[D_i] \geq \Pr[D_{i+1}]$ for all $i$. Now, 
        \begin{align*}
            \E[\pi(A)] = \sum_{i=1}^{g_1} \Pr[D_i] \geq g_1 \Pr[D_{g_1}],
        \end{align*}
        and 
        $$ \E[\pi(B)] = \sum_{i=g_1+1}^{g_1+g_1'+g_2'} \Pr[D_i] \leq (g_1'+ g_2') \Pr[D_{g_1}].$$
        The desired result now follows from the above two inequalities. 
    \end{proof}
\noindent
    We can now estimate the number of candidates from each of the groups $G_j$ who have been assigned to their most preferred institution.  
    \begin{claim}
        \label{cl:AB}
        With probability at least $1 - O(1/n^2)$, the number of candidates from $G_1$ and $G_2$ that are assigned their most preferred institution  choices are $\E[\pi(A)] +  \frac{g_1'}{g_1' + g_2'} \cdot \E[\pi(B)] \pm O(p\sqrt{n \log n})$  and $ \frac{g_2'}{g_1' + g_2'} \cdot \E[\pi(B)] \pm O(p\sqrt{n \log n})$ respectively. 
    \end{claim}
    \begin{proof}
    All the candidates in $A$ are assigned to $G_1$. By~\eqref{eq:bounded-diffA}, the number of candidates in $A$ that are assigned their top choices lies in the range 
        \begin{align}
        \label{eq:topA}
        \Ex[\pi(A)] \pm p\sqrt{8n \log n}
        \end{align}
        with probability at least $1-1/n^2$. 

        Now, we condition on the subset $\pi(B)$ of candidates in $B$ that get their first choice -- let this be $B'$. 
        Recall that we randomly select $g_1'$ candidates from $B$ and assign them to $G_1$. For a candidate $i \in B'$, let $X_i$ be the probability that $i$ is assigned to $G_1$ and $i \in B'$. It is easy to see that 
        $\E[X_i] = \frac{g_1'}{g_1'+g_2'}$. The random variables $X_i$ are not independent but are negatively correlated. Therefore, applying Hoeffding's bound to $\sum_i X_i$, we see that with probability at least $1-1/n^2$, the number of candidates in $B' \cap G_1$ is $$\frac{g_1'|B'|}{g_1'+g_2'} \pm p\sqrt{8 n \log n}.$$ Further,~\eqref{eq:bounded-diffB} shows that $|B'| \in \E[\pi(B)] \pm p\sqrt{8 n \log n}$ with probability at least $1-2/n^2$. Combining the above two observations, we see that with probability at least $1-3/n^2$, the number of candidates from $G_1 \cap B$ that are assigned their top preferences lies in the range 
        \begin{align}
        \label{eq:topB}
        \frac{g_1' \E[\pi(B)]}{g_1'+g_2'} \pm p\sqrt{16 n \log n}. \end{align}
        Combining~\eqref{eq:topA} and~\eqref{eq:topB}, we see that the number of candidates in $G_1 \cap (A \cup B)$ that are assigned their top choices lies in the range 
        $$ \Ex[\pi(A)] + \frac{g_1' \E[\pi(B)]}{g_1'+g_2'} \pm O(p\sqrt{ n \log n}).$$
        with probability at least $1-4/n^2$.~\Cref{cl:boundcR} shows that all but $O(\sqrt{n \log n})$ selected candidates lie in the set $A \cup B$. This proves the desired bound on the number of candidates in $G_1$ that are assigned their first choice. The corresponding bound for $G_2$ follows similarly. 
    \end{proof}

\noindent
    Recall that $\pi_j$ denotes the {\em fraction} of candidates in $G_j$ that are assigned their top choice. ~\Cref{cl:AB} and~\Cref{cl:expectationbounds} show that 
    \begin{align}
        \label{eq:pi1pi2bound}
    \frac{\pi_2}{\pi_1} \leq \frac{n_1}{n_2} \cdot \frac{g_2' \pm O(p\sqrt{n \log n})}{g_1 + g_1' \pm O(p\sqrt{n \log n})} = \frac{n_1}{n_2} \frac{g_2'}{g_1 + g_1'} + O \left( \frac{p\sqrt{\log n}}{\eta_1 \eta_2 \sqrt{ n}} \right),
    \end{align}
    where the last inequality follows from the fact that $n_1/n_2 \geq 1/\eta_1$ and $g_1+g_1'$ is $\Omega(\eta_2 n) $ (using~\eqref{eq:Deltaeta2} and~\eqref{eq:gbound}). 
    Using~\eqref{eq:gbound} again, we get 
$$\frac{\pi_2}{\pi_1} \leq \frac{\Delta/\beta}{(1-\beta)+\Delta} + O \left( \frac{p\sqrt{\log n}}{\eta_1 \eta_2 \sqrt{ n}} \right) =  \frac{K-n_1(1-\beta)}{\beta K + (1-\beta) n_2} +   O \left( \frac{p\sqrt{\log n}}{\eta_1 \eta_2 \sqrt{ n}} \right),$$
where the last equality follows from the definition of $\Delta$. The above expression does not involve $g_1, g_1', g_2'$ and hence, holds even after removing the conditioning on these parameters. In other words, the above bound holds if the events $\cE_1, \cE_2, \cE_3$ occur.  

     Since $\cE_1, \cE_2, \cE_3$ occur with probability at least $1 - O(1/n^4)$, and the ratio $\frac{\min(\pi_1,\pi_2)}{\max(\pi_1, \pi_2)} \leq 1,$ we see that 
    $$\E \left[ \frac{\min(\pi_1,\pi_2)}{\max(\pi_1, \pi_2)} \right] \leq \frac{K-n_1(1-\beta)}{\beta K + (1-\beta) n_2} +   O \left( \frac{p\sqrt{\log n}}{\eta_1 \eta_2 \sqrt{ n}} \right) + O(1/n^4) = \frac{K-n_1(1-\beta)}{\beta K + (1-\beta) n_2} +   O \left( \frac{p\sqrt{\log n}}{\eta_1 \eta_2 \sqrt{ n}} \right) .$$ This proves the desired bound on $\sP(\cA)$. Thus we completed the proof of~\Cref{thm:effectOfBias}. The special setting of~\Cref{thm:specialcase} now follows immediately. We restate this result below:
    \effectofbiasspecial*
    \begin{proof}
        The use notation in the statement of~\Cref{thm:effectOfBias}. The parameters $\eta_1$ and $\eta_2$ can be set to 1. Since $n_1 = n_2 = K$, $K-n_1(1-\beta) = n_1 \beta$ and $K\beta + n_2(1-\beta) = n_1$. Therefore, $\sR(\Ast) = \beta \pm O \left(\frac{\sqrt{\log n}}{\sqrt{n}} \right),$
        and 
        $\sP(\Ast) \leq \beta + O \left(p\frac{\sqrt{\log n}}{\sqrt{n}} \right).$
        We now compute $\sU(\Ast).$ Observe that $\alpha_1 = K - \frac{K}{\beta+1} = \frac{\beta \cdot K}{\beta+1}$ and $\alpha_2  = \frac{K}{\beta+1}.$ Using these, we see that (up to an error of $O(1/n)$)
        $$f(\alpha_1, n_1) = \frac{2\beta+\beta^2}{2(\beta+1)^2}, \ f(\alpha_2, n_2) = \frac{2\beta+1}{2(\beta+1)^2}, \ f(K,n) = \frac{3}{8}.$$
        Substituting these values in~\Cref{thm:effectOfBias} yields the desired bound on $\sU(\Ast)$. 
    \end{proof}
\noindent
     The following result shows that when $\beta=1$, the algorithm has near-optimal fairness and utility also follows directly from the proof of~\Cref{thm:effectOfBias}. 
     \begin{corollary}
         \label{cor:betaone}
         Consider an instance where the utilities of the candidates are drawn from the uniform distribution on $[0,1]$  and the preference lists are drawn from an arbitrary distribution. Assume there are constants  $\eta_1, \eta_2$ such that $n_j \geq \eta_1 n$ for each $j \in \{1,2\}$, $K \geq \eta_2 n$.
         When the bias parameter $\beta=1$, 
             \begin{align*}
                 \textstyle \sR(\Ast)& \geq   
                    \textstyle 1
                 -
                \textstyle O \left( \frac{\sqrt{\log n}}{\eta_1 \eta_2 \sqrt{n}}\right)\\ 
                 \textstyle   \sU(\Ast) & \geq  
                 \textstyle 
                     1 - \textstyle O\left(\frac{\sqrt{\log n}}{\eta_2\sqrt{n}} \right) \\
                \sP(\Ast) & \geq 1 -
                \textstyle O \left(  \frac{p\sqrt{\log n} }{\eta_1 \eta_2\sqrt{n}} \right).
             \end{align*}
     \end{corollary}
     \begin{proof}
         The bounds of $\cR(\Ast)$ and $\sU(\Ast)$ follow directly from~\Cref{thm:effectOfBias} by substituting $\beta=1$. The bound on $\sP(\Ast)$ requires an explanation because \Cref{thm:effectOfBias} only gives a lower bound on $\sP(\Ast)$. The main observation is that the set $A$ defined in~\Cref{cl:AB} is empty. Hence, as in~\eqref{eq:pi1pi2bound}, $\sP(\Ast)$ can be expressed as
         $$\frac{n_1}{n_2} \cdot \frac{g_2' \pm O(p\sqrt{n \log n})}{g_1' \pm O(p\sqrt{n \log n})}.$$
         The rest of the argument now follows as in the proof of~\Cref{thm:effectOfBias}.
     \end{proof}

     \subsubsection{Extension of~\Cref{thm:specialcase} to log-concave densities}
     \label{sec:logconcave}
     In this section, we show that the ~\Cref{thm:specialcase} extends to the setting where the utilities of candidates are drawn from an arbitrary log-concave density. We show that both the representational and the preference-based fairness metrics deteriorate with $\beta$. Finally, we give an example (\Cref{ex:logconcave}) showing that such deterioration may not happen for an arbitrary distribution over utilities. 

      \begin{theorem}
            \label{thm:specialcaselog}
            Let $\cD$ be a distribution on $[0,1]$  with log-concave density $f$ and $n_1 = n_2=K$. Then, for any utility maximizing algorithm $\cA$, 
            $\sR(\cA) \leq 2\beta \ln(\sfrac{1}{\beta}) \pm O \left( \frac{\sqrt{\log n}}{{\sqrt{n}}} \right).$
             and 
            $\sP(\Ast) \leq 2\beta \ln(1/\beta) + O \left( \frac{p\sqrt{\log n}}{{\sqrt{n}}} \right).$
            Further, as $\beta$ approaches 0, 
            $\sU(\Ast)$ approaches a value that is at most $\Ex[\sU(\Ast)]$ is at most $\frac{U(K,K)+U(O(\sqrt{\log K},K)}{U(K,2K)}.$
            Here $U(k,m)$ denotes the expected sum of the top $k$-values out of $m$ values chosen independently from the distribution given by $f$. 
        \end{theorem}
        \begin{remark}
             (i) As in the statement of~\Cref{thm:specialcase}, we can also derive an upper bound on the utility metric $\sU(\cA)$. But this bound shall depend on the specific density; (ii) We can also generalize the above statement to the more general setting of~\Cref{thm:effectOfBias}. 
        \end{remark}

        \begin{proof}
            We give a sketch of the main ideas in the proof. These can be formalized as in the proof of~\Cref{thm:effectOfBias}. We assume that $\beta \ln(1/\beta) \leq 0.5$, otherwise the desired bounds hold trivially. 
            
            \paragraph{Bounding $\sR(\Ast)$.}
            The proof proceeds in two steps. In the first step, we reduce the problem of bounding $\sR(\Ast)$ to the special case when $f(x) \propto e^{-\lambda x}$ for some parameter $\lambda.$ In the second step, we carry out explicit calculations for this special case. 
            
            \begin{figure}
                \begin{center}
                    \includegraphics[width=5in]{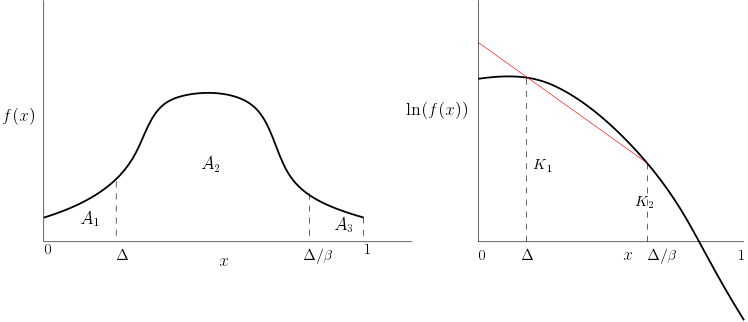}
                \end{center}
                \caption{Illustration of the proof of~\Cref{thm:specialcaselog}: The figure on the left plots the density $f(x)$. The quantities $A_1, A_2, A_3$ are the areas under the curve for the intervals $[0, \Delta], [\Delta, \Delta/\beta], [\Delta/\beta,1]$. The figure on the right plots $g(x):=\ln(f(x))$, which is a concave function. The line $L(x)$ is shown in red.  }
                \label{fig:logconc}
            \end{figure}

            \medskip
\noindent
            We now give details of the first step of the proof. 
            Given a value $v \in [0,1]$, let $X_v$ denote the number of candidates whose estimated utility lies in the range $[v,1]$. 
                The key idea, as in the proof of~\Cref{thm:effectOfBias}, is to find the value $\Delta$ such that $\Ex[X_\Delta] = K$ (recall that the total available capacity is $K$). One can then show using concentration bounds (as in~\Cref{thm:effectOfBias}) that the number of selected candidates is close (up to an error bound of $O(\sqrt{K})$) to $\Ex[X_\Delta]$. Now let $F(x)$ denote the c.d.f. of $f$, i.e., for a value $x \in [0,1]$, 
                $F(x) = \int_0^x f(x) dx.$
                Then $\Ex[X_\Delta \cap G_1] = K(1 - F(\Delta)$. Now, $X_\Delta \cap G_2$ consists of those candidates in $G_2$ whose true utility is at least $\Delta/\beta$. Therefore, 
                $\Ex[X_\Delta \cap G_2] = K(1-F(\Delta/\beta))$ (recall that $n_1 = K$). Thus, the parameter $\Delta$ satisfies the condition that 
                \begin{eqnarray}
                    K(1-F(\Delta)) + K(1-F(\Delta/\beta)) & = K \notag \\
                    i.e., \quad 1-F(\Delta/\beta) = F(\Delta). \label{eq:conc}
                \end{eqnarray}
\noindent
                This condition can be visualized as follows. We give some notation first.  Given a non-negative function $h:[0,1] \rightarrow \Re,$ values $v_1 \leq v_2$ lying in $[0,1]$, let $\Area_h(v_1,v_2)$ denote the area under $h$ from $v_1$ to $v_2$, i.e., $\int_{v_1}^{v_2} h(x) dx.$
                Let $A_1,A_2$ and $A_3$ denote $\Area_f(0,\Delta), \Area_f(\Delta, \Delta/\beta)$ and $\Area_f(\Delta/\beta,1)$ respectively (see~\Cref{fig:logconc}). Then the condition~\eqref{eq:conc} is the same as $A_3=A_1.$ Further, $\Ex[X_\Delta \cap G_1] = K(1-F(\Delta)) = K(A_2 + A_3)$ and $\Ex[X_\Delta \cap G_2] = K(1-F(\Delta/\beta)) = K \cdot A_3$. Since we are interested in the ratio of these two expectations, we define $$\gamma := \frac{A_3}{A_2+A_3} = \frac{A_1}{A_1+A_2},$$
                where the equality follows from the fact that $A_1=A_3$. Now, let $g(x)$ denote $\ln(f(x))$. Since $f$ is log-concave, $g$ is a concave function. Let $L(x)$ be the line that passes through the points $(\Delta, g(\Delta)$ and $(\Delta/\beta,g(\Delta/\beta))$ (see~\Cref{fig:logconc}). We assume that $L(x)$ is of the form $-\lambda x + b$, for a non-negative parameter $\lambda$ (the case when $\lambda > 0$ is similar). Let $\hf(x)$ denote $e^{L(x)}= Ce^{-\lambda x}$ for some parameter $C$. Since $g$ is concave,   $L(x) \geq g(x)$ for $x \in [0, \Delta] \cup [\Delta/\beta,1]$, and $L(x) \leq g(x)$ for $x \in [\Delta,\Delta/\beta]$.  Therefore, if $A_1', A_2', A_3'$ denote $\Area_{\hf}(0,\Delta), \Area_{\hf}(\Delta, \Delta/\beta)$ and $\Area_{\hf}(\Delta/\beta,1)$, then $A_1' \geq A_1, A_2' \leq A_2, A_3' \geq A_3$. Since the ratio $\frac{a}{a+b}$ is an increasing function of $a$ and decreasing function of $b$ for any non-negative real $a,b$, it follows that 
                 \begin{align}
                     \label{eq:gammacondition1}
                     \gamma \leq \min \left( \frac{A_1'}{A_1'+A_2'}, \frac{A_3'}{A_2'+A_3'} \right).
                 \end{align}
\noindent
                 Now, the above expression only involves the function $\hf$. Now we carry out the second step in the proof. There are two cases: 
                 \begin{itemize}
                     \item $\lambda \leq \frac{\beta}{\Delta} \ln(\sfrac{1}{\beta})$: We know  that $\hf(x)$ is of the form $C e^{-\lambda x}$. Hence, 
                     $$\frac{A_1'}{A_1'+A_2'} = \frac{\Area_{\hf}(0,\Delta)}{\Area_{\hf}(0,\Delta/\beta)} \leq 2 \lambda \Delta \leq  2\beta \ln(\sfrac{1}{\beta}),$$
                     where the second last inequality follows from the fact that $\Area_{\hf}(0,\Delta) \leq \max_{x \in [0,1]} \hf(X) \cdot \Delta = C \Delta$ and $\Area_{\hf}(0, \Delta/\beta) = \frac{C-C e^{-\lambda \Delta/\beta}}{\lambda} \geq \frac{C-C\beta}{\lambda} \geq \frac{C}{2\lambda},$ where we have used the fact that  $\beta \leq 1/2$ (since $\beta \ln(1/\beta)$ has been assumed to be at most $0.5$). 
                     \item $\lambda > \frac{\beta}{\Delta} \ln(\sfrac{1}{\beta}):$ In this case, we see that 
                     $$ \frac{A_3'}{A_2'+A_3'} = \frac{e^{-\lambda \Delta/\beta}-e^{-\lambda}}{e^{-\lambda \Delta}-e^{-\lambda}} \leq \frac{e^{-\lambda \Delta/\beta}}{e^{-\lambda \Delta}}, $$
                     where the last hand inequality follows from the fact that if $0 \leq a \leq b$, then $\frac{a}{b} \geq \frac{a-t}{b-t}$ for any $t \geq 0.$ Since the r.h.s. above is a decreasing function of $\lambda$ (because $\beta < 1$), we see that it is at most $\frac{\beta}{e^{-\beta \ln(1/\beta)}} \leq 2\beta,$
                     because $\beta \ln (1/\beta)$ has been assumed to be at most $0.5$. 
                 \end{itemize}
                 Thus, we see that in either of the two cases, $\gamma \leq 2 \beta \ln(1/\beta)$. Now, let $S$ be the set of candidates selected by $\Ast$. Arguing as in the proof of~\Cref{thm:effectOfBias}, we can show that 
                 $$\Ex \left[\frac{|S \cap G_2|}{|S \cap G_2|} \right] = \frac{A_3 \pm O(\sqrt{K \log k)}}{A_2+A_3 \pm O(\sqrt{K \log K}} + O \left( \frac{1}{n} \right).$$
                 Since $A_2 + A_3 \geq 0.5$ (as $A_1 + A_2 + A_3=1$ and $A_1=A_3$), the above can be written as 
                 $$\gamma \pm O \left( \frac{\sqrt{\log n}}{\sqrt{n}} \right).$$
                 Since $\gamma \leq 2\beta \ln(1/\beta)$, this completes the required bound on $\sR(\Ast)$. 
                \paragraph{Bounding $\sU(\Ast)$.} Given positive integers $k,m$, let $U(k,m)$ be the expectation of the top $k$ values from $m$ values chosen  independently from the density $f$.  As $\beta$ goes to 0, we have shown that the number of selected agents from $G_1$ goes to $O(\sqrt{\log K})$ with high probability. Therefore, as $\beta$ goes to 0, $\Ex[\sU(\Ast)]$ is at most $$\frac{U(K,K)+U(O(\sqrt{\log K},K)}{U(K,2K)},$$
                where the denominator gives the optimal utility of the top-$K$ candidates from the set of $2K$ available candidates, and the numerator gives an upper bound on the utilities of top $O(\sqrt{K \log K})$ agents from $G_2$ and the agents from $G_1$. This completes the desired bound on $\sU(\Ast)$. 
                \paragraph{Bounding $\sP(\Ast)$.} Finally, we bound $\sU(\Ast)$. The proof proceeds as in the proof in~\Cref{thm:effectOfBias} with the following changes: let $\Delta$ be defined as above by~\eqref{eq:conc}. As in Section~\ref{sec:BoundingR}, we define the events $\cE_1, \cE_2, \cE_3$ as follows:  $\cE_1$ is the event that $G_1(\beta,1) \in (1-F(\beta))n_1 \pm 8 \sqrt{n \log n}$, $\cE_2$ is the event that $G_1(\beta-\Delta,1) \in (1-F(\beta-\Delta))n_1 \pm 8 \sqrt{n \log n}$
                and $\cE_3$ is the event that $G_2(\beta-\Delta,1) \in (1-F(1-\Delta/\beta))n_2 \pm 8 \sqrt{n \log n}$. The rest of the arguments proceed as in Section~\ref{sec:boundingsU}. This completes the proof of the desired result. 
        \end{proof}
\noindent
        We now show an example illustrating that the result in~\Cref{thm:specialcaselog} does not extend to arbitrary distributions. 
        \begin{example}
        \label{ex:logconcave}
            Consider a distribution that places $0.5$ probability mass at $0$ and $1$ each (although this is a discrete distribution, the example can be easily extended to continuous distributions as well). Assume $\beta \in (0,1)$, $|G_1|=|G_2|=K$ and $p=1$. Among the $K$ candidates from each group $G_1$, roughly $K/2$ candidates have a utility of $1$, and the rest have $0$ utility. Similarly, about $K/2$ candidates from $G_2$ have utility $\beta$, and the remaining have utility $0$. Thus, among the top $K$ candidates, both groups have roughly $K/2$ candidates each. Thus, the representational fairness remains close to 1 even when $\beta$ goes to 0. Contrast this with the case when the distribution $f$ is log-concave. As~\Cref{thm:specialcaselog} shows, in such a setting the representational fairness goes to 0 as $\beta$ goes to 0. 
        \end{example}
    \subsection{Proof of Theorem~\ref{thm:institutionWiseConstraints}}\label{sec:proofof:thm:institutionWiseConstraints}
        In this section, we prove \cref{thm:institutionWiseConstraints}.
        For ease of readability, we restate \cref{thm:institutionWiseConstraints} below.

        \thminstitutionWiseConstraints*

       \noindent 
        The algorithm $\Ainst$ proceeds by running $\Ast$ on two independent instances, $\cI_j, j \in \{1,2\}$. 
        In $\cI_j$, there are $n_j$ candidates belonging to group $G_j$ and institution $\ell \in [p]$ has $k_{\ell,j} := k_\ell \cdot n_j/n$ seats. Let $M_j$ be the assignment obtained by running $\Ast$ on $\cI_j$; and $M$ denote the resulting assignment for $G_1 \cup G_2$, i.e., all the candidates in $[n]$. 
        The proof is divided into three steps.
        In the first step, we prove the representational fairness guarantee of $M$, i.e., $\sR(\Ainst)=1$, and discuss how it implies the bound on the utility ratio $\sU(\Ainst)$ in the second step. %
        In the third step, we analyze \ouralgo{} to conclude the proof. The bounds on representational fairness and utility follow from the institution-wise constraints. The bulk of the technical non-triviality in showing the desired bound on $\sU(\Ainst)$ lies in proving~\Cref{lem:gap}.

        \subsubsection{Step 1: Bounding $\sR$} 
        \label{sec:boundingRinst}

            Let $S$ be the set of candidates matched in $M$, i.e., $S := M^{-1}([p]).$
            Let $S_j := S \cap G_j$ be the candidates in $G_j$ that are assigned by $M_j$. Let $K_j$ denote $K \cdot n_j/n$.  Observe that $K_j$ is the total available capacity in the instance $\cI_j$. 
            Since $n \geq K$, $n_j \geq K_j$ and hence $|S_j| = K_j$. Recall that $q_j$ denotes the fraction of the candidates from $G_j$ that are assigned by $M$. Therefore, 
            $$\frac{q_1}{q_2} = \frac{K_1/n_1}{K_2/n_2} = \frac{K/n}{K/n} = 1.$$
            Similarly, $\frac{q_2}{q_1} = 1$. This shows that $\sR(\Ainst) = 1.$

        \subsubsection{Step 2: Bounding $\sU$}
        \label{sec:boundingUinst}
        We now bound~$\sU(\Ainst)$. 
        A random input for the instance can be generated as follows: (i) sample utilities and preferences for the $n$ candidates independently from the distribution $\cD$ and $\cL$ respectively; (ii) select a subset of $n_1$ candidates uniformly at random from the set of $n$ candidates, and assign them to $G_1$, and assign the remaining candidates to $G_2$. 

        Now, we condition on the utility values of the $n$ candidates obtained in step~(i) above. We show that conditioned on these values, the expected utility ratio is close to 1, and hence, $\sU(\Ainst)$ is close to 1 after removing this conditioning as well. 

        Let $S$ be the set of candidates with utility among the top $K$ values. We use~\Cref{thm:Hoeffding} for negatively correlated random variables as follows. For a candidate $i \in S$, let $X_i$ be the indicator random variable that is 1 if and only if $i$ is assigned to group $G_1$. Then the variables $X_i$'s are negatively correlated, and $\Ex[X_i] = \frac{n_1}{n}$.
        Now observe that $|S \cap G_1| = \sum_{i=1}^K X_i$, and hence, $\Ex[|S \cap G_1| = \frac{n_1 K}{n}$. Applying~\Cref{thm:Hoeffding}, we see that 
        $$\Pr \left[|S \cap G_1| \leq \frac{n_1 K}{n} - 2 \sqrt{K \log K}\right] \leq \frac{1}{2K^2}.$$
        Similarly, 
         $$\Pr \left[|S \cap G_2| \leq \frac{n_2 K}{n} - 2 \sqrt{K \log K} \right] \leq \frac{1}{2K^2}.$$
         The above two inequalities show that with probability at least $1-\frac{1}{K^2}$, $\Ainst$ selects  the top $K-4 \sqrt{K \log K}$ candidates from $S$. Let $U$ denote the total utility of the candidates in $S$. Then total utility of the top $K-4 \sqrt{K \log K}$ candidates in $S$ is at least $U \left(1-\frac{4 \sqrt{K \log K}}{K} \right).$
        Thus, 
        $\Ex[sU(\Ainst)] \geq \left( 1 - \frac{1}{K^2} \right) \cdot \left(1-\frac{4\sqrt{K \log K}}{K} \right) \geq 1 - O \left( \frac{\sqrt{\log K}}{\sqrt{K}} \right) = 1 - O \left( \frac{\sqrt{\log n}}{\sqrt{\eta_2 n}} \right).$     
This gives the desired bound on $\sU(\Ainst).$

            \subsubsection{Step 3: Bounding $\sP$}
            \label{sec:boundingPinst}
             Recall that $\cI_j$ denotes the instance corresponding to candidates in the group $G_j$, and $M_j$ denotes the assignment obtained by running $\Ast$ on $G_j$. Let $\pi(M_j)$ denote the number of candidates that are assigned the most preferred choice by $M_j$. Using~\Cref{lem:concentrationbounds}, we see  
             for each $j \in \{1,2\}$:
             $$\Pr[|\pi(M_j)- \E[\pi(M_j)]| \geq 8p\sqrt{K \log K}] \leq 1/K^2.$$
             Applying~\Cref{lem:gap} with $a=n/n_j$, we get 
             $$\left|\frac{n_j}{n} \cdot \Ex[\pi(M)] - \Ex[\pi(M_j)] \right| = O(p \sqrt{K \log K}),$$
             where $\pi(M)$ denotes the expected number of candidates that are assigned their first choice in the assignment $M$ obtained when we run $\Ast$ on the entire input $\cI$. Combining the above two inequalities, we see that with probability at least $1-1/K^2$, the following event happens for each $j \in \{1,2\}$: 
             \begin{align}
                 \label{eq:fairnessevent}
                 \pi(M_j) = \frac{n_j}{n} \cdot \Ex[\pi(M)] \pm O(p \sqrt{K \log K}). 
             \end{align}
        Recall that $\pi_j$ denotes $\pi(M_j)/n_j$. Hence, 
        $$\frac{\pi_2}{\pi_1} = \frac{\Ex[\pi(M)] \pm O(np \sqrt{K \log K})/n_2}{\Ex[\pi(M)] \pm O(np \sqrt{K \log K})/n_1}.$$
        Now, $\Ex[\pi(M)] \geq \min_\ell k_\ell = \eta_3 K$ and $n_1, n_2 \geq \eta_1 n$. Therefore, we can bound the above expression as follows
        $$\frac{\pi_2}{\pi_1} \geq 1 -  O\left(\frac{p \sqrt{ \log K}}{\eta_1 \eta_3 \sqrt{K}} \right).$$
        The same bound applies for $\frac{\pi_1}{\pi_2}$. Since the above-mentioned events happen with a probability of at least $1-1/K^2$, we see that 
        $$ \Ex \left[ \frac{\min(\pi_1, \pi_2)}{\max(\pi_1, \pi_2)} \right] \geq \left( 1 - 1/K^2 \right) \cdot \left(1 -  O\left(\frac{p \sqrt{ \log K}}{\eta_1 \eta_3 \sqrt{K}} \right)\right) = 1 -  O\left(\frac{p \sqrt{ \log K}}{\eta_1 \eta_3 \sqrt{K}} \right). $$
        This proves the desired bound on $\sP(\Ainst)$, and completes the proof of~\Cref{thm:institutionWiseConstraints}.

\section{Conclusion,  Limitations, and Future Work}
  This paper formalizes and studies the impact of biases against socio-economic groups for a centralized selection problem.
Our first result shows that biases in utilities can result in assignments that not only lower utility and decrease representation but also fail to conform to preference-based notions of fairness 
 (\cref{thm:effectOfBias}). 
    We present a family of institution-wise representational constraints and show that the algorithm \ouralgo{} achieves both preference-based fairness and near-optimal utility (\cref{thm:institutionWiseConstraints}).
   We extensively evaluate the performance of \ouralgo{} on real-world and synthetic datasets, along with several deviations from the assumptions of \cref{thm:institutionWiseConstraints}. We observe that it achieves high utility and near-optimal preference-based fairness compared to baselines. 
    \cref{thm:institutionWiseConstraints} relies on the assumption that the utilities of individuals are independently and identically distributed. This i.i.d. assumption is not specific to our work and also applies to the results of the prior works \cite{KleinbergR18, celis2020interventions} for the special case of subset selection (where there is only one institution and hence, the preferences of the candidates are vacuous). Extending these prior works in the non-i.i.d. setting is already an interesting direction.
The proof of \cref{thm:institutionWiseConstraints} can be generalized to more than two groups by reserving capacity at each institution for each of the groups in proportion to their sizes and running $\Ast$ on several independent instances, one for each group.

   Future research could explore settings where different institutions value candidates differently, such as the hospital-resident matching problem, particularly in the presence of biases.  If institutions' utilities are uncorrelated, adapting our results might need completely new approaches. This could involve modeling how institutions' rankings of candidates correlate. A starting point might be the model proposed by ~\cite{remi2022statistical}, which assumes that different institutions' utilities for a candidate are correlated.

    Finally, we emphasize that discrimination is a systemic problem and our work only studies a part of it. 
    To holistically counter bias and avoid unintended negative impacts, it is also important to study the long-term consequences of interventions, e.g., investments in skill development \cite{heidari2021longterm} and the impact on biases in evaluations \cite{celis2021longterm}, and evaluate the efficacy of interventions in specific contexts both pre- and post-deployment.

\section*{Acknowledgments} 
This work was funded in part by grants from Tata Sons Private Limited, Tata Consultancy Services Limited, and Titan, and NSF Awards CCF-2112665 and IIS-2045951.
We acknowledge Anay Mehrotra for his help in executing some components of the empirical work and proofs at an earlier stage of the project.
We acknowledge Vineet Goyal and Bhaskar Ray
Chaudhari for initial discussions regarding theoretical results.

\newpage

\bibliographystyle{plain}
\bibliography{references}

\appendix

\newpage

\section{Additional Discussion on Theorem \ref{thm:institutionWiseConstraints}}
\label{sec:discussionAinst}
In this section, we give a detailed discussion on issues with potential approaches for deriving an algorithm with the guarantees mentioned in~\Cref{thm:institutionWiseConstraints}. We also discuss properties of the algorithm~$\Ainst.$

\paragraph{Properties of the algorithm $\Ainst$.}
We now state some additional properties of the algorithm $\Ainst$:
\begin{itemize}
    \item[(a)] {\bf Group-wise stability}: The algorithm $\Ainst$ has the property that the assignment it produces is stable when restricted to candidates from the same group. Indeed, by definition, $\Ainst$ executes $\Ast$ on each group (with certain institutional capacities). Since $\Ast$ outputs a stable assignment, it follows that $\Ainst$ outputs a group-wise stable assignment. 

    \item[(b)] {\bf Strategy proof property}: {Assuming candidates' socially salient groups are pre-determined, \ouralgo{} is strategy-proof for the candidates with respect to their respective groups.} 
Moreover, candidates have no incentive to misreport their preferences. Indeed, if a candidate that prefers institution $i$ as a top choice, reports another institution $i'$ as its top choice, then the likelihood that \ouralgo{} assigns it $i$ only decreases. 
    
    \item[(c)] {\bf Pareto efficiency}: The algorithm $\Ainst$ is Pareto-efficient. Consider two candidates $i$ and $i'$ with identical preference lists and $\hu_i > \hu_{i'}$. Then the institution assigned to $i$ is preferable (to $i$ and $i'$) to the institution assigned to $i'$. 
\end{itemize}

\paragraph{Algorithms that learn the bias parameter $\beta$.} One potential approach for proving~\Cref{thm:institutionWiseConstraints} would be learn the parameter $\beta$ from the given data. Once $\beta$ is known, the utilities of the disadvantaged group can be corrected.  In some contexts, if the scenario consisted of repeated interactions, and we had access to historical data, one may learn the most likely value of $\beta$. {Exploring various interventions in situations involving repeated centralized matching is a crucial direction for future research.}
However, we follow prior works (see e.g. \cite{KleinbergR18, celis2020interventions, EmelianovGGL20, mehrotra2022intersectional}) that focus on one round-settings. In such scenarios,  obtaining an estimate of $\beta$ can be difficult. 
{We also note that algorithms dependent on estimates of $\beta$ may degrade due to inaccuracies in these estimates and are not resilient to explicit biases, such as those from human evaluators. For instance, consider a setting where an interviewer estimates the utilities of the candidates. One potential approach could be to re-scale the estimated utilities reported by the interviewer. 
 However, if the interviewer is explicitly biased, they can assign disproportionately low scores to a particular group of candidates; this would make any attempt to learn $\beta$ and re-scale utilities insufficient. In contrast, the algorithms that are typically used in the real-world settings do not rely on estimates of $\beta$ and hence, are more resilient to explicit biases.}

\paragraph{Algorithms with group-wise representational constraints.}
We now consider algorithms that add 
  proportional representation constraints to $\Ast$.  More concretely, we select the top $|G_j|/n$ fraction of candidates from each group $G_j$. This ensures that the representational fairness of such an algorithm is 1. To maintain group-wise stability property, we now run $\Ast$ on these candidates. See~\Cref{alg:groupwiseconstraints} for a formal description of this algorithm $\Agroup$. While this algorithm can be shown to enhance representational fairness while obtaining near-optimal utility, preference-based fairness may be low. Indeed, consider an example consisting of two institutions, each with capacity $K/2$, and $|G_1| = |G_2| = K/2.$ Assume that each candidate prefers the first institution to the second one. Since $n=K$, the algorithm selects all the candidates. Assuming $\beta$ is close to 0, the observed utilities of candidates from $G_2$ would be higher than those from $G_1$. Therefore, the algorithm $\Agroup$ would assign the first institutions to candidates from the first group only. Thus, the preference-based fairness metric of $\Agroup$ will be close to 0. We validate this fact empirically as well (see~\Cref{fig:simulation:synthetic_data}), where $\Ainst$ performs much better than $\Agroup$ in the preference-based fairness metric. 

\vspace{-2mm}
\paragraph{Algorithms that explicitly enforce preference-based fairness.} We consider algorithms that explicitly enforce preference-based fairness and maximize utility subject to such a constraint. More concretely, consider the following algorithm (assume for the sake of simplicity that $|G_1|=|G_2|$): each institution $\ell$ considers the candidates that prefer $\ell$ as their first choice. We assign the highest utility candidates from this subset to institution $\ell$ while ensuring that no more than half the capacity at this institution is allocated to a single group (in case some slots are vacant at the end, we can fill them with the highest utility unassigned candidates). Such an algorithm $\cA$ would indeed achieve high preference-based fairness, but may not obtain near-optimal utility even when there is no bias. The following example illustrates the issue: suppose there are two institutions  $A$ and $B$, each with capacity $K/2$, and $K$ candidates in each of the two groups. The distribution over preferences is such that $0.75K$ candidates from each of the groups prefer institution $A$ to $B$. Each candidate's utility is drawn uniformly from $[0,1]$. 

\smallskip
We now show that the above-mentioned algorithm has sub-optimal utility in this example. We recall the following fact (see~\Cref{cl:orderstatisticbound}): suppose we choose $m$ values independently and uniformly at random from the interval $[0,1]$, and let $Y_{k,m}$ denote the sum of the highest $k$ values among these $m$ values. Then $\Ex[Y_{k,m}] = k - \sfrac{k(k+1)}{2(m+1)}$. 

\smallskip
Now, the expectation of optimal utility in the above-mentioned example, i.e., the expected utility of the top $K$ candidates is the same as $\Ex[Y_{K,2K}] = 1-\frac{K}{2K+1} \approx 0.75K$. Now, there are $0.75K$ candidates in each group preferring institution $A$. The algorithm selects the top $0.25K$ such candidates from each group. Therefore, the expected utility of the candidates assigned to institution $A$ is:
$2\Ex[Y_{0.25K,0.75K}] \approx 0.5K(1-1/6) \approx 0.42K.$
Now, there are $0.25K$ candidates from each group preferring institution $B$, and the algorithm selects all of them. Therefore, the expected utility of the candidates assigned to institution $B$ is 
$2 \Ex[Y_{0.25K,0.25K}] = 0.25K.$
Thus, the total expected utility of the candidates selected by this algorithm = $0.125K + 0.42K = 0.67K$, which is below the optimal value of $0.75K$. Intuitively, this gap occurs because the set of candidates with optimal utility is expected to have $0.75K$ candidates that prefer $A$; whereas the algorithm selects only $0.5K$ such candidates.  

\smallskip
This example can be extended to create an instance such that the expected optimal utility is almost twice the expected utility of candidates selected by the algorithm. Consider a setting with $n$ candidates, $p$ institutions, and total capacity $K$ with each institution having $K/p$ capacity. Assume that $n \gg K$ and the utility of each agent is drawn uniformly from $[0,1]$. Now, suppose the distribution over preference lists is such that  $n-K$ candidates prefer the first institution as their top choice, and for every other institution $\ell$, $K/p$ candidates prefer $\ell$ as their top choice. Now, the expected optimal utility of top $K$ candidates is about $K(1-\sfrac{K}{2n}) \approx K$ since $n \gg K$. For an institution $\ell$ other than the first institution, let $J_\ell$ denote the $K/p$ candidates that prefer $\ell$ as their first choice. The algorithm would assign $J_\ell$ to institution $\ell$. The total expected utility of candidates in $J_\ell$ is $\frac{K}{2p}.$ The total utility of the candidates assigned to the first institution can be at most $\frac{K}{p}$. Summing over all institutions, we see that the total expected utility of selected candidates is at most  $\frac{Kp}{2p} + K \approx \frac{K}{2},$ which is roughly half of the expected optimal utility.

\section{Additional Discussion and Figures for the Real-World Simulation With JEE 2009 Scores}\label{sec:additional_JEE}

        \subsection{Additional figures and implementation details}

        \paragraph{Distribution of scores.}
                The 23 IITs are among the most prestigious engineering schools in India.
                Each year, admissions into the IITs are decided through a centralized matching system, where applicants are ranked based on their scores in the Joint Entrance Exam (JEE). Their score is their estimated utility, which is then matched with preference across different major-institute pairs \cite{baswana2019centralized}. 
                In 2009, the JEE test scores were released in response to a Right to Information application \cite{rti_against_jee}. 
                The preferences of candidates were not released.
                The data contains the test scores of 384,977 students who took the IIT-JEE 2009.
                In addition to scores, the dataset has each student's self-reported binary gender and birth category.
                Birth category is officially assigned as a socioeconomic label, and the most privileged group is the general category \cite{baswana2015joint,Sowell_2008}.
                A previous analysis of the dataset shows the discrepancies between student birth category or gender as well as score distributions \cite{CelisKMV23}.
            Under the $\beta$-bias model, this would equate to a bias parameter of $\beta = 0.69$ when the protected attribute is gender and $\beta = 0.52$ when the protected attribute is the birth category.

        \begin{table}[ht!]
\centering
\begin{tabular}{|l|l|l|}
\hline
\textbf{Major-Institution Pair} & \textbf{Opening Rank} & \textbf{Closing Rank} \\ \hline
CSE (4Yr), IIT Bombay & 3 & \textbf{86} \\ \hline
EE (4Yr), IIT Bombay & 8 & \textbf{109} \\ \hline
CSE (4Yr), IIT Delhi & 1 & \textbf{154} \\ \hline
CSE (4Yr), IIT Kanpur & 2 & \textbf{181} \\ \hline
CSE (4Yr), IIT Madras & 5 & \textbf{215} \\ \hline
EE (4Yr), IIT Delhi & 108 & \textbf{241} \\ \hline
EE w/ Micro (5Yr), IIT Bombay & 117 & \textbf{245} \\ \hline
CSE (5Yr), IIT Delhi & 187 & \textbf{278} \\ \hline
EE (4Yr), IIT Madras & 32 & \textbf{310} \\ \hline
EE w/ Info. \& Comm. Tech. (5Yr), IIT Delhi & 284 & \textbf{369} \\ \hline
EE w/ Comm. \& SP (5Yr), IIT Bombay & 266 & \textbf{379} \\ \hline
EE (4Yr), IIT Kanpur & 39 & \textbf{416} \\ \hline
CSE (5Yr), IIT Kanpur & 216 & \textbf{422} \\ \hline
ME (4Yr), IIT Bombay & 72 & \textbf{494} \\ \hline
CSE (5Yr), IIT Madras & 333 & \textbf{502} \\ \hline
CSE (4Yr), IIT Kharag. & 276 & \textbf{527} \\ \hline
EE (5Yr), IIT Kanpur & 423 & \textbf{608} \\ \hline
ME (4Yr), IIT Delhi & 237 & \textbf{634} \\ \hline
ME w/ CAD \& Auto. (5Yr), IIT Bombay & 419 & \textbf{637} \\ \hline
EE w/ Micro \& VLSI (5Yr), IIT Madras & 339 & \textbf{716} \\ \hline
ME w/ CIM (5Yr), IIT Bombay & 556 & \textbf{757} \\ \hline
EE w/ Power (4Yr), IIT Delhi & 477 & \textbf{758} \\ \hline
EEC (4Yr), IIT Kharag. & 133 & \textbf{762} \\ \hline
EE w/ Comm. Eng. (5Yr), IIT Madras & 458 & \textbf{764} \\ \hline
Math. \& Comp. (5Yr), IIT Delhi & 348 & \textbf{789} \\ \hline
ME (4Yr), IIT Kanpur & 497 & \textbf{806} \\ \hline
ME (4Yr), IIT Madras & 275 & \textbf{820} \\ \hline
CSE (5Yr), IIT Kharag. & 431 & \textbf{877} \\ \hline
EE (4Yr), IIT Kharag. & 596 & \textbf{920} \\ \hline
Chem. Eng. (4Yr), IIT Bombay & 244 & \textbf{928} \\ \hline
EE w/ Power Sys. \& Elec. (5Yr), IIT Madras & 773 & \textbf{937} \\ \hline
CSE (4Yr), IIT Roorkee & 471 & \textbf{984} \\ \hline
ME (5Yr), IIT Kanpur & 808 & \textbf{992} \\ \hline
\end{tabular}
\caption{Top 33 major-institution pairs sorted by closing ranks in the 2009 IIT-JEE.
                See \cref{sec:additional_JEE} for a discussion of the dataset.}
\label{fig:jee_ordering_top31}
\vspace{-3mm}
\end{table}

        \paragraph{Utilities.}
            The notion of latent utility is not available through our real-world data, so we consider JEE scores as estimated utility. This is because if we consider candidates of equal ``latent utility'', the ones from disadvantaged groups perform poorer on standardized tests, our measure of estimated utility \cite{elsesser2019lawsuitSATACT}.
            In India, this could be because fewer girls attend primary school \cite{alderman1998gender,censusinfo}, and many girls are forced to drop out of school to get married or help with housework \cite{girl_drop_out}. 
            Because of these societal differences, we expect a female student to perform worse on the JEE than a male student of comparable ``latent utility''. %
            Students in birth categories other than the general category face similar disadvantages \cite{kumar2021castenexus}.

        \paragraph{Preferences.}
            The dataset does not include individual preferences, so we aggregated information from the opening and closing ranks of major-institution pairs, $(M, I)$.
            We ordered pairs by their closing ranks $C_{(M, I)}$ in the general category and considered the range $[O_{(M, I)}, C_{(M, I)}]$ of the ranks of students admitted to $(M, I)$ (Table \ref{fig:jee_ordering_top31}). %
            We observed that if $C_{(M_1, I_1)}$ is substantially less than $C_{(M_2, I_2)}$, then most applicants will strictly prefer the earlier ranked program, which is consistent with those of journalists \cite{mind2022IITallotment,verma2022JOSAA}.
            We model preferences over programs using the Mallows distribution (specified by $d_{\rm KT}$) defined by a dispersion parameter $0 \leq \phi \leq 1$ and the central ranking $\rho$ of program-institute pairs by their closing ranks $C_{(M, I)}$. The capacities of each program were publicly available online.
            We considered the closing All India Rank for the general category as opposed to across all categories. This is because the JOSSA has already instituted quotas for students from more disadvantaged groups to be admitted into certain programs. As such, a difference in individual preferences in the underprivileged groups would affect the overall closing ranks significantly and the general ranks give a better picture.

        \paragraph{Algorithm implementation.}
        While implementing our algorithm, we implemented limits on the number of candidates that were considered for the sake of lowering runtime. When considering all 384,977 candidates, there was a fixed bound: the candidate in the disadvantaged group that had the lowest JEE score possible to still be matched under group- or institution-wise constraints but not under the unconstrained algorithm. Any candidate, regardless of group, with a lower JEE score would never be considered to be matched. For the sake of runtime, we found the lowest AIR applicable with both birth and gender as the protected attribute and stopped generating preferences after that rank (5000 and 7000, respectively).

        \subsection{Simulation that varies the dispersion parameter of the preference distribution using synthetic data} \label{sec:additional:empirical:dispersionpref}
    This simulation considers the effect of dispersion of preferences on the three algorithms we present by altering the dispersion parameter, $\phi$, of the Mallows distribution. A small $\phi$ indicates that preferences are more closely aligned to the central ranking, and $\phi \to 1$ indicates that preferences are drawn more randomly than based on the central ranking. Since \cref{thm:institutionWiseConstraints} does not depend on the shape of the preference distribution, we expect the preference-based fairness of $\Ainst$ to be near optimal regardless of $\phi$, but we would like to see the effect of $\phi$ on $\Agroup$ and $\Ast$.

    \paragraph{Setup.} We fix $n=1000$, $p=5$, and $k_i=100$ for $i \in [p]$, For $\cD \in \{\cD_{\rm Gauss}, \cD_{\rm Pareto}\}$, we set $\beta = \frac{3}{4}$ and vary $\phi \in [0, 1]$. We report the average and standard error of $\sP^{(1)}, \sP^{(3)}$, and $\sU$ over 50 iterations.

    \paragraph{Results and observations.} The results are reported in \cref{fig:simulation:corrpref}. Our main observations are that $\phi$ does not seem to have a large impact on preference-based fairness for $\Ainst$,  $\Agroup$ and $\Ast$ generally increase with $\phi$, and $\phi$ does not affect the utility ratio.
    As expected by \cref{thm:institutionWiseConstraints}, $\Ainst$ holds a high preference-based fairness ($\sP^{(1)} \geq 0.9, \sP^{(3)} \geq 0.95$). 
    As $\phi$ increases, preference-based fairness for $\Agroup$ and $\Ast$ generally increase.
    The exception is with $\Ast$ when utilities are sampled from $\cD_{\rm Pareto}$, where the preference-based fairness of the unconstrained algorithm does not change with $\phi$. 
    To closely model real-world scenarios, we choose for $\phi \leq 0.5$ because the central ranking is often important. 
    We also choose $\phi \geq 0.1$ to be able to differentiate between $\Agroup$ and $\Ast$. When utilities are drawn from $\cD_{\rm Gauss}$, the difference between $\sP^{(1)}(\Agroup)$ and $\sP^{(1)}(\Ast)$ is within 0.05 when $\phi \leq 0.10$. 
    Based on these constraints, we set $\phi = 0.25$ as the default in our simulations.
    The utility ratio does not change with $\phi$ because it does not depend on preference distribution. This is further discussed in \cref{sec:additional:empirical:additional_plot_synthetic}.
    Additionally, we observe that when varying $\phi$ between 0 and 1, the group-wise algorithm increases at a much earlier $\phi$ in synthetic data (\cref{fig:simulation:corrpref}) compared to real-world data (Table \ref{fig:jee_ordering_top31}). This is likely because the real-world dataset had 33 institutions compared to the 5 institutions used in synthetic simulations. With that many institutions, it would be much more difficult to increase the preference-based fairness, especially measured by the $\sP^{(1)}$ metric. 
\begin{figure}[h!]
            \centering
            \includegraphics[width=0.8\linewidth]{figures_FINAL/legend.pdf}
            \subfigure[{$\cD_{\rm Gauss}$ and $\sP^{(1)}$}]{
                \includegraphics[width=0.31\linewidth, trim={0cm 0cm 0cm 0cm},clip]{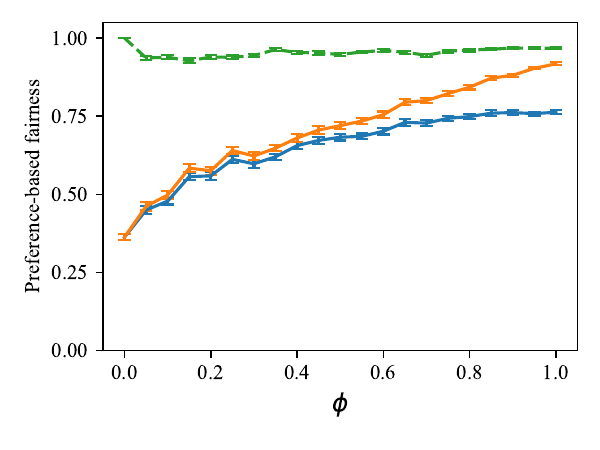}
            }
            \subfigure[{$\cD_{\rm Gauss}$ and $\sP^{(3)}$}]{
                \includegraphics[width=0.31\linewidth, trim={0cm 0cm 0cm 0cm},clip]{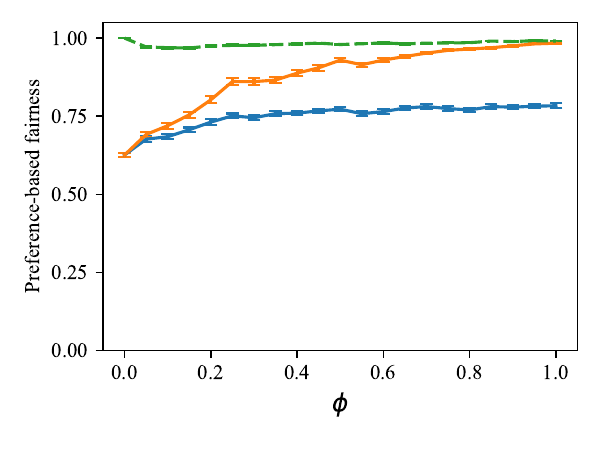}
            }
            \subfigure[{$\cD_{\rm Gauss}$ and $\sU$}]{
                \includegraphics[width=0.31\linewidth, trim={0cm 0cm 0cm 0cm},clip]{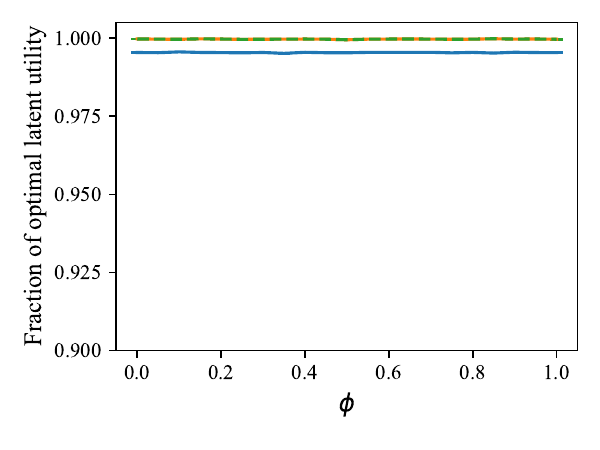}
            }
            \par 
            \subfigure[{$\cD_{\rm Pareto}$ and $\sP^{(1)}$}]{
                \includegraphics[width=0.31\linewidth, trim={0cm 0cm 0cm 0cm},clip]{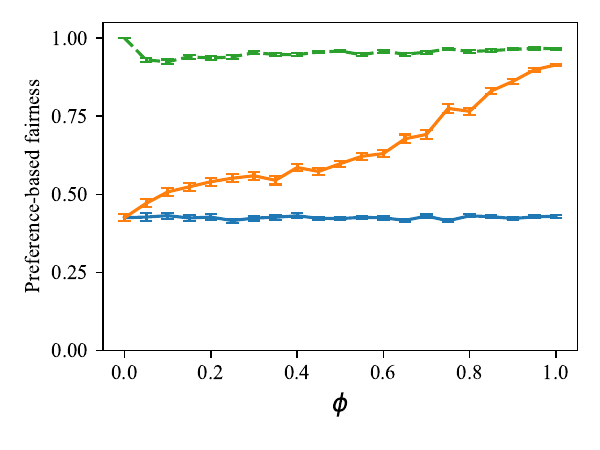}
            }
            \subfigure[{$\cD_{\rm Pareto}$ and $\sP^{(3)}$}]{
                \includegraphics[width=0.31\linewidth, trim={0cm 0cm 0cm 0cm},clip]{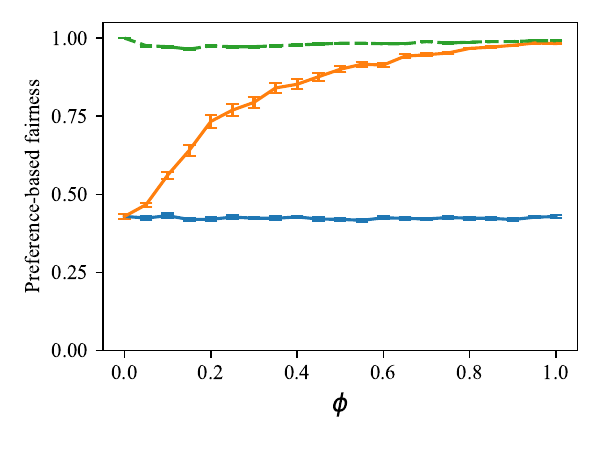}
            }
            \subfigure[{$\cD_{\rm Pareto}$ and $\sU$}]{
                \includegraphics[width=0.31\linewidth, trim={0cm 0cm 0cm 0cm},clip]{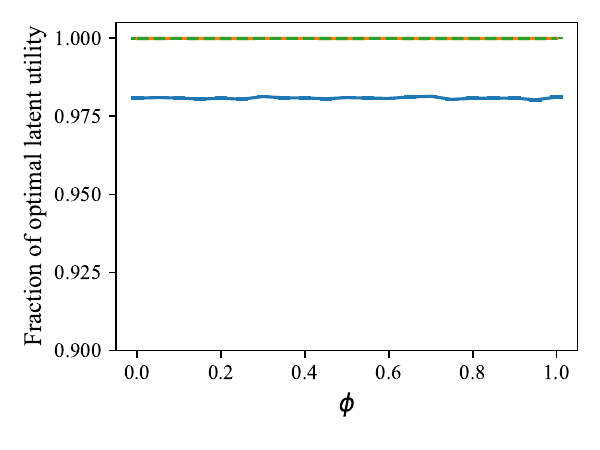}
            }
            \caption{
                $\sP^{(1)}$, $\sP^{(3)}$, and $\sU$ measured for synthetic data when the dispersion parameter for the Mallows distribution is varied over $\phi \in [0, 1]$. In (a), we see the preference-based fairness measured by $\sP^{(1)}$ when utilities are generated from $\cD_{\rm Gauss}$. In (b), we measure $\sP^{(3)}$ when utilities are generated from $\cD_{\rm Gauss}$. (c) shows the utility ratio when utilities are generated from $\cD_{\rm Gauss}$. (d), (e), and (f) show the results of (a), (b), and (c), respectively, when utilities are generated from $\cD_{\rm Pareto}$. 
                Our main observation is that $\phi$ does not have a large impact on preference-based fairness for $\Ainst$, while $\Agroup$ and $\Ast$ generally increase with $\phi$. See \cref{sec:additional:empirical:dispersionpref} for details and discussion.
                The $x$-axis denotes $\alpha$, the $y$-axis denotes $\sP^{(1)}$, $\sP^{(3)}$, or $\sU$, and the error bars denote the standard error of the mean over 50 iterations.
            }
            \label{fig:simulation:corrpref}
            \vspace{-1.5mm}
        \end{figure}

\section{Additional Discussion and Plots for Simulation with Synthetic Data from Section~\ref{sec:empirical}}\label{sec:additional:empirical:additional_plot_synthetic}
        In this section, we report the preference-based fairness, $\sP^{(1)}$, for additional values of $\beta$ in the simulation from \cref{sec:empirical} (\cref{fig:simulation:synthetic_data:other_beta}). We also discuss the utility ratio in the simulation.
        \subsection{Discussion of utility ratio}
        We do not include results about utility ratio for synthetic data because they do not vary with respect to preferences for any of our algorithms. Specifically, $\Ast$ always selects the $K$ candidates with the highest observed utilities, where $K$ is the total capacity of all institutions. Likewise, under the assumption that $|G_1| = |G_2|$, $\Agroup$ and $\Ainst$ both assign the top $K/2$ candidates in terms of utility from each group to institutions. This guarantees that $\sU(\Agroup) = \sU(\Ainst) \approx 1$ for all $\gamma$. Intuitively, $\sU(\Ast)$ increases with respect to $\beta$: this is detailed in \cref{thm:specialcase} and visualized in \cref{fig:specialcase_nonoise}(a).

        \subsection{Additional plots}
        \paragraph{Results and observations.} 
        The results of $\sP^{(1)}$ for additional values of $\beta$ support the discussion made in \cref{sec:empirical}. One additional observation is that as $\beta$ increases, $\sP^{(1)}(\Ainst)$ is not affected but $\sP^{(1)}(\Agroup)$ and $\sP^{(1)}(\Ast)$ increase at comparable levels of $\beta$. This is because the estimated utilities of the disadvantaged candidates increase as $\beta$ increases.

        \begin{figure}[h]
            \centering
            \vspace{0.2in}
            \subfigure[{$\cD_{\rm Gauss}$ and $\beta=\frac{2}{4}$}]{
                \includegraphics[width=0.31\linewidth, trim={0cm 0cm 0cm 0cm},clip]{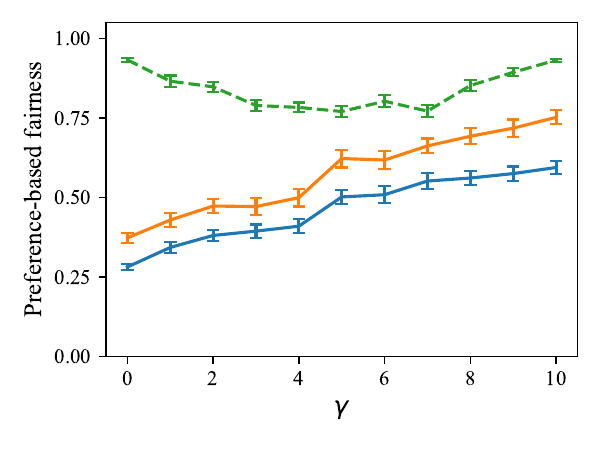}
            }
              \vspace{0.2in}
            \subfigure[{$\cD_{\rm Gauss}$ and $\beta=\frac{3}{4}$}]{
                \includegraphics[width=0.31\linewidth, trim={0cm 0cm 0cm 0cm},clip]{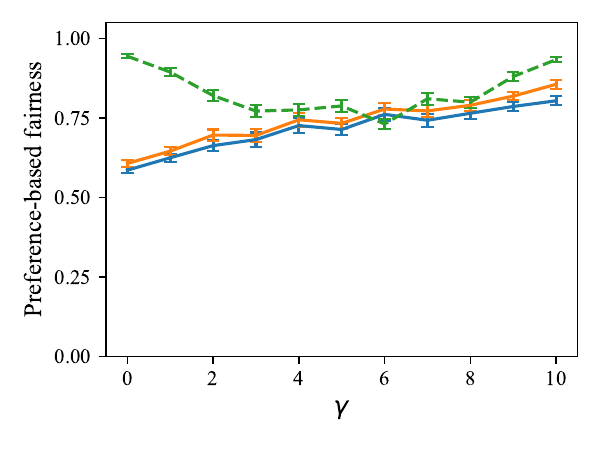}
            }
            \par 
            \subfigure[{$\cD_{\rm Pareto}$ and $\beta=\frac{2}{4}$}]{
                \includegraphics[width=0.31\linewidth, trim={0cm 0cm 0cm 0cm},clip]{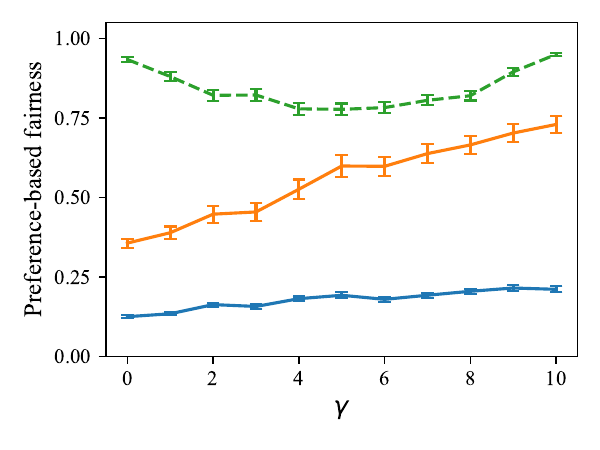}
            }
            \subfigure[{$\cD_{\rm Pareto}$ and $\beta=\frac{3}{4}$}]{
                \includegraphics[width=0.31\linewidth, trim={0cm 0cm 0cm 0cm},clip]{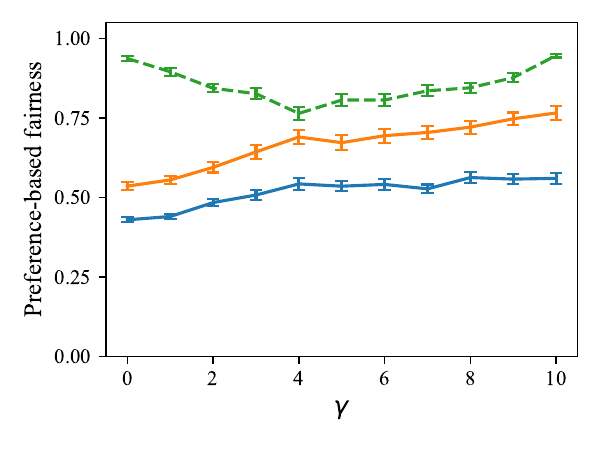}
            }
            \caption{
            Preference-based fairness as measured by the top-1 metric, $\sP^{(1)}$, with synthetic data where non-i.i.d. preferences are generated from Mallows distributions. The $x$-axis denotes $\gamma$, the Kendall-Tau distance between the central rankings, and the error bars denote the standard error of the mean over 50 iterations. (a) shows $\sP^{(1)}$ when utilities are generated from $\cD_{\rm Gauss}$ with $\beta = \frac{2}{4}$. (b) modifies $\beta$ to $\beta = \frac{3}{4}$. (c) and (d) are equivalent to (a) and (b), respectively, when utilities are drawn from $\cD_{\rm Pareto}$. See \cref{sec:simulation:synthetic_data} and \cref{sec:additional:empirical:additional_plot_synthetic} for details and discussion. We observe that institution-wise constraints achieve higher preference-based fairness than group-wise and unconstrained settings.}
            \label{fig:simulation:synthetic_data:other_beta}
        \end{figure}

\section{Empirically Studying Institution-Wise Constraints Under Relaxed Bounds} \label{sec:additional:empirical:relaxedbounds}

\begin{figure}[ht!]
            \centering
            \includegraphics[width=0.8\linewidth]{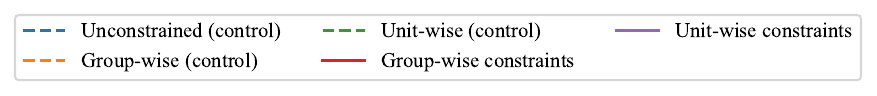}
            \par 
            \subfigure[{$\cD_{\rm Gauss}$ and $\sP^{(1)}$}]{
                \includegraphics[width=0.31\linewidth, trim={0cm 0cm 0cm 0cm},clip]{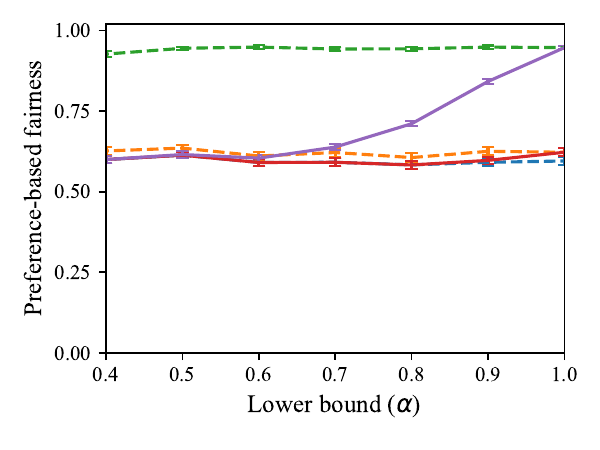}
            }
            \subfigure[{$\cD_{\rm Gauss}$ and $\sP^{(3)}$}]{
                \includegraphics[width=0.31\linewidth, trim={0cm 0cm 0cm 0cm},clip]{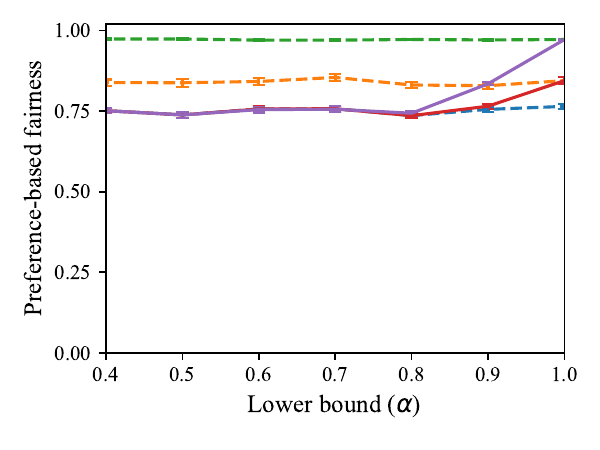}
            }
            \subfigure[{$\cD_{\rm Gauss}$ and $\sU$}]{
                \includegraphics[width=0.31\linewidth, trim={0cm 0cm 0cm 0cm},clip]{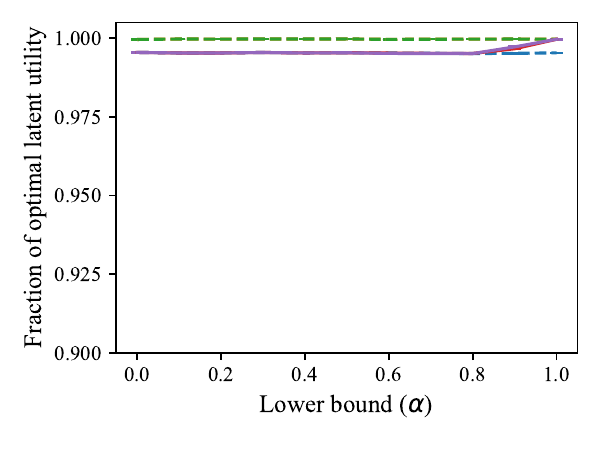}
            }
            \par
             \subfigure[{$\cD_{\rm Pareto}$ and $\sP^{(1)}$}]{
                \includegraphics[width=0.31\linewidth, trim={0cm 0cm 0cm 0cm},clip]{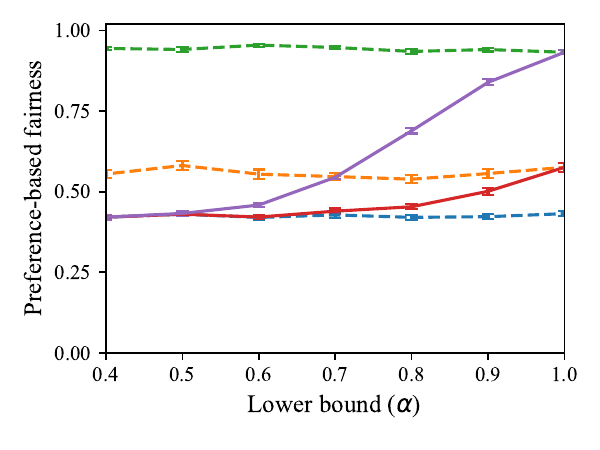}
            }
            \subfigure[{$\cD_{\rm Pareto}$ and $\sP^{(3)}$}]{
                \includegraphics[width=0.31\linewidth, trim={0cm 0cm 0cm 0cm},clip]{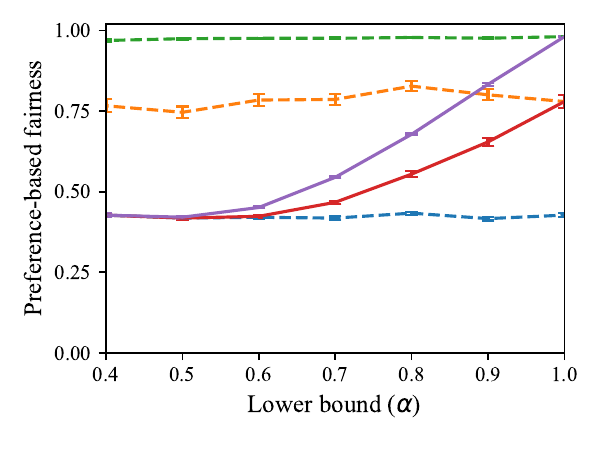}
            }
            \subfigure[{$\cD_{\rm Pareto}$ and $\sU$}]{
                \includegraphics[width=0.31\linewidth, trim={0cm 0cm 0cm 0cm},clip]{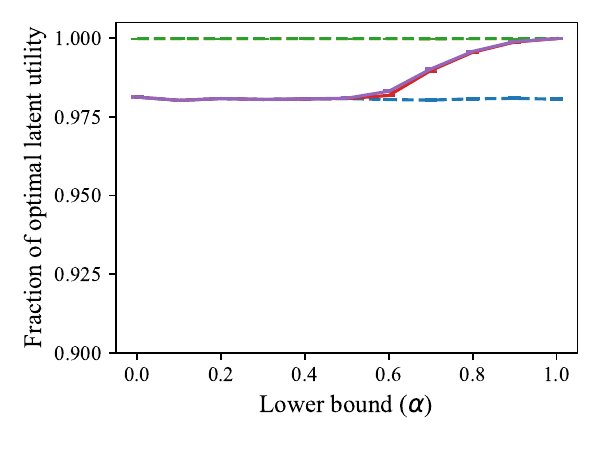}
            }
            \caption{
                Preference-based fairness is measured by $\sP^{(1)}$ and $\sP^{(3)}$, and the utility ratio measured by $\sU$ for relaxed institution-wise and group-wise bounds parameterized by strictness ratio $0 \leq \alpha \leq 1$. We compare results with relaxed bounds (solid lines) and results with strict bounds (dotted lines). In (a), we measure $\sP^{(1)}$ when utilities are drawn from $\cD_{\rm Gauss}$ In (b), we measure $\sP^{(3)}$ with utility distribution $\cD_{\rm Gauss}$. In (c), we measure the utility ratio for $0 \leq \alpha \leq 1$ with utilities from $\cD_{\rm Gauss}$. (d), (e), and (f) are equivalent to (a), (b), and (c), respectively, except that utilities are drawn from $\cD_{\rm Pareto}$ instead of $\cD_{\rm Gauss}$. Our main observation is that institution-wise constraints and group-wise constraints only begin to increase preference-based fairness compared to no constraints when the lower bound is high. See \cref{sec:additional:empirical:relaxedbounds} for details and discussion. The $x$-axis denotes $\alpha$, and the error bars denote the standard error of the mean over 50 iterations.
            }
            \label{fig:relaxedbounds}
        \end{figure}
    
    This section considers the effectiveness of $\Ainst$ under relaxed bounds. Specifically, instead of instituting a strict proportional quota for each institution, we consider the effect of a relaxed bound parameterized by $0 \leq \alpha \leq 1$. If $\alpha = 1$, the bound is strict and if $\alpha = 0$, there is no bound. In general, the quotas required by $\Ainst$ are $\alpha$ times the available capacity at each institution, with the remaining capacity available for candidates from either group. This is similar when considering relaxed group-wise constraints, except on an institution-wise basis rather than on a representational basis. Pseudocode for the algorithms are provided in \cref{sec:additional:pseudocode}. The purpose of this simulation is to analyze the practicality of more relaxed real-world constraints beyond the theoretical work done in \cref{thm:institutionWiseConstraints}. To do this, we measure the preference-based fairness and utility ratio of relaxed constraints compared to stringent ones (as required by $\Ainst$).
    
    \paragraph{Setup.} We fix $n=1000$, $p=5$, and $k_i=100$ for $i \in [p]$. For $\cD \in \{\cD_{\rm Gauss}, \cD_{\rm Pareto}\}$, we set $\beta = \frac{3}{4}$ because $\beta$ is often high in real-world situations. We generate preferences from a single central ranking using a Mallows distribution with $\phi = 0.25$ ($\gamma = 0$). We vary the relaxed bound parameter, $\alpha$, from $0$ (no bounds) to $1$ (strict bounds) with two groups of the same size, i.e., $|G_1|=|G_2|$. We report the average and standard error of $\sP^{(1)}, \sP^{(3)}$, and $\sU$ over 50 iterations.

    \paragraph{Results and observations.} The results are presented in \cref{fig:relaxedbounds}. 
    We know that the bounds instituted by $\Ainst$ and $\Agroup$ will not affect the preference-based fairness until the bound exceeds the proportion that $\Ast$ achieves. Based on the empirical results for $\beta = \frac{3}{4}$, we started the bounds at $\alpha = 0.4$ to better illustrate the changes.
    Our main observation is that preference-based fairness begins to increase for $\Ainst$ and $\Agroup$ as $\alpha > 0.5$, but the utility ratio does not increase until between $\alpha = 0.6$ ($\cD_{\rm Pareto}$) or $\alpha = 0.8$ ($\cD_{\rm Gauss}$).
    Under a Pareto distribution, $\sP^{(1)}$ for $\Ainst$ and $\Agroup$ begin to increase at $\alpha = 0.5$ from 0.45 to 0.90 ($\Ainst$) and 0.55 ($\Agroup$) as $\alpha \to 1$.
    Under a truncated Gaussian distribution, $\sP^{(1)}$ for $\Ainst$ and $\Agroup$ begin to increase at $\alpha = 0.6$ from 0.55 to 0.90 ($\Ainst$) and 0.60 ($\Agroup$) as the bound becomes 50\%.
    We see that institution-wise constraints are only effective when strictly enforced (for $\beta = \frac{3}{4}$, $\alpha > 0.5$). For situations in which $\beta$ is higher, the bound would have to be even higher. 
    We also see that preference-based fairness for the relaxed version of $\Ainst$ begins to increase at a lower $\alpha$ than the corresponding utility ratio does. Lower bounds may not be effective at all, particularly if the unconstrained algorithm already guarantees the low ratio.

\section{Empirical Results Considering Noise in Estimated Utilities}\label{sec:additional:empirical:noise}

   \subsection{Simulation considering noise in the bias parameter} \label{sec:additional:empirical:betanoise}
    In this section, we consider noise in the bias parameter of the $\beta$-bias model, which we refer to as the noisy $\beta$-bias model. Instead of a fixed bias parameter, we sample $\beta$ from a Gaussian distribution truncated to the interval $[0, 1]$ with a standard deviation of $0.1$. Examples of the noisy $\beta$-bias model centered around $\beta = 0.1, 0.2, 0.3$ are shown in \cref{fig:betaexamples}. This standard deviation ensures that for non-extreme values of $\beta$ ($0.1 \leq \beta \leq 0.9$), the expected value closely adheres to the center of the distribution (within $0.02$). Under these non-ideal conditions, we observe how closely empirical data adheres to the theoretical results of \cref{thm:institutionWiseConstraints}, as well as the effect of noise on $\Agroup$ and $\Ast$.

\begin{figure}[ht!]
            \centering
            \subfigure[{$\beta = 0.1$}]{
                \includegraphics[width=0.31\linewidth, trim={0cm 0cm 0cm 0cm},clip]{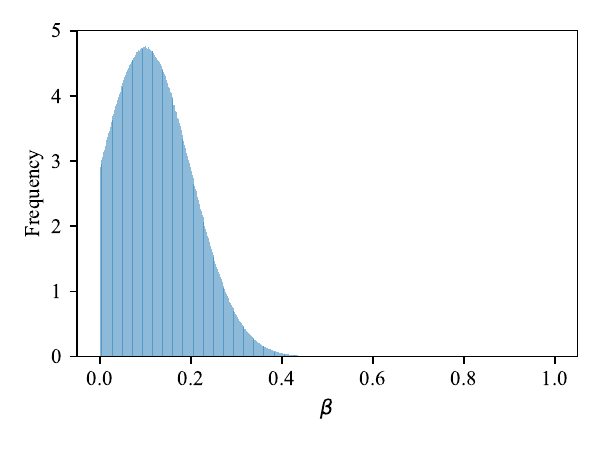}
            }
            \subfigure[{$\beta = 0.2$}]{
                \includegraphics[width=0.31\linewidth, trim={0cm 0cm 0cm 0cm},clip]{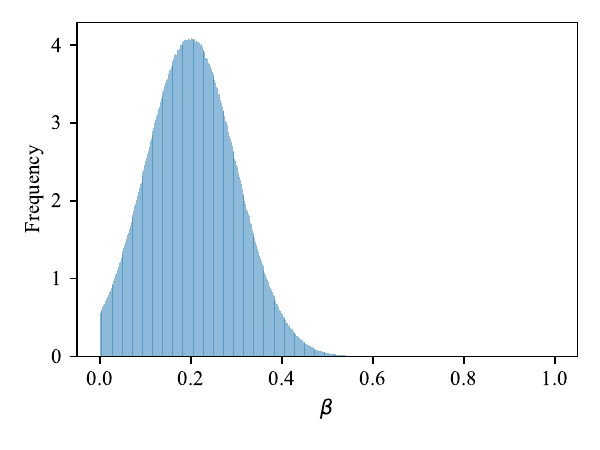}
            }
            \subfigure[{$\beta = 0.3$}]{
                \includegraphics[width=0.31\linewidth, trim={0cm 0cm 0cm 0cm},clip]{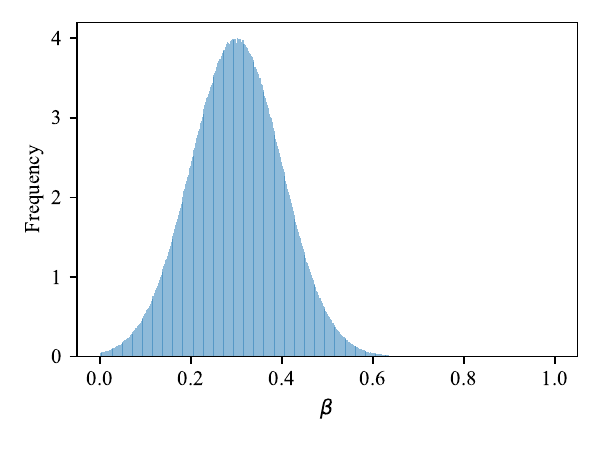}
            }
            \caption{
                Distributions of $\beta$ under the noisy $\beta$-bias model with standard deviation $0.1$. (a) shows the distribution when $\beta = 0.1$ with mean 0.129. (b) shows the distribution when $\beta = 0.2$ with mean $0.205$. (c) shows the distribution when $\beta = 0.3$ with mean $0.300$. The noisy $\beta$-bias model is detailed in \cref{sec:additional:empirical:betanoise}.
            }
            \label{fig:betaexamples}
        \end{figure}

    \paragraph{Setup.} We fix the parameters $n=1000$, $p=5$, and $k_\ell = 100$ for each $\ell \in [5]$. We vary $\beta \in [0, 1]$ and sample a bias parameter for each person in the disadvantaged group from the Gaussian distribution centered at $\beta$, truncated to the interval $[0, 1]$ with standard deviation 0.1. We draw latent utilities from $\cD \in \{\cD_{\rm Gauss}, \cD_{\rm Pareto}\}$, and draw preferences i.i.d. from a Mallows distribution with $\phi = 0.25$. We plot preference-based fairness $(\sP^{(1)}, \sP^{(3)})$ and utility ratio $(\sU)$ over 50 iterations.

    \paragraph{Results and observations.} Results are plotted in \cref{{fig:simulation:noisebias}}. 
    Our main observation is that even with noise in $\beta$, the preference-based fairness of $\Ainst$ is not affected (within 0.01) and is still high.  
    When looking at $\Agroup$ and $\Ast$, noise can skew preference-based fairness for extreme $\beta$ ($\beta \leq 0.15, \beta \geq 0.85$). 
    Specifically, when $\beta$ approaches 1, preference-based fairness in the presence of noise is less than preference-based fairness without noise. Under the condition of $\cD_{\rm Pareto}$ and $\sP^{(3)}$, the difference in $\sP^{(3)}$ was 0.1 for $\Agroup$ and 0.2 for $\Ast$ at $\beta = 1$.
    Noise also affected the utility ratio. At low $\beta,$ there was a difference between the utility ratio of the simulations with noise and without noise, but as $\beta$ approached 1, the utility ratios converged. Significantly, for $\beta \geq 0.6$, the difference between the utility ratio with and without noise was within 0.02.
    Overall, noise had minimal impact on preference-based fairness for institution-wise constraints, but had a notable impact on group-wise and unconstrained algorithms for $\beta \geq 0.85$ or $\beta \leq 0.15$. The noise did not have a significant effect on the utility ratio for large $\beta$.

     \begin{figure}[ht!]
            \centering
            \includegraphics[width=0.8\linewidth]{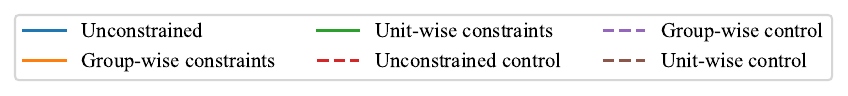}
            \subfigure[{$\cD_{\rm Gauss}$ and $\sP^{(1)}$}]{
                \includegraphics[width=0.31\linewidth, trim={0cm 0cm 0cm 0cm},clip]{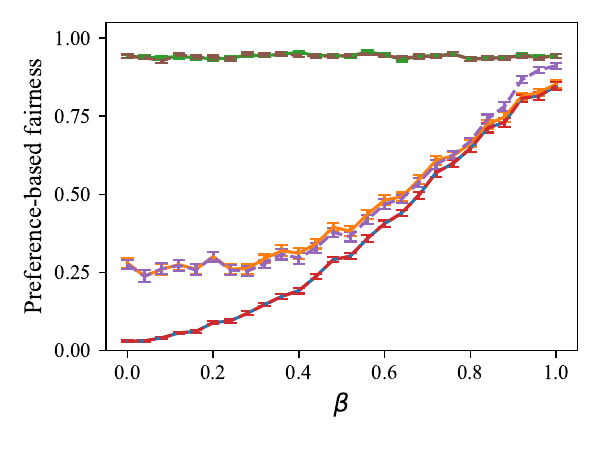}
            }
            \subfigure[{$\cD_{\rm Gauss}$ and $\sP^{(3)}$}]{
                \includegraphics[width=0.31\linewidth, trim={0cm 0cm 0cm 0cm},clip]{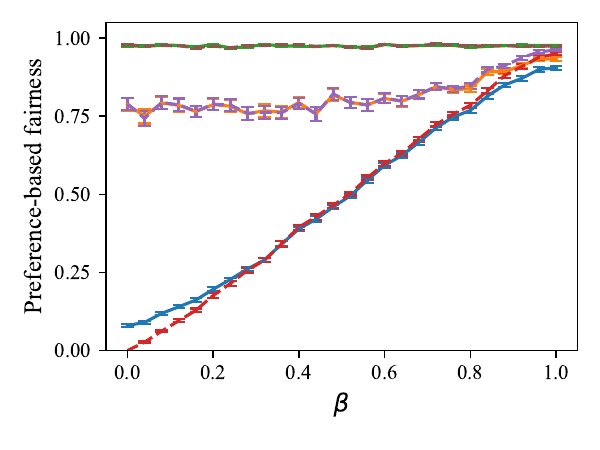}
            }
            \subfigure[{$\cD_{\rm Gauss}$ and $\sU$}]{
                \includegraphics[width=0.31\linewidth, trim={0cm 0cm 0cm 0cm},clip]{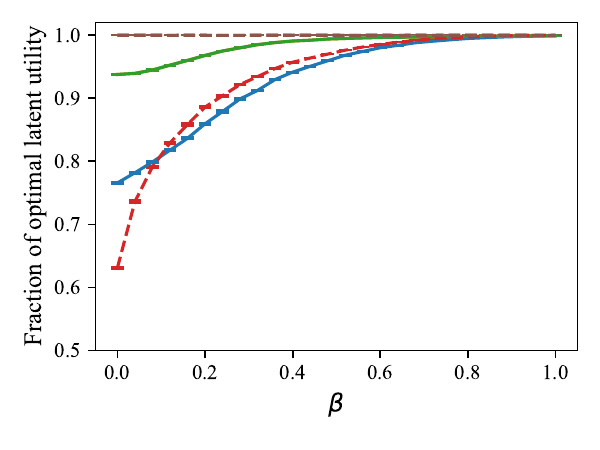}
            }
            \par 
            \subfigure[{$\cD_{\rm Pareto}$ and $\sP^{(1)}$}]{
                \includegraphics[width=0.31\linewidth, trim={0cm 0cm 0cm 0cm},clip]{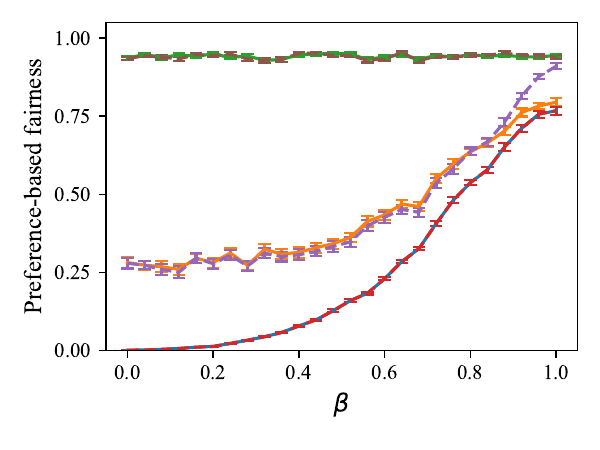}
            }
            \subfigure[{$\cD_{\rm Pareto}$ and $\sP^{(3)}$}]{
                \includegraphics[width=0.31\linewidth, trim={0cm 0cm 0cm 0cm},clip]{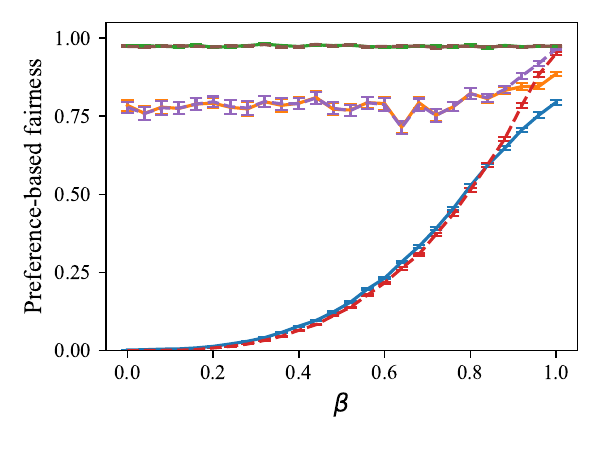}
            }
            \subfigure[{$\cD_{\rm Pareto}$ and $\sU$}]{
                \includegraphics[width=0.31\linewidth, trim={0cm 0cm 0cm 0cm},clip]{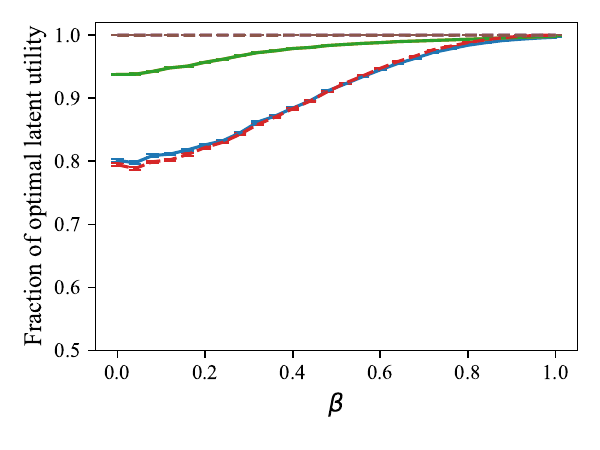}
            }
            \caption{
                Preference-based fairness is measured by $\sP^{(1)}$ and $\sP^{(3)}$, and utility ratio is measured under the noisy $\beta$-bias model.
                The control lines with no noise (dotted) are compared to the experimental lines with noise (solid). 
                (a) measures $\sP^{(1)}$ when utilities are generated from $\cD_{\rm Gauss}$. (b) measures $\sP^{(3)}$ under the same utility distribution. (c) shows the utility ratio when utilities are sampled from $\cD_{\rm Gauss}$. Note that the utility ratio lines for $\Agroup$ and $\Ainst$ overlap completely. (d), (e), and (f) are equivalent to (a), (b), and (c), respectively, except that utilities are drawn from $\cD_{\rm Pareto}$ instead of $\cD_{\rm Gauss}$.
                5
                Our main observation is that noise decreases the preference-based fairness at high $\beta$ and decreases the utility ratio of $\Ainst$ at low $\beta$. See \cref{sec:additional:empirical:betanoise} for details and discussion.
                The $x$-axis denotes $\beta$, the $y$-axis denotes $\sP^{(1)}$, $\sP^{(3)}$, or $\sU$, and the error bars denote the standard error of the mean over 50 iterations.
            }
            \label{fig:simulation:noisebias}
        \end{figure}

        \subsection{Simulation with noise in estimated utilities under the implicit variance model}\label{sec:simulation:noise}\label{sec:additional:empirical:impvar}
            In this section, we consider noise in the estimated utilities under the implicit variance model, which is an instance in which \cref{thm:institutionWiseConstraints} may not hold \cite{EmelianovGGL20}. The main purpose of this simulation is to test whether or not the results of \cref{thm:institutionWiseConstraints} hold under a different bias model. Since noise minimally changes the distribution as a whole, we expect institution-wise constraints to be relatively unaffected and would like to see how well group-wise constraints hold as a comparison.
            Given parameters $\delta_1\geq 0$ and $\delta_2 \geq 0$ controlling the extent of the noise, for each item $i$, the estimated utility of item $i$ is 
            \[
                \hu_{i, \sigma} \coloneqq \begin{cases}
                    u_i + \delta_1 \cdot \zeta_i & \text{if } i\in G_1,\\ 
                    u_i+ \delta_2 \cdot \zeta_i & \text{if } i\in G_2, 
                \end{cases} \quad \text{where}\quad
                \zeta_i \sim \cN(0,1).
            \]

            \paragraph{Setup.}
                We fix the parameters $n=1000,$ $p=5$, and $k_\ell=100$ for each $\ell\in [5]$.
                For each choice of $\cD\in \inbrace{\cD_{\rm Pareto}, \cD_{\rm Gauss}}$, we vary both $\delta_1$ and $\delta_2$ over $[0,2]$ and report the average preference-based fairness (both $\sP^{(1)}$ and $\sP^{(3)}$) over 50 iterations using both group-wise constraints ($\Agroup$) and institution-wise constraints ($\Ainst$).
                For this simulation, we let $\phi = 0.25$. As such, we are assuming that there is a consistent central ranking with variances in individual preferences.

            \paragraph{Results and observations.} We report results for $\sP^{(1)}$ in \cref{fig:noisebias1} and for $\sP^{(3)}$ in \cref{fig:noisebias5}.
            Our main observation is that as expected, noise had minimal effect on $\Ainst$. 
            We also see that when $\delta_1 = \delta_2$, $\Agroup$ does nearly as well as $\Ainst$, but the noise has a much greater effect on $\Agroup$.
            For all values of $\delta_1$ and $\delta_2$ tested, $\Ainst$ retained a high preference-based fairness. $\sP^{(1)} \geq 0.93$ and $\sP^{(3)} \geq 0.96$ for both $\cD_{\rm Gauss}$ and $\cD_{\rm Pareto}$. 
            $\Agroup$ also retained high preference-based fairness when $\delta_1 = \delta_2$ ($\sP^{(1)} \geq 0.90$, $\sP^{(3)} \geq 0.95$). However, when $\delta_1 \neq \delta_2$, there are cases where $\sP^{(1)} \leq 0.50$ and $\sP^{(3)} \leq 0.7$.
            The difference between the preference-based fairness resulting from $\Ainst$ and $\Agroup$ can approach 0.50 ($\sP^{(1)}$) and 0.27 ($\sP^{(3)}$) such as when $\delta_1 = 2$ and $\delta_2 = 0$.
            Institution-wise constraints maintain near-optimal preference-based fairness under noise in estimated utilities, while group-wise constraints only do so when the noise parameters are equal.

            \begin{figure}[ht!]
            \centering
            \subfigure[{$\Agroup$ and $
            \cD_{\rm Gauss}$}]{
                \includegraphics[width=0.28\linewidth, trim={0cm 0cm 0cm 0cm},clip]{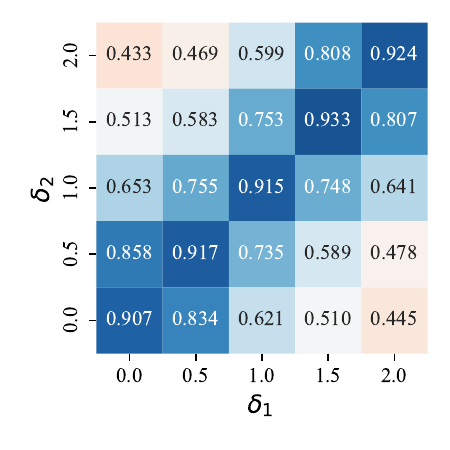}
            }
            \subfigure[{$\Agroup$ and $
            \cD_{\rm Pareto}$}]{
                \includegraphics[width=0.28\linewidth, trim={0cm 0cm 0cm 0cm},clip]{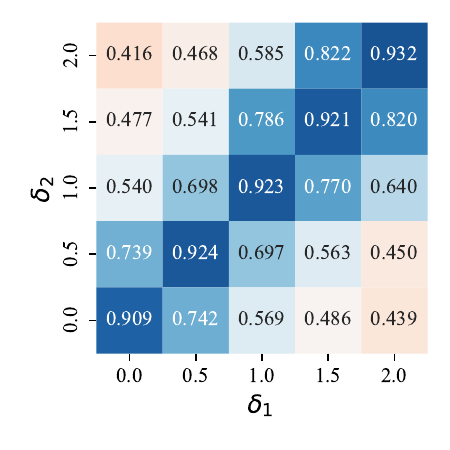}
            }
            \par 
            \subfigure[{$\Ainst$ and $
            \cD_{\rm Gauss}$}]{
                \includegraphics[width=0.28\linewidth, trim={0cm 0cm 0cm 0cm},clip]{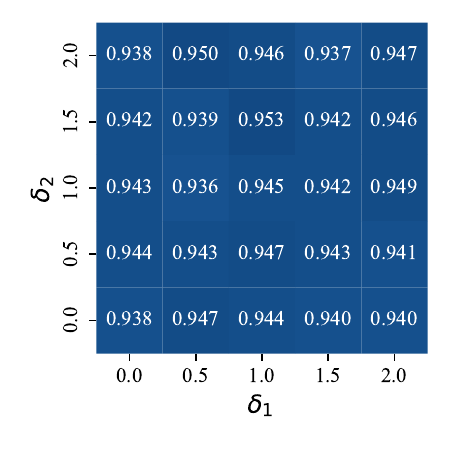}
            }
            \subfigure[{$\Ainst$ and $
            \cD_{\rm Pareto}$}]{
                \includegraphics[width=0.28\linewidth, trim={0cm 0cm 0cm 0cm},clip]{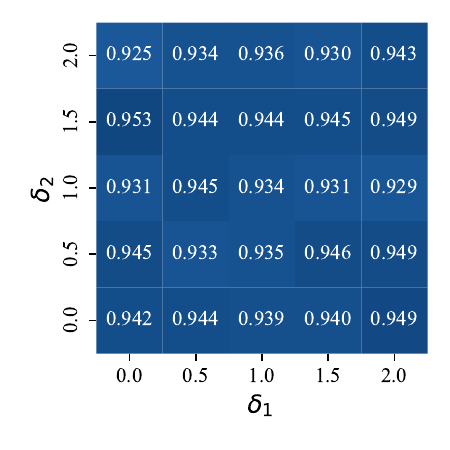}
            }
            \caption{
                Preference-based fairness as measured by $\sP^{(1)}$ under the implicit variance model, where the heatmap scales from 0 (red) to 1 (blue). (a) shows $\sP^{(1)}(\Agroup)$ when utilities are generated from $\cD_{\rm Gauss}$. (b) shows $\sP^{(1)}(\Agroup)$ when utilities are generated from $\cD_{\rm Pareto}$.  (c) shows $\sP^{(1)}(\Ainst)$ when utilities are generated from $\cD_{\rm Gauss}$. (d) shows $\sP^{(1)}(\Ainst)$ when utilities are generated from $\cD_{\rm Pareto}$. 
                Our main observation is that even with noise, institution-wise constraints maintain high preference-based fairness, while group-wise constraints may not. See \cref{sec:additional:empirical:impvar} for details and discussion.
                The $x$-axis denotes $\delta_1$, the $y$-axis denotes $\delta_2$, and the values are the average result over 50 iterations.
        }
            \label{fig:noisebias1}
            \vspace{-3mm}
        \end{figure}

        \begin{figure}[ht!]
            \centering
            \subfigure[{$\Agroup$ and $
            \cD_{\rm Gauss}$}]{
                \includegraphics[width=0.28\linewidth, trim={0cm 0cm 0cm 0cm},clip]{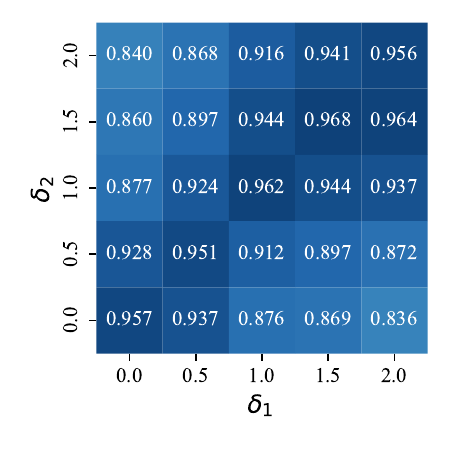}
            }
            \subfigure[{$\Agroup$ and $
            \cD_{\rm Pareto}$}]{
                \includegraphics[width=0.28\linewidth, trim={0cm 0cm 0cm 0cm},clip]{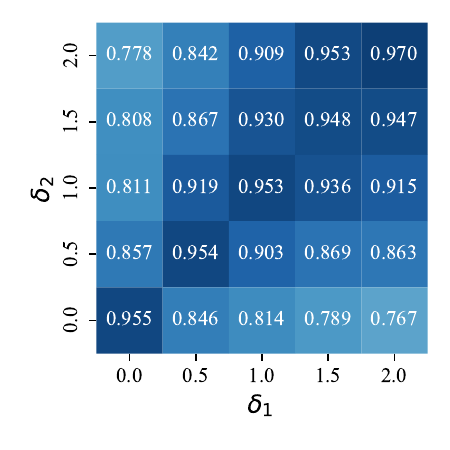}
            }
            \par 
            \subfigure[{$\Ainst$ and $
            \cD_{\rm Gauss}$}]{
                \includegraphics[width=0.28\linewidth, trim={0cm 0cm 0cm 0cm},clip]{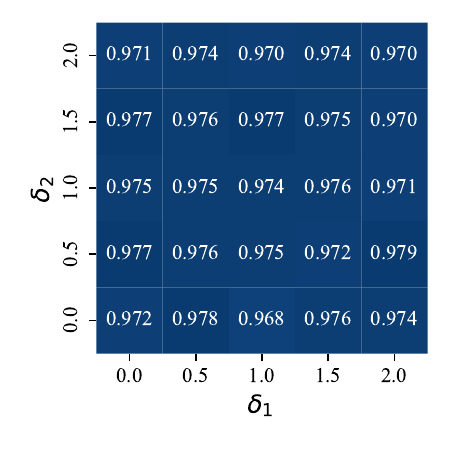}
            }
            \subfigure[{$\Ainst$ and $
            \cD_{\rm Pareto}$}]{
                \includegraphics[width=0.28\linewidth, trim={0cm 0cm 0cm 0cm},clip]{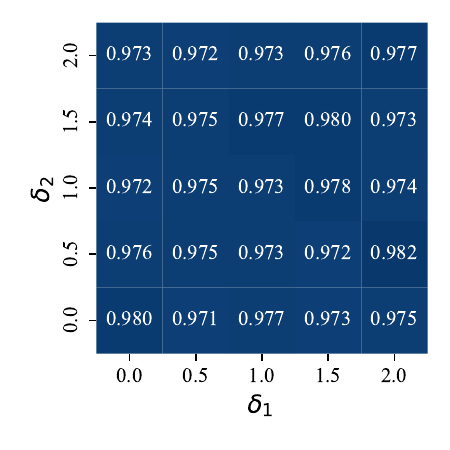}
            }
            \caption{
                Preference-based fairness as measured by $\sP^{(3)}$ under the implicit variance model, where the heatmap scales from 0 (red) to 1 (blue). (a) shows $\sP^{(3)}(\Agroup)$ when utilities are generated from $\cD_{\rm Gauss}$. (b) shows $\sP^{(3)}(\Agroup)$ when utilities are generated from $\cD_{\rm Pareto}$.  (c) shows $\sP^{(3)}(\Ainst)$ when utilities are generated from $\cD_{\rm Gauss}$. (d) shows $\sP^{(3)}(\Ainst)$ when utilities are generated from $\cD_{\rm Pareto}$. 
                Our main observation is that even with noise, institution-wise constraints maintain high preference-based fairness, while group-wise constraints may not. See \cref{sec:additional:empirical:impvar} for details and discussion.
                The $x$-axis denotes $\delta_1$, the $y$-axis denotes $\delta_2$, and the values are the average result over 50 iterations.
            }
            \label{fig:noisebias5}
        \end{figure}

\section{Empirical Validation of \cref{thm:specialcase}} \label{sec:additional:empirical:specialcase}
    In this section, we empirically show that \cref{thm:specialcase} extends to more robust settings than those required by the theorem. 
    In particular, \cref{thm:specialcase} assumes that the latent utilities of candidates are drawn from the uniform distribution on $[0,1]$ and the $\beta$-bias model is used to generate estimated utilities. An extension of the theorem, \cref{thm:specialcaselog}, proves a similar result for all log-concave densities, where upper bounds on $\sR$ and $\sU$ are given.
    Here, we conduct an empirical study using synthetic data to show that the implications of \cref{thm:specialcase} and \cref{thm:specialcaselog} continue to hold when the utilities are drawn from the truncated Gaussian distribution $\cD_{\rm Gauss}$ and the Pareto distribution $\cD_{\rm Pareto}$ (\cref{sec:paretogaussiantest4.1}); when the bias parameter $\beta$ is stochastic (\cref{sec:othermodels1-4.1}); and when the utilities are biased using the implicit variance model (\cref{sec:othermodels2-4.1}). 

    \subsection{Pareto and truncated Gaussian distributions of utilities}
    \label{sec:paretogaussiantest4.1}
    We validate the results of \cref{thm:specialcase} under the uniform distribution. Furthermore, we also validate \cref{thm:specialcaselog}, which extends the results of \cref{thm:specialcase} to all log-concave densities, by also testing $\cD_{\rm Gauss}$ and $\cD_{\rm Pareto}$. 

    \paragraph{Setup.} We generate latent utilities $u$, where $u \sim \cD \in \{\cD_{\rm Unif}, \cD_{\rm Gauss}, \cD_{\rm Pareto}\}$. In setting $n=10000$, $p=5$, and $k_\ell = 1000$ for each $\ell \in [5]$, preferences are generated from a Mallows distribution with $\phi = 0.25$ and $\Ast$ is run on the resulting data. We set $n=10000$ to minimize the error term. We vary $\beta$ in the range $[0, 1]$. We calculate $\sR$, $\sP^{(1)}$, and $\sU$ for over 50 iterations. 

    \paragraph{Results and observations}. Results are shown in \cref{fig:specialcase_nonoise}. Our main observation is that for the uniform distribution, $\sR$ and $\sU$ closely follow the theoretical results states in \cref{thm:specialcase}. This means the error term is likely small and negligible.
    In this simulation, $\sR$ for the uniform distribution closely models theoretical results $(\pm 0.01)$ and $\sU$ closely models theoretical predictions for all $\beta$. Since $\sP$ has an upper bound consistent with $\sR$, that statement also holds. We calculate $\sP^{(1)}$ to visualize the lower bound of $\sP$. 
    Under the premise of no noise for $\cD_{\rm Gauss}$ or $\cD_{\rm Pareto}$, we find that the results for $\cD_{\rm Gauss}$ closely follow the results for the uniform distribution, with differences $\leq 0.1$ for both representational fairness and utility ratio. The Pareto distribution did not closely match the predicted results and had a lower representational fairness for all $\beta$. The Pareto distribution resulted in a higher utility ratio for $\beta \leq 0.15$ and a lower one for higher $\beta$. 
    In general, the results for all three distributions match the results of \cref{thm:specialcase} for $\sR$ and $\sU$, with the occasional exception from $\cD_{\rm Pareto}$.

    \begin{figure}[ht!]
            \centering
            \includegraphics[width=0.6\linewidth]{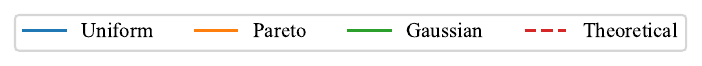}
            \par 
            \subfigure[{$\sR$}]{
                \includegraphics[width=0.31\linewidth, trim={0cm 0cm 0cm 0cm},clip]{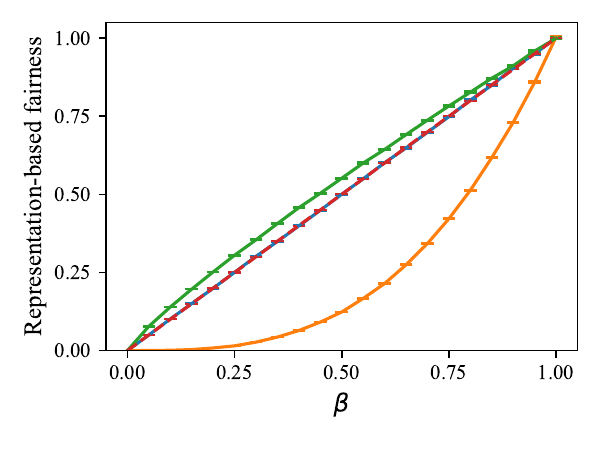}
            }
            \subfigure[{$\sP^{(1)}$}]{
                \includegraphics[width=0.31\linewidth, trim={0cm 0cm 0cm 0cm},clip]{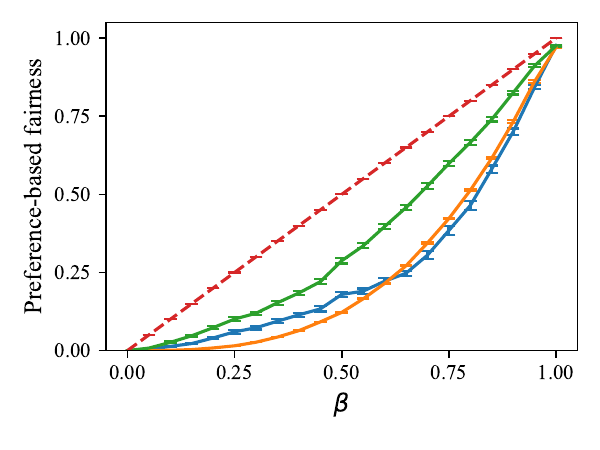}
            }
            \subfigure[{$\sU$}]{
                \includegraphics[width=0.31\linewidth, trim={0cm 0cm 0cm 0cm},clip]{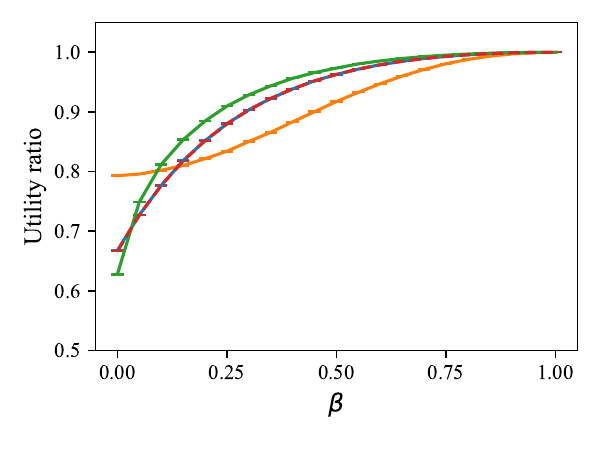}
            }
            \caption{
                Representation-based fairness is measured by $\sR$, preference-based fairness is measured by $\sP^{(1)}$, and utility ratio is measured under the $beta$-bias model with different utility distributions. (a) measures $\sR$ when utilities are generated from $\cD_{\rm Unif}$, $\cD_{\rm Gauss}$, and $\cD_{\rm Pareto}$, as well as the theoretical results for $\cD_{\rm Unif}$. (b) measures $\sP^{(1)}$, the lower bound for preference-based fairness. (c) measures the utility ratios of the three utility distributions. Our main observation is that for the uniform distribution, $\sR$ and $\sU$ closely follow the theoretical results stated in \cref{thm:specialcase}. See \cref{sec:additional:empirical:specialcase} for details and discussion.
                The $x$-axis denotes $\beta$, the $y$-axis denotes $\sR$, $\sP^{(1)}$, or $\sU$, and the error bars denote the standard error of the mean over 50 iterations.
            }
            \label{fig:specialcase_nonoise}
        \end{figure}

    \subsection{Noise in the bias parameter}
    \label{sec:othermodels1-4.1}
    We further test the robustness of \cref{thm:specialcase} by adding noise to the bias parameter, with the procedure detailed in \cref{sec:additional:empirical:betanoise}.

    \paragraph{Setup.} We generate latent utilities $u$, where $u \sim \cD \in \{\cD_{\rm Unif}, \cD_{\rm Gauss}, \cD_{\rm Pareto}\}$. In setting $n=10000$, $p=5$, and $k_\ell = 1000$ for each $\ell \in [5]$, preferences are generated from a Mallows distribution with $\phi = 0.25$ and $\Ast$ is run on the resulting data. We set $n=10000$ to minimize the error term. We vary $\beta \in [0, 1]$ with noise, where $\beta$ is sampled over a Gaussian distribution truncated to $[0, 1]$ with a standard deviation of 0.1 (detailed in \cref{sec:additional:empirical:betanoise}). We calculate $\sR$, $\sP^{(1)}$, and $\sU$ for over 50 iterations. 

    \paragraph{Results and observations}. Results are shown in \cref{fig:specialcase_beta}. Our main observation is that even under noise in the bias parameter, for the uniform distribution, $\sR$ and $\sU$ closely follow the theoretical results stated in \cref{thm:specialcase}. The representational fairness and utility ratio of all three distributions resembled the theoretical results, although $\cD_{\rm Pareto}$ deviated the most.
    We observe that the data is skewed for extreme values of $\beta$ ($\beta \leq 0.25, \beta \geq 0.75$). When considering the graph measuring $\sR$ with noise, we see that $\sR$ for the uniform distribution and $\cD_{\rm Gauss}$ is higher than without noise at low beta (At $\beta = 0$, within $0.10-0.15$). At $\beta = 1$, $\sR$ for all 3 distributions with noise is lower than without noise (within $0.1 - 0.25$). When looking at $\sP^{(1)}$, the most noticeable difference is when $\beta \to 1$. At $\beta = 1$, the difference lies between 0.15 ($\cD_{\rm Unif}$) and 0.25 ($\cD_{\rm Gauss}$). When looking at $\sU$, there is almost no difference for $\beta \geq 0.25$, but for low $\beta$, noise decreases $\sU$.
        Even under noise in the bias parameter, all three distributions of utilities follow the trends indicated in the theoretical results of \cref{thm:specialcase} for $\sR$ and $\sU$, with occasional exceptions in the case of  $\cD_{\rm Pareto}$.

    \begin{figure}[ht!]
            \centering
            \includegraphics[width=0.6\linewidth]{figures_FINAL/specialcase/gsbounds_legend.pdf}
            \par
             \subfigure[{$\sR$ with noise in $\beta$}]{
                \includegraphics[width=0.31\linewidth, trim={0cm 0cm 0cm 0cm},clip]{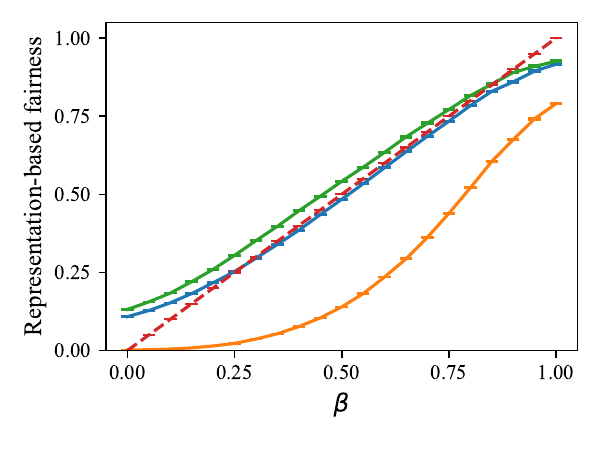}
            }
            \subfigure[{$\sP^{(1)}$ with noise in $\beta$}]{
                \includegraphics[width=0.31\linewidth, trim={0cm 0cm 0cm 0cm},clip]{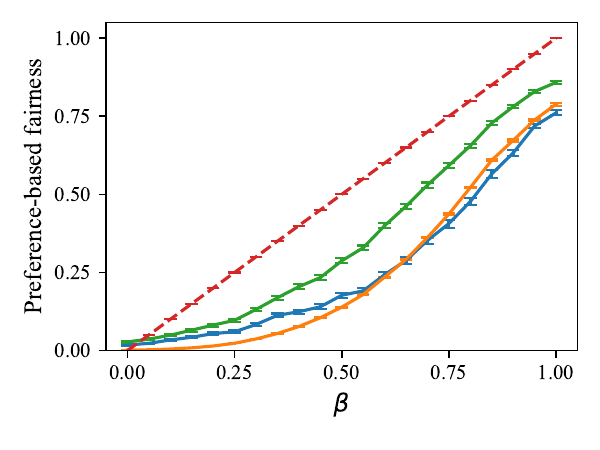}
            }
            \subfigure[{$\sU$ with noise in $\beta$}]{
                \includegraphics[width=0.31\linewidth, trim={0cm 0cm 0cm 0cm},clip]{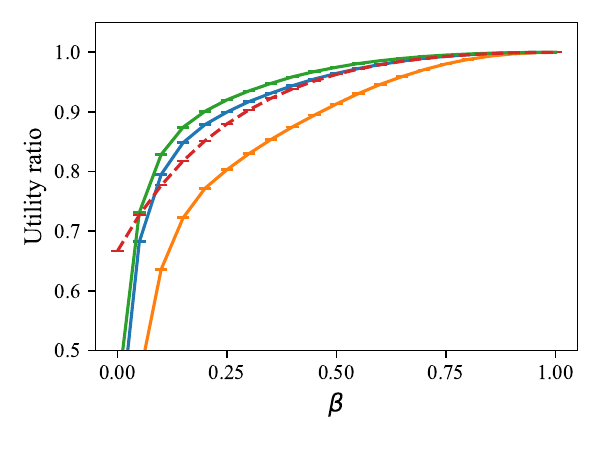}
            }
            \caption{
                Representation-based fairness is measured by $\sR$, preference-based fairness is measured by $\sP^{(1)}$, and utility ratio is measured under the noisy $\beta$-bias model with different utility distributions. (a) measures $\sR$ when utilities are generated from $\cD_{\rm Unif}$, $\cD_{\rm Gauss}$, and $\cD_{\rm Pareto}$, as well as the theoretical results for $\cD_{\rm Unif}$. (b) measures $\sP^{(1)}$, the lower bound for preference-based fairness. (c) measures the utility ratios of the three utility distributions. Our main observation is that noise skews the results for extreme $\beta$, but they generally still adhere to \cref{thm:specialcase}. See \cref{sec:additional:empirical:specialcase} for details and discussion.
                The $x$-axis denotes $\beta$, the $y$-axis denotes $\sR$, $\sP^{(1)}$, or $\sU$, and the error bars denote the standard error of the mean over 50 iterations.
            }
            \label{fig:specialcase_beta}
        \end{figure}

    \subsection{Implicit variance model}
    \label{sec:othermodels2-4.1}
    We also test \cref{thm:specialcase} by adding noise to estimated utilities under a simplified version of the implicit variance model detailed in \cref{sec:additional:empirical:impvar}. While considering the implicit variance model, we fix the variance parameter $\delta_1$ for group $G_1$ to 0 and only vary the parameter $\delta_2$ for group $G_2$ between $0$ and $2$.

    \paragraph{Setup.} We generate latent utilities $u$, where $u \sim \cD \in \{\cD_{\rm Unif}, \cD_{\rm Gauss}, \cD_{\rm Pareto}\}$. We add noise to half of the utilities under the implicit variance model. Here $n=10000$, $p=5$, and $k_\ell = 1000$ for each $\ell \in [5]$, preferences are generated from a Mallows distribution with $\phi = 0.25$ and $\Ast$ is run on the resulting data. We set $n=10000$ to minimize the error term. We fix $\beta = 1$ and vary $\delta = \delta_2$, the noise parameter, from $0$ to $2$. We calculate $\sR$, $\sP^{(1)}$, and $\sU$ for over 50 iterations. 

    \paragraph{Results and observations}. Results are shown in \cref{fig:specialcase_impvar}. Our main observation is that the representational fairness and utility ratio of all three distributions resembled the theoretical results, while preference-based fairness for all three distributions decreased as $\delta$ increased.
    In this simulation, $\sR$ for the uniform distribution closely models theoretical results $(\pm 0.05)$ and $\sU$ almost exactly models theoretical predictions for all $\beta$. Since $\sP$ has an upper bound consistent with $\sR$, that statement also holds. We calculate $\sP^{(1)}$ to visualize the lower bound of $\sP$. 
    Under the implicit variance model, where noise was added to the estimated utilities, we see that it does not affect $\sU$ at all ($\sU = 1$). 
    In terms of representational fairness, we observe that the fairness resulting from $\cD_{\rm Gauss}$ and $\cD_{\rm Pareto}$ decrease slightly due to noise. They drop to $0.90$ and $0.80$, respectively, and then plateau after $\delta \geq 0.5$. When looking at preference-based fairness in \cref{fig:specialcase_impvar}(b), we see that as noise increases, the lower bound on preference-based fairness ($\sP^{(1)}$) decreases from 1 to as low as $0.45$ for $\cD_{\rm Pareto}$ and $\cD_{\rm Gaussian}$  and $0.30$ for $\cD_{\rm Unif}$ at $\delta = 2$.
    We do see that although representational fairness and utility ratio stay consistent, noise can decrease the preference-based fairness guaranteed by these three distributions of preferences, specifically when there is noise in the estimated utilities under the implicit variance model.

        \begin{figure}[ht!]
            \centering
            \includegraphics[width=0.6\linewidth]{figures_FINAL/specialcase/gsbounds_legend.pdf}
            \par
             \subfigure[{$\sR$ with noise in utilities}]{
                \includegraphics[width=0.31\linewidth, trim={0cm 0cm 0cm 0cm},clip]{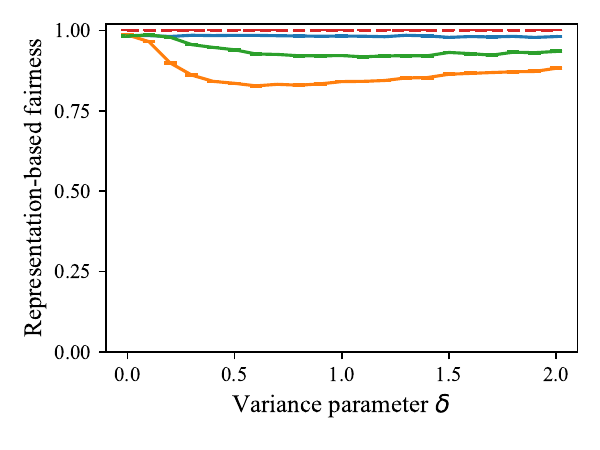}
            }
            \subfigure[{$\sP^{(1)}$ with noise in utilities}]{
                \includegraphics[width=0.31\linewidth, trim={0cm 0cm 0cm 0cm},clip]{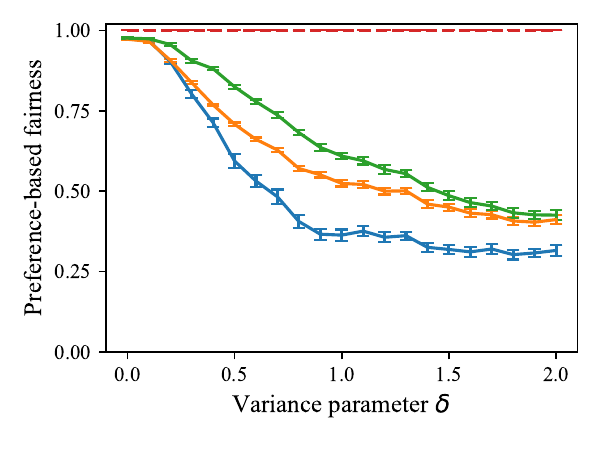}
            }
            \subfigure[{$\sU$ with noise in utilities}]{
                \includegraphics[width=0.31\linewidth, trim={0cm 0cm 0cm 0cm},clip]{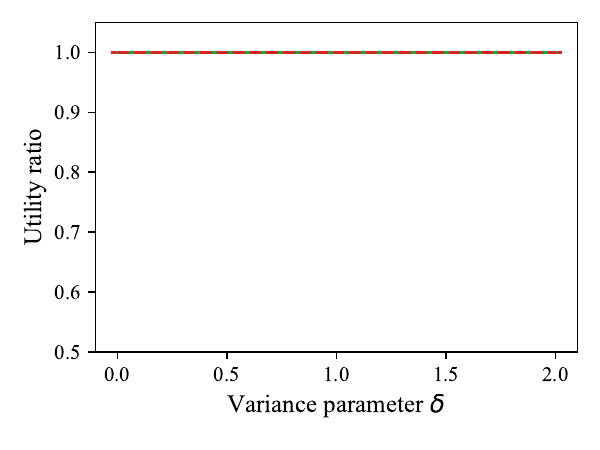}
            }
            \caption{
                Representation-based fairness is measured by $\sR$, preference-based fairness is measured by $\sP^{(1)}$, and utility ratio is measured under the implicit variance model with different utility distributions. (a) measures $\sR$ when utilities are generated from $\cD_{\rm Unif}$, $\cD_{\rm Gauss}$, and $\cD_{\rm Pareto}$, as well as the theoretical results for $\cD_{\rm Unif}$. (b) measures $\sP^{(1)}$, the lower bound for preference-based fairness. (c) measures the utility ratios of the three utility distributions. Our main observation is that $\sR$ and $\sU$ closely follow the theoretical results stated in \cref{thm:specialcase}, but preference-based fairness decreases as noise increases. See \cref{sec:additional:empirical:specialcase} for details and discussion.
                The $x$-axis denotes $\delta$, the $y$-axis denotes $\sR$, $\sP^{(1)}$, or $\sU$, and the error bars denote the standard error of the mean over 50 iterations.
            }
            \label{fig:specialcase_impvar}
        \end{figure}

\end{document}